\documentclass[11pt,letter,english]{article}
\usepackage{enumerate}
\usepackage[left=1in,right=1in,bottom=1in,top=1in]{geometry}
\usepackage{xspace}

\usepackage{hyperref}
\hypersetup{
    colorlinks=true,
    linkcolor=blue,
    filecolor=blue,      
    urlcolor=blue,
    citecolor=blue,
}
\usepackage[shortalphabetic]{amsrefs}
\usepackage{amsmath}
\usepackage{amsthm}
\usepackage{amssymb}
\usepackage{wasysym}  %\Leftcircle
\usepackage{xcolor}
\usepackage{tikz}
\usetikzlibrary{calc}
\usetikzlibrary{positioning}
\usepackage{tikz-cd}
\usepackage{titling,changepage}

\usepackage[all]{xy}
\newcommand{\GpI}{\ensuremath{\mathsf{GpI}}\xspace}
\newcommand{\GphI}{\ensuremath{\mathsf{GphI}}\xspace}
\newcommand{\NP}{\ensuremath{\mathsf{NP}}\xspace}
\newcommand{\coNP}{\ensuremath{\mathsf{coNP}}\xspace}

\newcommand{\ActComp}{\ensuremath{\mathsf{ActComp}}\xspace}
\newcommand{\CohIso}{\ensuremath{\mathsf{CohoIso}}\xspace}
\newcommand{\cc}[1]{\mathsf{#1}\xspace}

\usetikzlibrary{arrows}

\numberwithin{figure}{section}

\numberwithin{equation}{section}
\newtheorem{thm}[equation]{Theorem}
\newtheorem{mainthm}{Theorem}

\newtheorem*{thm-width-and-color}{Theorem~\ref{thm:width-and-color}}
\newtheorem*{thm-central}{Theorem~\ref{thm:central}}
\newtheorem*{thm*}{Theorem}
\newtheorem{lem}[equation]{Lemma}
\newtheorem{prop}[equation]{Proposition}
\newtheorem{cor}[equation]{Corollary}
\newtheorem{lemma}[equation]{Lemma}

\newtheorem{obs}[equation]{Observation}

\theoremstyle{definition}
\newtheorem{defn}[equation]{Definition}

\theoremstyle{remark}
\newtheorem{remark}[equation]{Remark}

\DeclareMathOperator{\End}{End}
\DeclareMathOperator{\Hom}{Hom}
\DeclareMathOperator{\Rad}{Rad}
\DeclareMathOperator{\Aut}{Aut}

\DeclareMathOperator{\isom}{Isom}
\DeclareMathOperator{\Isom}{Isom}
\DeclareMathOperator{\Adj}{Adj}

\DeclareMathOperator{\GL}{GL}

\DeclareMathOperator{\GammaL}{\Gamma L}

\DeclareMathOperator{\Gal}{Gal}

\DeclareMathOperator{\poly}{poly}

\DeclareMathOperator{\Norm}{Norm}

\newcommand{\smallmat}[1]{\left(\begin{array}{cc} #1 \end{array}\right)}

\DeclareMathOperator{\Soc}{Soc}

\DeclareMathOperator{\Sym}{Sym}

\newcommand{\bmto}{\rightarrowtail}
\newcommand{\F}{\mathbb{F}}

\newcommand{\Z}{\mathbb{Z}}
\newcommand{\N}{\mathbb{N}}
\newcommand{\la}{\langle}
\newcommand{\ra}{\rangle}
\newcommand{\pseudo}{\Psi\hspace*{-1mm}\isom}
\newcommand{\pAut}{\Psi\hspace*{-1mm}\Aut}

\renewcommand{\leq}{\leqslant}
\renewcommand{\geq}{\geqslant}

\newcommand{\cG}{\mathcal{G}}
\newcommand{\cE}{\mathcal{E}}
\newcommand{\cH}{\mathcal{H}}

\newcommand{\cV}{\mathcal{V}}

\newcommand{\bA}{\mathbf{A}}

\newcommand{\bB}{\mathbf{B}}

\newcommand{\bG}{\mathbf{G}}
\newcommand{\bH}{\mathbf{H}}

\newcommand{\width}{\text{width}}

\newcommand{\colorratio}{\text{color-ratio}}

\newcommand{\LinER}{\mathrm{LinER}}
\newcommand{\Mat}{\mathrm{M}}
\newcommand{\gbinom}[3]{{\genfrac{[}{]}{0pt}{}{#1}{#2}}_{#3}}

 \newcommand{\joshsay}[1]{} %{\textcolor{red}{Josh says: #1}}
 \newcommand{\jamessay}[1]{} %{\textcolor{blue}{James says: #1}}
 \newcommand{\petesay}[1]{} %{\textcolor{brown}{Pete says: #1}}
 
 \newcommand{\yinansay}[1]{} %{\textcolor{orange}{Yinan says: #1}}
 
\usepackage{algorithm}
\usepackage[noend]{algpseudocode}

\algblock[TryCatchFinally]{try}{endtry}
\algcblock[TryCatchFinally]{TryCatchFinally}{finally}{endtry}
\algcblockdefx[TryCatchFinally]{TryCatchFinally}{catch}{endtry}
	[1]{\textbf{catch} #1}
	{\textbf{end try}}

\makeatletter
\def\BState{\State\hskip-\ALG@thistlm}
\makeatother

\begin{document}
\title{\vspace{-0.75in}Incorporating Weisfeiler--Leman into algorithms for group isomorphism}
\author{Peter A. Brooksbank\thanks{
	Department of Mathematics,
	Bucknell University,
	Lewisburg, PA 17837,
	United States. 
{\tt pbrooksb@bucknell.edu}}
\and
Joshua A. Grochow\thanks{
	Departments of Computer Science and Mathematics,
	University of Colorado---Boulder,
	Boulder, CO 80309-0430,
	United States.
{\tt jgrochow@colorado.edu}}
\and 
Yinan Li\thanks{
CWI and QuSoft, Science Park 123, 1098XG Amsterdam, 
Netherlands. {\tt Yinan.Li@cwi.nl}.
}
\and
Youming Qiao\thanks{
	Center for Quantum Software and Information,
	University of Technology Sydney,
	Ultimo NSW 2007,
	Australia. {\tt Youming.Qiao@uts.edu.au}}
\and
James B. Wilson\thanks{
	Department of Mathematics,
	Colorado State University,
	Fort Collins, CO 80523,
	United States. {\tt James.Wilson@ColoState.Edu}}
	}

\date{\today}

\begin{titlingpage}
\usethanksrule
\changetext{0.0in}{0.5in}{}{-0.25in}{}
\maketitle

    \begin{abstract}
    % !TEX root = WL.tex

In this paper we combine many of the standard and more recent algebraic techniques for testing isomorphism of finite groups (\GpI) with combinatorial techniques that have typically been applied to Graph Isomorphism.  We show how to combine several state-of-the-art \GpI algorithms for specific group classes into an algorithm for general \GpI, namely: %code equivalence (Babai--Codenotti--Qiao, ICALP '12), 
composition series isomorphism (Rosenbaum--Wagner, \emph{Theoret. Comp. Sci.}, 2015; Luks, 2015), recursively-refineable filters (Wilson, \emph{J. Group Theory}, 2013), and low-genus \GpI (Brooksbank--Maglione--Wilson, \emph{J. Algebra}, 2017). Recursively-refineable filters---a generalization of subgroup series---form the skeleton of this framework, and we refine our filter by building a hypergraph encoding low-genus quotients, to which we then apply a hypergraph variant of the $k$-dimensional Weisfeiler--Leman technique. Our technique is flexible enough to readily incorporate additional hypergraph invariants or additional characteristic subgroups as they emerge.

After introducing this general technique, we prove three main results about its complexity:
\begin{itemize}
\item Let the \emph{width} of a filter be the dimension of the largest quotient of 
two adjacent subgroups of the filter; the \emph{color-ratio} of our hypergraph 
captures how much smaller a color class is compared to the layer of the filter it 
is coloring. When we use genus-$g$ quotients and hypergraph $k$-WL, we can solve 
isomorphism for solvable groups of order $n$ in time
\[
\left(\frac{n}{\colorratio}\right)^{\text{width}} \poly(n) + n^{O(gk)}
\]
In the ``base case'', where the solvable radical is itself low-genus and the 
semisimple part acts trivially, we can get a better guaranteed running time of 
$n^{O(\log \log n)}$, by combining cohomological techniques (Grochow--Qiao, CCC 
'14, \emph{SIAM J. Comput.}, 2017), code equivalence 
(Babai--Codenotti--Grochow--Qiao, SODA '11), and low-genus isomorphism ([BMW], 
\emph{ibid.}).

\item We introduce a new random model of finite groups. Unlike previous models, we prove that our model has good coverage, in that it produces a wide variety of groups, and in particular a number of distinct isomorphism types that is logarithmically equivalent to the number of all isomorphism types. In this random model, we show that our filter-and-1-WL refinement method results in constant \emph{average} width (the above result uses max width). 

\item For $p$-groups of class 2 and exponent $p$---widely believed to be the hardest cases of \GpI, and where we also expect the above techniques to get stuck---we improve on the average-case algorithm of Li--Qiao (FOCS '17). Our new algorithm is simpler and applies to a larger fraction of random $p$-groups of class 2 and exponent $p$. The previous algorithm was based on a linear-algebraic analogue of the individualize-and-refine technique; our new algorithm combines that technique with concepts from isomorphism of low-genus groups. We also implement this algorithm in {\sc MAGMA} and show that in experiments it improves over the default (brute force) algorithm for this problem.
\end{itemize}

    \end{abstract}

\pagenumbering{gobble}
\end{titlingpage}

\newpage
\pagenumbering{arabic}

%%%%%
\section{Introduction}
\label{sec:intro}
% !TEX root = WL.tex
The problem of deciding whether two finite groups are isomorphic (\GpI) has a century-old 
history that straddles several fields, including topology, computational algebra, and 
computer science.  It also has several unusual variations in complexity.  For example, the dense
input model---where groups are specified by their multiplication \emph{``Cayley''} tables,
has quasi-polynomial time complexity and it reduces \emph{to} the better known \emph{Graph Isomorphism} problem (\GphI); cf. \cite{Luks93}*{Section~10}.    Meanwhile,
\GpI for a sparse input model for groups, such as by permutations, matrices, or black-box
groups, reduces \emph{from} \GphI; cf. \citelist{\cite{HL}\cite{Lip}}.  At present sparse \GpI is in $\cc{\Sigma_2 P}$ 
and it is not known to lie in either \NP nor \coNP; see \cite{BS84}*{Propostion~4.8, Corollary~4.9}.  
In fact in the model of groups
input by generators and relations, Adian and Rabin famously showed \GpI is 
undecideable \cites{Adi57,Rab58}.

Following L. Babai's
breakthrough proof that \GphI is in quasi-polynomial time \cite{Bab16}, dense 
(Cayley table) \GpI is 
now an essential bottleneck to improving graph isomorphism.  So while the methods
we explore here can be applied in both the dense and the sparse models of \GpI, 
we concentrate our complexity claims on the dense case. In 
particular, when we say 
polynomial-time, we mean polynomial time in the group order unless specified 
otherwise. Our contribution here is to expose from within the structure of groups, graph theoretic properties
which relate to the difficulty of solving \GpI.  We expect this to facilitate the systematic use of combinatorial 
isomorphism techniques within \GpI  that interplay with existing algebraic strategies.

We introduce a colored (hyper-)graph based on an algebraic data structure
known as a {\em recursively-refineable filter} which identifies abelian groups and vector 
spaces layered together to form
the structure of a finite group.
Filters have been useful in several isomorphism tests~\citelist{\cite{Wilson:alpha}\cite{ Maglione:filters}\cite{Maglione:filters2}\cite{BOW}}. 
As the name suggests, filters can
be refined, and with each refinement the cost $T(n)$ of isomorphism testing decreases after
refinement to a function in $O(T(n)^{1/c})$, for some $c\geq 2$.  The more rounds
of refinement we can carry out the lower the cost of isomorphism.  Existing uses of filter refinement
find characteristic structure algebraically; our principal innovation is to add a combinatorial perspective 
to refinement. We color 
(co)dimension-$g$ subspaces of the layers of the filter using local isomorphism invariants. This parameter $g$ will be 
referred to as the {\em genus parameter}. 
The layers are in turn connected to each other according to
their position within the group, and this presents further opportunities for coloring.
With so much 
nuanced local information, the graph we associate to a group is well suited to 
individualization and refinement
techniques like the dimension-$k$ Weisfeiler--Leman procedure 
\cite{WL68,Babai79,ImmermanLander,CFI}. The 
critical work is to refine 
these graphs 
compatibly with the refinement of the filter (Theorem~\ref{thm:hypergraph}).  Thus, one maintains the relationship
between the group and graph isomorphism properties as we recursively refine.
While our methods do not apply to structures as general as semigroups
and quaisgroups, they can be adapted to other problems, such as ring isomorphism 
\cite{KS06}.

To explore the implications of this technique we introduce
a model for random finite groups.
In doing so we consider pitfalls identified in previously suggested 
models for random finite groups.  We are especially concerned with \emph{coverage}---%
the idea that we are able to easily sample from groups within natural classes 
such as non-solvable, solvable, nilpotent, and abelian---and that within each subclass
the number of isomorphism types is dense on a log scale (Theorem~\ref{thm:coverage}).
Log-scale is for now the best granularity we know for the enumeration of groups,
cf. \cite{BNV:enum}.  We then prove (Theorem~\ref{thm:random-refine}) that in our random model, genus-1 1-WL-refinement on 
average refines to a series of length $\Theta(\log n)$, which thereby achieves
the expected refinement length posed in \cite{Wilson:alpha}*{p. 876}. 
Following the refinement, the \emph{average} width of the filter is thus constant, though the cost of refinement increases. (If the maximum width were constant it would result in a polynomial-time average-case isomorphism test in our model.)

Finally within our random model there are several ``base cases'' where the 
recursive refinements become less likely, or where our analysis is inadequate.  
We demonstrate that in two of these
cases, isomorphism can be solved 
either in polynomial time  in the average-case sense (Theorem~\ref{thm:average}) or in
nearly polynomial time ($n^{O(\log \log n)}$) in the 
worst-case sense (Theorem~\ref{thm:central}).
The former also solves a related problem of average complexity of tensor 
equivalence.

Our strategy harnesses critical features of a great 
variety of existing approaches to isomorphism (code equivalence, filter
refinements, adjoint-tensor methods, bidirectional collision) and uses
Weisefeiler--Leman refinement as the top-level strategy to combine the 
various implications.  That diversity was not so much 
a plan but the result of hitting barriers and looking to the literature
for solutions.  The result, however, is a framework that is rather
flexible and is well suited to accommodate future ideas, both
algebraic and combinatorial, as featured here already.  That strength 
of course comes at a cost that the mechanics and analysis are rather involved.  
We expect that in time better analysis and simplified models will 
improve our understanding.

\subsection{The context of this work}
Much recent progress in \GpI has been had by considering special classes of 
groups; the recent papers \citelist{\cite{BMW}\cite{GQ}\cite{LGR}} survey and supplement these results. 
That has created  powerful but highly varied strategies with no obvious means of
synthesis.  Within our refinement model of computing \GpI we have the opportunity
to begin merging some of the many options that have been developed to date.  
To help explain our approach we consider examples of groups
of invertible matrices over finite fields of prime order, as graphically
communicated in Figure~\ref{fig:extend-random}.  In fact,
these examples will later evolve into the aforementioned random model for finite groups.
	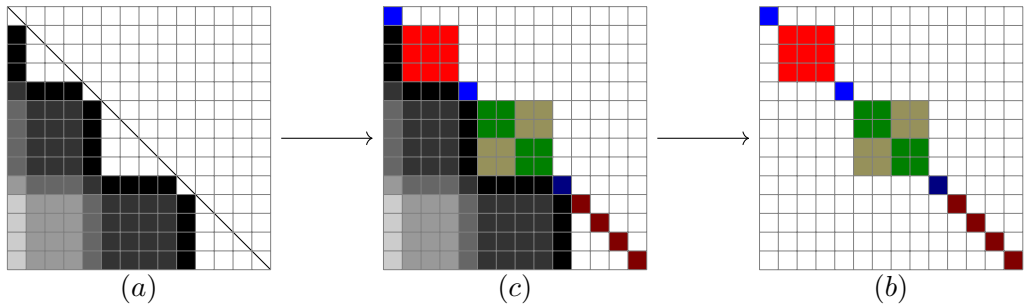
\begin{figure}[!h]
	\center
	\begin{tikzpicture}
	   	\node at (-5,-2) {$(a)$};  
		\node (a) at (-5,0) {\begin{tikzpicture}
    		% Fill the diagonal
		    \draw (0.00,3.50)  -- (3.50,0.00);
    		% gamma 1
	    	\fill[black] (0,3.25) rectangle (0.25,2.50);
		    \fill[black] (0.25,2.50) rectangle (1.00,2.25);
		    \fill[black] (1.00,2.25) rectangle (1.25,1.25);
		    \fill[black] (1.25,1.25) rectangle (2.25,1.00);
		    \fill[black] (2.25,1.00) rectangle (2.50,0);
		    % gamma 2
		    \fill[black!80!white] (0,2.50) rectangle (0.25,2.25);
		    \fill[black!80!white] (0.25,2.25) rectangle (1.00,1.25);
		    \fill[black!80!white] (1.00,1.25) rectangle (1.25,1.00);
		    \fill[black!80!white] (1.25,1.00) rectangle (2.25,0.00);
		    % gamma 3
		    \fill[black!60!white] (0,2.25) rectangle (0.25,1.25);
		    \fill[black!60!white] (0.25,1.25) rectangle (1.00,1.00);
		    \fill[black!60!white] (1.00,1.00) rectangle (1.25,0.00);
		    % gamma 4
		    \fill[black!40!white] (0,1.25) rectangle (0.25,1.00);
		    \fill[black!40!white] (0.25,1.00) rectangle (1.00,0.00);
		    % gamma 5
		    \fill[black!20!white] (0.00,1.00) rectangle (0.25,0.00);
		    \draw[step=0.25cm,gray,very thin] (0,0) grid (3.50,3.50);
    	\end{tikzpicture}};
    	\node at (0,-2) {$(c)$};  
		\node (c) at (0,0) {\begin{tikzpicture}
	    	% Fill the diagonal
	    	\draw (0.00,3.50)  -- (3.50,0.00);
		    % gamma 1
    		\fill[black] (0,3.25) rectangle (0.25,2.50);
		    \fill[black] (0.25,2.50) rectangle (1.00,2.25);
    		\fill[black] (1.00,2.25) rectangle (1.25,1.25);
		    \fill[black] (1.25,1.25) rectangle (2.25,1.00);
	    	\fill[black] (2.25,1.00) rectangle (2.50,0);
		    % gamma 2
	    	\fill[black!80!white] (0,2.50) rectangle (0.25,2.25);
    		\fill[black!80!white] (0.25,2.25) rectangle (1.00,1.25);
		    \fill[black!80!white] (1.00,1.25) rectangle (1.25,1.00);
    		\fill[black!80!white] (1.25,1.00) rectangle (2.25,0.00);
		    % gamma 3
    		\fill[black!60!white] (0,2.25) rectangle (0.25,1.25);
	    	\fill[black!60!white] (0.25,1.25) rectangle (1.00,1.00);
    		\fill[black!60!white] (1.00,1.00) rectangle (1.25,0.00);
		    % gamma 4
    		\fill[black!40!white] (0,1.25) rectangle (0.25,1.00);
		    \fill[black!40!white] (0.25,1.00) rectangle (1.00,0.00);
    		% gamma 5
		    \fill[black!20!white] (0.00,1.00) rectangle (0.25,0.00);
	    	\fill[blue] (0.00,3.25) rectangle (0.25,3.50);
    		\fill[red] (0.25,2.50) rectangle (1.00,3.25);
		    \fill[blue] (1.00,2.25) rectangle (1.25,2.50);
    		\fill[black!50!green] (1.25,1.75) rectangle (1.75,2.25);
		    \fill[black!50!yellow] (1.75,1.75) rectangle (2.25,2.25);
	    	\fill[black!50!green] (1.75,1.25) rectangle (2.25,1.75);
    		\fill[black!50!yellow] (1.25,1.25) rectangle (1.75,1.75);
		    \fill[black!50!blue] (2.25,1.00) rectangle( 2.50, 1.25);
	    	\fill[black!50!red] (2.50,0.75) rectangle( 2.75, 1.0);
    		\fill[black!50!red] (2.75,0.50) rectangle( 3.00, 0.75);
    		\fill[black!50!red] (3.00,0.25) rectangle( 3.25, 0.50);
	    	\fill[black!50!red] (3.25,0.00) rectangle( 3.50, 0.25);
    		\draw[step=0.25cm,gray,very thin] (0,0) grid (3.50,3.50);
    	\end{tikzpicture}};  
    	\node at (5,-2) {$(b)$};  
		\node (b) at (5,0) {\begin{tikzpicture}
	    	% Fill the diagonal
		    \fill[blue] (0.00,3.25) rectangle (0.25,3.50);
		    \fill[red] (0.25,2.50) rectangle (1.00,3.25);
		    \fill[blue] (1.00,2.25) rectangle (1.25,2.50);
		    \fill[black!50!green] (1.25,1.75) rectangle (1.75,2.25);
		    \fill[black!50!yellow] (1.75,1.75) rectangle (2.25,2.25);
		    \fill[black!50!green] (1.75,1.25) rectangle (2.25,1.75);
		    \fill[black!50!yellow] (1.25,1.25) rectangle (1.75,1.75);
		    \fill[black!50!blue] (2.25,1.00) rectangle( 2.50, 1.25);
		    \fill[black!50!red] (2.50,0.75) rectangle( 2.75, 1.0);
		    \fill[black!50!red] (2.75,0.50) rectangle( 3.00, 0.75);
		    \fill[black!50!red] (3.00,0.25) rectangle( 3.25, 0.50);
		    \fill[black!50!red] (3.25,0.00) rectangle( 3.50, 0.25);
    		\draw[step=0.25cm,gray,very thin] (0,0) grid (3.50,3.50);
	    \end{tikzpicture}};
	    \draw[->] (a) -- (c);
	    \draw[->] (c) -- (b);
	\end{tikzpicture}
    \caption{Diagrams of matrix groups can capture many of the well-studied examples
    of finite groups: (a) depicts a large variety of nilpotent groups; (b) depicts
    products of quasi- and almost-simple groups together with possible permutations of
    isomorphic blocks; and (c) depicts wide range of general finite groups decomposed 
    into
    smaller classes of groups.}\label{fig:extend-random}
    \end{figure}

\paragraph{First thread: connection with linear and multilinear algebras.} 
Algorithms and data structures for linear and multilinear structures are on the whole
far more evolved than counterparts for groups. This explains why 
progress for groups can be made by mapping problems into the realm of linear
and multilinear algebra.   Such a correspondence has been known for close to a century,
originating in work of Brahana \cite{Bra35} and Baer \cite{Bae}.  Consider groups 
$U$ of the following form. 
\begin{align*}
	U & \leq H(d_1,\ldots,d_{\ell};\F):=\left\{
	\begin{bmatrix} 
		I_{d_1} & a_{12} & a_{22} & \cdots \\ & I_{d_2} & a_{23} \\ 
			& & \ddots & \ddots \\ 
			& & & I_{d_{\ell}}
	\end{bmatrix} 
			~\middle|~
			a_{ij}\in M_{d_{i}\times d_j}(\F)\right\}.
\end{align*}
Figure~\ref{fig:extend-random}(a) illustrates a possible $U$.
In the creation of our random model we shall sample groups $U$ by selecting
random matrices in $H(d_1,\ldots,d_{\ell}; \F)$. A surprising necessity is that we sample
only sparse matrices. Although this might seem counter to the goals of seeding a group with lots
of entropy, we will demonstrate that groups with too much random seeding become virtually
identical (Theorem~\ref{thm:dense-random}).

As a general remark it will be helpful throughout this work to regard all groups
$U=\langle U,\cdot, ^{-1},1\rangle$ as having been enriched by the addition of a second
binary operation $[a,b]=a^{-1}b^{-1}ab$ known as \emph{commutation}.  In this way
groups behave much more like rings than they do like semigroups or quasigroups.  In
particular, $[,]$  very nearly distributes over the usual binary
operation $\cdot$ in $U$, in that $[ab,c]=b^{-1}[a,c]b[b,c]$.  That explains the link to multilinear algebra.
In the case of $U$:
\begin{align*}
	\left[
		\begin{bmatrix} 
		I_{d_1} & a_{12} & \cdots  \\ 
		& \ddots & \ddots \\ 
			& & I_{d_{\ell}} 
		\end{bmatrix},
		\begin{bmatrix} 
		I_{d_1} & b_{12} &\cdots \\
			& \ddots & \ddots \\ 
			&  & I_{d_{\ell}}
		\end{bmatrix}
	\right]	
	& = 
		\begin{bmatrix} 
		I_{d_1} & a_{12}+b_{12} & a_{22}+b_{22}+a_{12}b_{23}-b_{12}a_{23} & \cdots \\ 
			& \ddots & \ddots \\ 
			& & I_{d_{\ell}}
		\end{bmatrix}.
\end{align*}
Stripping away the addition leaves us to compare bilinear (and later multilinear)
products such as $(a_{ij},b_{jk})\mapsto a_{ij}b_{jk}$, under base changes.  We
treat these as functions $\F^a\times \F^b\bmto \F^c$, where $\bmto$ indicates the function
is multilinear.  Equivalently, we must study the orbits of groups $\GL(a,\F)\times \GL(b,\F)\times \GL(c,\F)$
acting on elements of the tensor space $\F^a\otimes\F^b\otimes \F^c$.
Such reductions of group isomorphism to multilinear equivalence, and more general \emph{tensor
equivalence} problems have been the key to the recent progress 
on some of the largest and most difficult instances of \GpI \citelist{
\cite{BMW}\cite{BW:isom}\cite{LQ}\cite{IQ}
\cite{LW12}\cite{BOW}\cite{Wilson:profile}}.  The strategies buried within those methods are nevertheless
quite distinct.  For example, several focus on $*$-algebras and properties of rings
and modules acting on tensors.  Others focus on tensors as high-dimensional arrays, 
and perform individualization and refinement techniques on slices
of this data structure.  Our model of refinement allows for both strategies.

\paragraph{Second thread: relationship to code equivalence.}
Now consider the types of groups we could place on the block diagonal of
the matrix group examples in Figure~\ref{fig:extend-random}.  These could include groups like $\GL(d_i,\F)$.
We could also use subgroups such as $\GL(1,\F_{p^{d_i}})=\F_{p^{d-1}}^{\times}$,
as well as natural families of geometrically interesting groups such as orthogonal,
unitary, and symplectic groups.  We may even embed the same group several times 
into multiple blocks on the diagonal, e.g. 
$\left\langle \left[\begin{smallmatrix} A & 0 \\ 0 & A\end{smallmatrix}\right] : A\in \GL(e,\F)\right\rangle$.
Those blocks could further be permuted producing groups of block monomial matrices
such as $\left\langle \left[\begin{smallmatrix} 0 & A \\ A & 0\end{smallmatrix}\right]: A\in \GL(e,\F)\right\rangle$. 
We can capture the spirit of such a group graphically in Figure~\ref{fig:extend-random}(b).
Indeed, our random group model builds random semi-simple and quasi-semisimple groups
in just this way.

 Isomorphism in the context of groups of this kind has been approached mostly through the
use of code equivalence.  For example, for semisimple groups---those with no 
non-trivial abelian normal subgroups---there is an algorithm that runs in 
time polynomial in the group order \cites{BCGQ,BCQ}, as well as an algorithm that is efficient in practice \cite{CH03}. The key algorithmic idea 
is 
dynamic programming, and its use follows the one by Luks in the context of 
hypergraph isomorphism and code equivalence \cite{Luks99}.  Later \cite{GQ} 
considers the further
implications when the groups centrally extend abelian groups similar to the general
family we have described in this thread.

\paragraph{Third thread: composition series and filters.}
In recent years there has been some progress on improving general isomorphism 
using subgroup chains.
Rosenbaum and Wagner~\cite{RosenbaumWagner} demonstrated that one can fix a
composition series $\mathcal{C}(G)$ for group $G$ and then, given a composition series $\mathcal{C}(H)$ for
another group $H$, efficiently 
decide if there is an isomorphism $G\to H$ sending $\mathcal{C}(G)$ to 
$\mathcal{C}(H)$. Luks gave an improvement of  that test \cite{LuksCompSeries}.
In this way, the putative cost of $n^{\log n+O(1)}$ steps to decide isomorphism by 
brute-force is reduced to the number of possible composition series, which is at
most $n^{(\frac{1}{2}+o(1))\log n}$. 

Filters can use characteristic subgroups to recursively find more characteristic subgroups,
ultimately producing a large enough collection of fixed subgroups that
an isomorphism test along
the lines of Rosenbaum--Wagner becomes efficient.  For several families of
groups such refinements have been discovered \citelist{\cite{Wilson:alpha}\cite{Maglione:filters-U}}.
Our approach here extends the filtration process by taking the methods known and
combining them into a colored hypergraph where individualization-refinement techniques
can be applied.  The goal is to make it even more likely to
reach a situation in which the Rosenbaum--Wagner and Luks algorithms can be applied 
efficiently.

%%%
\subsection{An outline of the Weisfeiler--Leman procedure for groups}
\label{subsec:WL-outline}

Our approach to \GpI uses recursively-refineable filters to build and refine a colored 
hypergraph within, and between, abelian layers of a given group. 
A filter $\phi$ on a group $G$ assigns to a $c$-tuple $s=(s_1,\ldots,s_c)$ of natural numbers (including $0$)
a normal subgroup $\phi_s$ of $G$ subject to natural compatibility requirements. 
Let $\Norm(G)$ denote the set of normal subgroups of $G$, and for $A,B\subseteq G$ let
$[A,B]=\langle [a,b]\mid a\in A, b\in b\rangle$.

\begin{defn}[Filter~~\cite{Wilson:alpha}]
A \emph{filter} on a group $G$ is a map 
 $\phi\colon \N^d \to \Norm(G)$, where 
\begin{align}
\label{eq:filters}
(\forall s,t \in \N^d) &&
s \leq_{{\rm lex}} t~~ \Longrightarrow~~ \phi_s \geq \phi_t~~ \mbox{and}~~[\phi_s, \phi_t] \leq \phi_{s+t}.
\end{align}
\end{defn}

Note that the first condition implies that the subgroups $\phi_s$ form a descending chain
of subgroups, though in general it is not a proper chain.   
Computationally we only store the lexicographically least label $s$ for each distinct subgroup
$\phi_s$ in the image of $\phi$.  Thus, a filter's image is bounded by the length of the
longest subgroup chain.  For a group of order $n$ this is at most $\log n$. 

We begin with a filter $\phi:\mathbb{N}^c\to \Norm(G)$ known 
from the structure of general finite groups, and then refine by increasing the value of $c$.
That refinements exist is proved in \cite{Wilson:alpha} and that they 
can be computed efficiently is shown in \cite{Maglione:filters}.
Our initial value for $c$ will be the number of distinct primes $p_1,\ldots,p_c$ dividing $n=|G|$.
For each prime $p_i$, we let $O_{p_i}(G)$ denote the intersection of all Sylow $p_i$-subgroups
of $G$,  the maximum normal subgroup having order a power of $p_i$. 
Let $e_i=(\ldots,0,\underset{i}{1},0,\ldots)\in\mathbb{N}^c$, sorted lexicographically (so that
$e_i< e_{i+1}$), and define
$\phi\colon \mathbb{N}^c\to \Norm(G)$ as follows:
\begin{align*}
	\phi_s & = \left\{\begin{array}{ll}
		G & s=0,\\
		\prod_{j=i}^{c} O_{p_j}(G) & s= e_i,\\~
		[\phi_{s_ie_i},G]\phi_{s_i e_i}^{p_i} & s=(s_i+1)e_i,\\
		\prod_{i=1}^c \phi_{s_i e_i}  & s= \sum_{i=1}^c s_i e_i.
	\end{array}\right.
\end{align*}
Here the product $\prod_{i} \phi_{s_i e_i}$ means the normal subgroup generated by
the terms $\phi_{s_i e_i}$. 
For example the group $S_4$ of permutations on $4$ letters would have
\begin{align*}
	\phi_{(0,0)} & = S_4 \geq \phi_{(1,0)}=O_2(S_4)O_3(S_4)=\langle (12)(34),(13)(24)\rangle \geq 
		\phi_{(2,0)}=\phi_{(0,1)}=O_3(S_4)=1.
\end{align*}

The \emph{boundary filter} $\partial\phi\colon \N^d \to \Norm(G)$ is defined by 
$\partial \phi_s = \langle \phi_{s+t} : t \in \N^d \backslash \{0\} \rangle$ (if $d=1$, then $\partial \phi_s = 
\phi_{s+1}$), and the quotients $L_s := \phi_s / \partial \phi_s$ are the {\em layers} of $\phi$.
Note that for each $s\neq 0$, $L_s$ is abelian, and in fact a $\mathbb{Z}[\phi_0/\partial\phi_0]$-module.
In the selected filter above these are in fact $\F_{p_i}$-vector spaces for some $i$.
The set
$L(\phi)=\bigoplus_{s\neq 0} L_s$,
with homogeneous bilinear products
\begin{align*}
	[,]_{st}& :L_s\times L_t\bmto L_{s+t}:
	  (x\partial\phi_s,y\partial\phi_t)\mapsto [x,y]\partial\phi_{s+t},
\end{align*}
is a graded Lie algebra whose graded components are invariant under ${\rm Aut}(G)$~\cite{Wilson:alpha}*{Theorem~3.1}.

A bilinear map (bimap) $L_s \times L_t \bmto L_{s+t}$ is said to have \emph{genus} $g$ if it is defined over a field $\F$ such that $\dim_{\F} L_{s+t} \leq g$, or (see \cite{BMW} for details) if it is built from such maps by certain elementary products (such as direct products, but even ``central'' products are allowed). We will primarily be concerned with the case where $\F = \Z_p$ and we consider bimaps 
whose codomain has dimension at most $g$, but our results extend without difficulty to the more general notion of genus.

In our setting the layers of $\phi$ are elementary abelian, and our approach is to build a hypergraph
whose vertices are the union of the points (1-spaces) in the projective geometries of the layers.
For $s\in\N^d$, let ${\rm PG}_k(L_s)$ denote
the set of $(k+1)$-dimensional subspaces of $L_s$. 
Define a family of hypergraphs $\cH^{(g)}(\phi)$,
where $1 < g \in \Z$ is a parameter, with vertices and 
hyperedges defined as follows: 
\begin{enumerate}
\item[] The vertex set of $\cH^{(g)}(\phi)$ is $\mathcal{V}=\bigcup_{s\in \N^d}{\rm PG}_0(L_s)$.
\item[] The hyperedge set of $\cH^{(g)}(\phi)$ is $\mathcal{E}=\bigcup_{s\in\N^d} {\rm PG}_g(L_s)\cup {\rm PG}_{\dim L_s - g}(L_s)\;\cup\;\bigcup_{s\neq t}\mathcal{K}_{st}$, where $\mathcal{K}_{st}$ is a hypergraph with edges and 3-edges on ${\rm PG}_0(L_s)\cup {\rm PG}_0(L_t) \cup {\rm PG}_0(L_{s+t})$.
\end{enumerate}

 Having defined the hypergraph $\cH^{(g)}(\phi)$, 
we shall apply the 
$k$-dimensional Weisfeiler-Leman (WL) procedure to it in an appropriate 
way.  Briefly, it is a hypergraph version of the WL 
procedure \cite{WL68} on graphs \cite{Babai79,ImmermanLander}. When $k=1$, such a 
WL procedure on hypergraphs was recently studied by B\"oker \cite{Boker}.

To this end we 
obtain an algorithm that, given a finite 
group $G$ and integers $g,k\geq 1$, computes a suitable
characteristic filter $\phi\colon \N^d\to \Norm(G)$, where 
$N=O_{\infty}(G)=\prod_{p} O_p(G)$ is the Fitting subgroup,  
and an associated hypergraph $\cH^{(g,k)}(\phi)$. 
Further, it colors the hyperedges
$\mathcal{E}$ of  $\cH^{(g,k)}(\phi)$ in a certain desirable way. If
$\chi\colon \mathcal{E}\to \N$ is a coloring of hyperedges, denote the corresponding colored 
hypergraph by $\cH_{\chi}^{(g,k)}(\phi)$.

%%%
\begin{mainthm}
\label{thm:hypergraph}
There is a deterministic algorithm that, given a finite group $G$ and integers $g, k \geq 1$, constructs 
the Fitting subgroup $N=O_{\infty}(G)$, a characteristic filter $\phi\colon 
\N^d\to \Norm(G)$ whose 
non-zero layers are elementary 
abelian  ${\rm Aut}(G)$-modules, the hypergraph $\cH=\cH^{(g,k)}(\phi)$, and a coloring 
$\chi\colon \mathcal{V}(\cH) \cup \mathcal{E}(\cH)\to \N$
satisfying: 
   \begin{enumerate}[(i)]
   \item $\cH^{(g,k)}_{\chi}(\phi)$ is hereditary in the following sense:
   for each $s\in\N^d-\{0\}$, the vertex-and-edge-colored hypergraph obtained by restricting 
   $\cH_{\chi}^{(g,k)}(\phi)$ to $G/\phi_s$ is a refinement
   of the colored hypergraph for $G/\phi_s$ based on the filter 
   $\phi$ truncated at $\phi_s$.
   
   \item $\cH^{(g,k)}_\chi(\phi)$ is also hereditary in $k$ in the following sense: 
   the underlying hypergraphs of $\cH^{(g,k)}(\phi)$ and $\cH^{(g,k+1)}(\phi)$ are identical, 
   and the coloring of the latter refines the coloring of the former.
   
   \item If $G\cong G'$, there is a colored hypergraph isomorphism $f\colon \cH^{(g,k)}_{\chi}(\phi)\to \cH^{(g,k)}_{\chi'}(\phi')$
   such that
   \begin{align*}
   \forall e\in\mathcal{E}(\mathcal{H}), & \qquad  \chi(e)=\chi'(f(e)), \\
   \forall v \in \mathcal{V}(\mathcal{H}), & \qquad \chi(v) = \chi'(f(v)).
   \end{align*}
   \end{enumerate}
The time complexity is $|G|^{O(gk)}$.
\end{mainthm} 

The algorithm to construct the colored hypergraph $\cH^{(g,k)}_{\chi}(\phi)$ is an iterative procedure that we
describe in detail in Section~\ref{sec:hypergraph}. Within a fixed iteration, we apply a Weisfeiler--Leman 
type individualization procedure to obtain a stable coloring (a hypergraph analogue of $k$-dimensional WL). We then use that stable coloring to search 
for characteristic structure in $G$ not already captured by the filter $\phi$. If we succeed, we use this structure
to refine $\phi$ and iterate. 

Given the result of our WL-algorithm and applying Luks's extension \cite{LuksCompSeries} of the
Rosenbaum--Wagner composition series comparison \cite{RosenbaumWagner}, whenever we refine we improve our 
isomorphism test, resulting in:

\begin{mainthm} \label{thm:width-and-color}
Let $\phi = \phi_{g,k}$ and $\mathcal{H}^{(g,k)}=\mathcal{H}_{\chi}^{(g,k)}(\phi)$  denote the filter and colored hypergraph from Theorem~\ref{thm:hypergraph}. Let $\width(\phi)$ denote the maximum dimension of any layer of $\phi$, and let $\colorratio(\mathcal{H})$ be the product over all layers $s$ of $|L_s| / |C_s|$, where $C_s$ is the smallest color class in layer $L_s$. Then given a nilpotent group $N$ of order $n$, isomorphism can be tested in time
\[
 \left(\frac{n}{\colorratio(\mathcal{H}^{(g,k)})} \right)^{\width(\phi_{g,k})} \poly(n) + n^{O(gk)}.
\]
\end{mainthm}
We extend this with an individualize-and-refine technique in Section~\ref{sec:iso-IR}, though for that we do not have as cleanly stated an upper bound.

\begin{remark} \label{rmk:solvable-long}
The initial filter described above can be extended to solvable groups, and in particular the solvable radical $\Rad(G)$ of any group, by doing something similar to the above %Section~\ref{subsec:WL-outline} 
within each layer of the Fitting series. This would let us extend all our results from using the Fitting subgroup $O_\infty(G)$ to using the solvable radical $\Rad(G)$ instead, and would extend Theorem~\ref{thm:width-and-color} from nilpotent to solvable groups.
\end{remark}

%%%%%
\subsection{An outline of the random model}
\label{subsec:random}
Unlike sampling a random graph, where edges can freely be added or omitted,
sampling groups of a fixed order requires some delicacy.  For example, there are 15
isomorphism types of groups of order 16 but only 1 each of orders 15 and 17.  Sampling 
random groups has hitherto been approached in one of the following two ways.
\begin{description}
\item[Quotient Sampling.] Fix a free group $F[X]$ of all strings on an alphabet
$X\cup X^{-1}$, and consider quotients by normal subgroups $N=\langle S\rangle$ sampled
by choosing $S\subset F$ by some aleatory process.

\item[Subgroup Sampling.] Fix an automorphism group of a structure, such as the group
$\Sym(\Omega)$ of permutations of a set $\Omega$, or the group $\GL(V)$ of invertible 
linear transformations of a vector space $V$.  Consider subgroups $H=\langle S\rangle$ where $S$ is sampled
by some aleatory process.
\end{description} 
Evidently, both methods yield groups,  
but neither offers sufficient variability when restricted to finite groups.
For instance Gromov studied quotient sampling
as a function of the word lengths of elements in $S$, finding most quotients are $1$, $\mathbb{Z}/2$,
or infinite \cite{Gromov}.  Also, subgroup sampling in $G=\Sym(\Omega)$ (respectively $\GL(V)$) has been shown by 
Dixon, Kantor--Lubotsky~\cite{KL}, and others to essentially sample $A_n$, $S_n$ 
(respectively, subgroups ${\rm SL}(V)\leq H\leq {\rm GL}(V)$).

To escape these conditions we adopt a method of sampling that appears antithetical
to random models: we strongly bias our random selections.  We settle on a model related
to subgroup sampling in $\GL(d,p)$ since this affords us easy-to-use group operations.  
(Note, Novikov--Boone demonstrated that the word problem in the free group is
undecidable and thus working with quotients $F[X]/N$ is not in general feasible 
\cites{Nov55,Boo59}.)   

First, we sample random upper $(d\times d)$-unitriangular matrices $u_1,\ldots, u_{\ell}\in U(d,p)$ 
but we insist that they are $\epsilon$-sparse, for some constant $\epsilon$. Then
\begin{align*}
	U= \langle u_1,\ldots, u_{\ell}\rangle.
\end{align*}
samples a subgroup whose order is a power of the prime $p$ characteristic
of our fixed field $\F$.  As we shall demonstrate in Theorem~\ref{thm:dense-random}, 
\emph{without limiting our randomness to 
sparse matrices} the groups $U$ will almost always contain the following subgroup.
\begin{align*}
	\gamma_2(U(d,p)) & = \left\{\begin{bmatrix}
		1 & 0 & * & \\
		 & \ddots & \ddots & \ddots & \\
		 &  & 1 & 0 & *\\
		 & & & 1 & 0 \\
		 & & & &  1
	\end{bmatrix}\right\}.
\end{align*}
In essence, this is a $p$-group analogue of the observations we made about 
sampling in 
$S_n$ and $\GL(d,p)$.  However, sampling with sparsity gives substantial variation, 
as illustrated simply by comparing orders in Figure~\ref{fig:scatter-plot}.
An interesting recent study by
R. Gilman describes a similar situation for permutations analyzed by
Kolomogorov complexity \cite{Gilman}.

\begin{figure}
\center

% GNUPLOT: LaTeX picture
\setlength{\unitlength}{0.240900pt}
\ifx\plotpoint\undefined\newsavebox{\plotpoint}\fi
\sbox{\plotpoint}{\rule[-0.200pt]{0.400pt}{0.400pt}}%
\begin{picture}(1500,900)(0,0)
\sbox{\plotpoint}{\rule[-0.200pt]{0.400pt}{0.400pt}}%
\put(90.0,82.0){\rule[-0.200pt]{4.818pt}{0.400pt}}
\put(70,82){\makebox(0,0)[r]{$0$}}
\put(1419.0,82.0){\rule[-0.200pt]{4.818pt}{0.400pt}}
\put(90.0,211.0){\rule[-0.200pt]{4.818pt}{0.400pt}}
\put(70,211){\makebox(0,0)[r]{$10$}}
\put(1419.0,211.0){\rule[-0.200pt]{4.818pt}{0.400pt}}
\put(90.0,341.0){\rule[-0.200pt]{4.818pt}{0.400pt}}
\put(70,341){\makebox(0,0)[r]{$20$}}
\put(1419.0,341.0){\rule[-0.200pt]{4.818pt}{0.400pt}}
\put(90.0,470.0){\rule[-0.200pt]{4.818pt}{0.400pt}}
\put(70,470){\makebox(0,0)[r]{$30$}}
\put(1419.0,470.0){\rule[-0.200pt]{4.818pt}{0.400pt}}
\put(90.0,599.0){\rule[-0.200pt]{4.818pt}{0.400pt}}
\put(70,599){\makebox(0,0)[r]{$40$}}
\put(1419.0,599.0){\rule[-0.200pt]{4.818pt}{0.400pt}}
\put(90.0,729.0){\rule[-0.200pt]{4.818pt}{0.400pt}}
\put(70,729){\makebox(0,0)[r]{$50$}}
\put(1419.0,729.0){\rule[-0.200pt]{4.818pt}{0.400pt}}
\put(90.0,858.0){\rule[-0.200pt]{4.818pt}{0.400pt}}
\put(70,858){\makebox(0,0)[r]{$60$}}
\put(1419.0,858.0){\rule[-0.200pt]{4.818pt}{0.400pt}}
\put(90.0,82.0){\rule[-0.200pt]{0.400pt}{4.818pt}}
\put(90,41){\makebox(0,0){$0$}}
\put(90.0,838.0){\rule[-0.200pt]{0.400pt}{4.818pt}}
\put(259.0,82.0){\rule[-0.200pt]{0.400pt}{4.818pt}}
\put(259,41){\makebox(0,0){$5$}}
\put(259.0,838.0){\rule[-0.200pt]{0.400pt}{4.818pt}}
\put(427.0,82.0){\rule[-0.200pt]{0.400pt}{4.818pt}}
\put(427,41){\makebox(0,0){$10$}}
\put(427.0,838.0){\rule[-0.200pt]{0.400pt}{4.818pt}}
\put(596.0,82.0){\rule[-0.200pt]{0.400pt}{4.818pt}}
\put(596,41){\makebox(0,0){$15$}}
\put(596.0,838.0){\rule[-0.200pt]{0.400pt}{4.818pt}}
\put(765.0,82.0){\rule[-0.200pt]{0.400pt}{4.818pt}}
\put(765,41){\makebox(0,0){$20$}}
\put(765.0,838.0){\rule[-0.200pt]{0.400pt}{4.818pt}}
\put(933.0,82.0){\rule[-0.200pt]{0.400pt}{4.818pt}}
\put(933,41){\makebox(0,0){$25$}}
\put(933.0,838.0){\rule[-0.200pt]{0.400pt}{4.818pt}}
\put(1102.0,82.0){\rule[-0.200pt]{0.400pt}{4.818pt}}
\put(1102,41){\makebox(0,0){$30$}}
\put(1102.0,838.0){\rule[-0.200pt]{0.400pt}{4.818pt}}
\put(1270.0,82.0){\rule[-0.200pt]{0.400pt}{4.818pt}}
\put(1270,41){\makebox(0,0){$35$}}
\put(1270.0,838.0){\rule[-0.200pt]{0.400pt}{4.818pt}}
\put(1439.0,82.0){\rule[-0.200pt]{0.400pt}{4.818pt}}
\put(1439,41){\makebox(0,0){$40$}}
\put(1439.0,838.0){\rule[-0.200pt]{0.400pt}{4.818pt}}
\put(90.0,82.0){\rule[-0.200pt]{0.400pt}{186.938pt}}
\put(90.0,82.0){\rule[-0.200pt]{324.974pt}{0.400pt}}
\put(1439.0,82.0){\rule[-0.200pt]{0.400pt}{186.938pt}}
\put(90.0,858.0){\rule[-0.200pt]{324.974pt}{0.400pt}}
\put(1279,817){\makebox(0,0)[r]{Matrix density $10\%$}}
\put(1299.0,817.0){\rule[-0.200pt]{24.090pt}{0.400pt}}
\put(90,82){\usebox{\plotpoint}}
\multiput(562.00,82.58)(1.329,0.493){23}{\rule{1.146pt}{0.119pt}}
\multiput(562.00,81.17)(31.621,13.000){2}{\rule{0.573pt}{0.400pt}}
\multiput(596.00,93.92)(1.329,-0.493){23}{\rule{1.146pt}{0.119pt}}
\multiput(596.00,94.17)(31.621,-13.000){2}{\rule{0.573pt}{0.400pt}}
\multiput(630.00,82.58)(0.635,0.497){49}{\rule{0.608pt}{0.120pt}}
\multiput(630.00,81.17)(31.739,26.000){2}{\rule{0.304pt}{0.400pt}}
\multiput(663.58,108.00)(0.498,0.573){65}{\rule{0.120pt}{0.559pt}}
\multiput(662.17,108.00)(34.000,37.840){2}{\rule{0.400pt}{0.279pt}}
\multiput(697.00,145.92)(0.654,-0.497){49}{\rule{0.623pt}{0.120pt}}
\multiput(697.00,146.17)(32.707,-26.000){2}{\rule{0.312pt}{0.400pt}}
\multiput(731.00,121.58)(1.329,0.493){23}{\rule{1.146pt}{0.119pt}}
\multiput(731.00,120.17)(31.621,13.000){2}{\rule{0.573pt}{0.400pt}}
\multiput(765.00,134.58)(1.290,0.493){23}{\rule{1.115pt}{0.119pt}}
\multiput(765.00,133.17)(30.685,13.000){2}{\rule{0.558pt}{0.400pt}}
\multiput(798.00,145.92)(0.654,-0.497){49}{\rule{0.623pt}{0.120pt}}
\multiput(798.00,146.17)(32.707,-26.000){2}{\rule{0.312pt}{0.400pt}}
\put(90.0,82.0){\rule[-0.200pt]{113.705pt}{0.400pt}}
\multiput(866.58,121.00)(0.497,0.974){63}{\rule{0.120pt}{0.876pt}}
\multiput(865.17,121.00)(33.000,62.182){2}{\rule{0.400pt}{0.438pt}}
\multiput(899.58,182.73)(0.498,-0.558){65}{\rule{0.120pt}{0.547pt}}
\multiput(898.17,183.86)(34.000,-36.865){2}{\rule{0.400pt}{0.274pt}}
\multiput(933.00,145.92)(1.329,-0.493){23}{\rule{1.146pt}{0.119pt}}
\multiput(933.00,146.17)(31.621,-13.000){2}{\rule{0.573pt}{0.400pt}}
\multiput(967.58,134.00)(0.498,2.490){65}{\rule{0.120pt}{2.076pt}}
\multiput(966.17,134.00)(34.000,163.690){2}{\rule{0.400pt}{1.038pt}}
\multiput(1001.58,295.70)(0.497,-1.785){63}{\rule{0.120pt}{1.518pt}}
\multiput(1000.17,298.85)(33.000,-113.849){2}{\rule{0.400pt}{0.759pt}}
\multiput(1034.58,181.46)(0.498,-0.945){65}{\rule{0.120pt}{0.853pt}}
\multiput(1033.17,183.23)(34.000,-62.230){2}{\rule{0.400pt}{0.426pt}}
\multiput(1068.58,121.00)(0.498,0.573){65}{\rule{0.120pt}{0.559pt}}
\multiput(1067.17,121.00)(34.000,37.840){2}{\rule{0.400pt}{0.279pt}}
\put(832.0,121.0){\rule[-0.200pt]{8.191pt}{0.400pt}}
\multiput(1135.00,160.58)(1.329,0.493){23}{\rule{1.146pt}{0.119pt}}
\multiput(1135.00,159.17)(31.621,13.000){2}{\rule{0.573pt}{0.400pt}}
\multiput(1169.58,169.41)(0.498,-0.960){65}{\rule{0.120pt}{0.865pt}}
\multiput(1168.17,171.21)(34.000,-63.205){2}{\rule{0.400pt}{0.432pt}}
\put(1102.0,160.0){\rule[-0.200pt]{7.950pt}{0.400pt}}
\multiput(1270.00,106.92)(1.329,-0.493){23}{\rule{1.146pt}{0.119pt}}
\multiput(1270.00,107.17)(31.621,-13.000){2}{\rule{0.573pt}{0.400pt}}
\multiput(1304.00,93.92)(1.329,-0.493){23}{\rule{1.146pt}{0.119pt}}
\multiput(1304.00,94.17)(31.621,-13.000){2}{\rule{0.573pt}{0.400pt}}
\multiput(1338.00,82.58)(1.329,0.493){23}{\rule{1.146pt}{0.119pt}}
\multiput(1338.00,81.17)(31.621,13.000){2}{\rule{0.573pt}{0.400pt}}
\multiput(1372.00,93.92)(1.290,-0.493){23}{\rule{1.115pt}{0.119pt}}
\multiput(1372.00,94.17)(30.685,-13.000){2}{\rule{0.558pt}{0.400pt}}
\put(1203.0,108.0){\rule[-0.200pt]{16.140pt}{0.400pt}}
\put(90,82){\makebox(0,0){$+$}}
\put(124,82){\makebox(0,0){$+$}}
\put(157,82){\makebox(0,0){$+$}}
\put(191,82){\makebox(0,0){$+$}}
\put(225,82){\makebox(0,0){$+$}}
\put(259,82){\makebox(0,0){$+$}}
\put(292,82){\makebox(0,0){$+$}}
\put(326,82){\makebox(0,0){$+$}}
\put(360,82){\makebox(0,0){$+$}}
\put(394,82){\makebox(0,0){$+$}}
\put(427,82){\makebox(0,0){$+$}}
\put(461,82){\makebox(0,0){$+$}}
\put(495,82){\makebox(0,0){$+$}}
\put(528,82){\makebox(0,0){$+$}}
\put(562,82){\makebox(0,0){$+$}}
\put(596,95){\makebox(0,0){$+$}}
\put(630,82){\makebox(0,0){$+$}}
\put(663,108){\makebox(0,0){$+$}}
\put(697,147){\makebox(0,0){$+$}}
\put(731,121){\makebox(0,0){$+$}}
\put(765,134){\makebox(0,0){$+$}}
\put(798,147){\makebox(0,0){$+$}}
\put(832,121){\makebox(0,0){$+$}}
\put(866,121){\makebox(0,0){$+$}}
\put(899,185){\makebox(0,0){$+$}}
\put(933,147){\makebox(0,0){$+$}}
\put(967,134){\makebox(0,0){$+$}}
\put(1001,302){\makebox(0,0){$+$}}
\put(1034,185){\makebox(0,0){$+$}}
\put(1068,121){\makebox(0,0){$+$}}
\put(1102,160){\makebox(0,0){$+$}}
\put(1135,160){\makebox(0,0){$+$}}
\put(1169,173){\makebox(0,0){$+$}}
\put(1203,108){\makebox(0,0){$+$}}
\put(1237,108){\makebox(0,0){$+$}}
\put(1270,108){\makebox(0,0){$+$}}
\put(1304,95){\makebox(0,0){$+$}}
\put(1338,82){\makebox(0,0){$+$}}
\put(1372,95){\makebox(0,0){$+$}}
\put(1405,82){\makebox(0,0){$+$}}
\put(1439,82){\makebox(0,0){$+$}}
\put(1349,817){\makebox(0,0){$+$}}
\put(1405.0,82.0){\rule[-0.200pt]{8.191pt}{0.400pt}}
\put(1279,776){\makebox(0,0)[r]{$50\%$}}
\multiput(1299,776)(20.756,0.000){5}{\usebox{\plotpoint}}
\put(1399,776){\usebox{\plotpoint}}
\put(90,82){\usebox{\plotpoint}}
\multiput(90,82)(20.756,0.000){2}{\usebox{\plotpoint}}
\multiput(124,82)(20.756,0.000){2}{\usebox{\plotpoint}}
\put(173.02,82.00){\usebox{\plotpoint}}
\multiput(191,82)(20.756,0.000){2}{\usebox{\plotpoint}}
\multiput(225,82)(20.756,0.000){2}{\usebox{\plotpoint}}
\put(276.80,82.00){\usebox{\plotpoint}}
\multiput(292,82)(20.756,0.000){2}{\usebox{\plotpoint}}
\multiput(326,82)(20.756,0.000){2}{\usebox{\plotpoint}}
\put(380.58,82.00){\usebox{\plotpoint}}
\multiput(394,82)(20.756,0.000){2}{\usebox{\plotpoint}}
\put(442.84,82.00){\usebox{\plotpoint}}
\multiput(461,82)(20.756,0.000){2}{\usebox{\plotpoint}}
\multiput(495,82)(20.756,0.000){2}{\usebox{\plotpoint}}
\put(546.62,82.00){\usebox{\plotpoint}}
\multiput(562,82)(20.756,0.000){2}{\usebox{\plotpoint}}
\multiput(596,82)(20.756,0.000){2}{\usebox{\plotpoint}}
\put(650.40,82.00){\usebox{\plotpoint}}
\multiput(663,82)(20.756,0.000){2}{\usebox{\plotpoint}}
\put(712.66,82.00){\usebox{\plotpoint}}
\multiput(731,82)(20.756,0.000){2}{\usebox{\plotpoint}}
\multiput(765,82)(20.756,0.000){2}{\usebox{\plotpoint}}
\put(816.44,82.00){\usebox{\plotpoint}}
\multiput(832,82)(20.756,0.000){2}{\usebox{\plotpoint}}
\put(878.71,82.00){\usebox{\plotpoint}}
\multiput(899,82)(20.756,0.000){2}{\usebox{\plotpoint}}
\multiput(933,82)(20.756,0.000){2}{\usebox{\plotpoint}}
\put(982.49,82.00){\usebox{\plotpoint}}
\multiput(1001,82)(20.756,0.000){2}{\usebox{\plotpoint}}
\multiput(1034,82)(20.756,0.000){2}{\usebox{\plotpoint}}
\put(1086.26,82.00){\usebox{\plotpoint}}
\multiput(1102,82)(20.756,0.000){2}{\usebox{\plotpoint}}
\put(1148.53,82.00){\usebox{\plotpoint}}
\multiput(1169,82)(19.387,7.413){2}{\usebox{\plotpoint}}
\multiput(1203,95)(11.358,17.372){3}{\usebox{\plotpoint}}
\multiput(1237,147)(16.303,-12.845){2}{\usebox{\plotpoint}}
\multiput(1270,121)(11.358,17.372){3}{\usebox{\plotpoint}}
\multiput(1304,173)(5.838,19.918){6}{\usebox{\plotpoint}}
\multiput(1338,289)(8.294,19.026){4}{\usebox{\plotpoint}}
\multiput(1372,367)(19.311,-7.607){2}{\usebox{\plotpoint}}
\multiput(1405,354)(11.513,17.270){3}{\usebox{\plotpoint}}
\put(1439,405){\usebox{\plotpoint}}
\put(90,82){\makebox(0,0){$\times$}}
\put(124,82){\makebox(0,0){$\times$}}
\put(157,82){\makebox(0,0){$\times$}}
\put(191,82){\makebox(0,0){$\times$}}
\put(225,82){\makebox(0,0){$\times$}}
\put(259,82){\makebox(0,0){$\times$}}
\put(292,82){\makebox(0,0){$\times$}}
\put(326,82){\makebox(0,0){$\times$}}
\put(360,82){\makebox(0,0){$\times$}}
\put(394,82){\makebox(0,0){$\times$}}
\put(427,82){\makebox(0,0){$\times$}}
\put(461,82){\makebox(0,0){$\times$}}
\put(495,82){\makebox(0,0){$\times$}}
\put(528,82){\makebox(0,0){$\times$}}
\put(562,82){\makebox(0,0){$\times$}}
\put(596,82){\makebox(0,0){$\times$}}
\put(630,82){\makebox(0,0){$\times$}}
\put(663,82){\makebox(0,0){$\times$}}
\put(697,82){\makebox(0,0){$\times$}}
\put(731,82){\makebox(0,0){$\times$}}
\put(765,82){\makebox(0,0){$\times$}}
\put(798,82){\makebox(0,0){$\times$}}
\put(832,82){\makebox(0,0){$\times$}}
\put(866,82){\makebox(0,0){$\times$}}
\put(899,82){\makebox(0,0){$\times$}}
\put(933,82){\makebox(0,0){$\times$}}
\put(967,82){\makebox(0,0){$\times$}}
\put(1001,82){\makebox(0,0){$\times$}}
\put(1034,82){\makebox(0,0){$\times$}}
\put(1068,82){\makebox(0,0){$\times$}}
\put(1102,82){\makebox(0,0){$\times$}}
\put(1135,82){\makebox(0,0){$\times$}}
\put(1169,82){\makebox(0,0){$\times$}}
\put(1203,95){\makebox(0,0){$\times$}}
\put(1237,147){\makebox(0,0){$\times$}}
\put(1270,121){\makebox(0,0){$\times$}}
\put(1304,173){\makebox(0,0){$\times$}}
\put(1338,289){\makebox(0,0){$\times$}}
\put(1372,367){\makebox(0,0){$\times$}}
\put(1405,354){\makebox(0,0){$\times$}}
\put(1439,405){\makebox(0,0){$\times$}}
\put(1349,776){\makebox(0,0){$\times$}}
\sbox{\plotpoint}{\rule[-0.400pt]{0.800pt}{0.800pt}}%
\sbox{\plotpoint}{\rule[-0.200pt]{0.400pt}{0.400pt}}%
\put(1279,735){\makebox(0,0)[r]{$100\%$}}
\sbox{\plotpoint}{\rule[-0.400pt]{0.800pt}{0.800pt}}%
\put(1299.0,735.0){\rule[-0.400pt]{24.090pt}{0.800pt}}
\put(90,82){\usebox{\plotpoint}}
\multiput(1305.41,82.00)(0.503,0.768){61}{\rule{0.121pt}{1.424pt}}
\multiput(1302.34,82.00)(34.000,49.045){2}{\rule{0.800pt}{0.712pt}}
\multiput(1339.41,134.00)(0.503,1.731){61}{\rule{0.121pt}{2.929pt}}
\multiput(1336.34,134.00)(34.000,109.920){2}{\rule{0.800pt}{1.465pt}}
\multiput(1373.41,250.00)(0.503,2.188){59}{\rule{0.121pt}{3.642pt}}
\multiput(1370.34,250.00)(33.000,134.440){2}{\rule{0.800pt}{1.821pt}}
\multiput(1406.41,392.00)(0.503,6.799){61}{\rule{0.121pt}{10.859pt}}
\multiput(1403.34,392.00)(34.000,430.462){2}{\rule{0.800pt}{5.429pt}}
\put(90,82){\makebox(0,0){$\ast$}}
\put(124,82){\makebox(0,0){$\ast$}}
\put(157,82){\makebox(0,0){$\ast$}}
\put(191,82){\makebox(0,0){$\ast$}}
\put(225,82){\makebox(0,0){$\ast$}}
\put(259,82){\makebox(0,0){$\ast$}}
\put(292,82){\makebox(0,0){$\ast$}}
\put(326,82){\makebox(0,0){$\ast$}}
\put(360,82){\makebox(0,0){$\ast$}}
\put(394,82){\makebox(0,0){$\ast$}}
\put(427,82){\makebox(0,0){$\ast$}}
\put(461,82){\makebox(0,0){$\ast$}}
\put(495,82){\makebox(0,0){$\ast$}}
\put(528,82){\makebox(0,0){$\ast$}}
\put(562,82){\makebox(0,0){$\ast$}}
\put(596,82){\makebox(0,0){$\ast$}}
\put(630,82){\makebox(0,0){$\ast$}}
\put(663,82){\makebox(0,0){$\ast$}}
\put(697,82){\makebox(0,0){$\ast$}}
\put(731,82){\makebox(0,0){$\ast$}}
\put(765,82){\makebox(0,0){$\ast$}}
\put(798,82){\makebox(0,0){$\ast$}}
\put(832,82){\makebox(0,0){$\ast$}}
\put(866,82){\makebox(0,0){$\ast$}}
\put(899,82){\makebox(0,0){$\ast$}}
\put(933,82){\makebox(0,0){$\ast$}}
\put(967,82){\makebox(0,0){$\ast$}}
\put(1001,82){\makebox(0,0){$\ast$}}
\put(1034,82){\makebox(0,0){$\ast$}}
\put(1068,82){\makebox(0,0){$\ast$}}
\put(1102,82){\makebox(0,0){$\ast$}}
\put(1135,82){\makebox(0,0){$\ast$}}
\put(1169,82){\makebox(0,0){$\ast$}}
\put(1203,82){\makebox(0,0){$\ast$}}
\put(1237,82){\makebox(0,0){$\ast$}}
\put(1270,82){\makebox(0,0){$\ast$}}
\put(1304,82){\makebox(0,0){$\ast$}}
\put(1338,134){\makebox(0,0){$\ast$}}
\put(1372,250){\makebox(0,0){$\ast$}}
\put(1405,392){\makebox(0,0){$\ast$}}
\put(1439,845){\makebox(0,0){$\ast$}}
\put(1349,735){\makebox(0,0){$\ast$}}
\put(90.0,82.0){\rule[-0.400pt]{292.453pt}{0.800pt}}
\sbox{\plotpoint}{\rule[-0.200pt]{0.400pt}{0.400pt}}%
\put(90.0,82.0){\rule[-0.200pt]{0.400pt}{186.938pt}}
\put(90.0,82.0){\rule[-0.200pt]{324.974pt}{0.400pt}}
\put(1439.0,82.0){\rule[-0.200pt]{0.400pt}{186.938pt}}
\put(90.0,858.0){\rule[-0.200pt]{324.974pt}{0.400pt}}
\end{picture}

\caption{Plots of the orders of 100 subgroups sampled as $\langle u_1,\ldots,u_5\rangle\leq U(10,3)$
with three different densities $\epsilon$: $(+,\epsilon=1/10)$, $(\times,\epsilon=1/2)$,
and $(\ast,\epsilon=1)$. Greater density makes group order, and structure, less 
varied. The $X$-axis is labelled by the group order, while the $Y$-axis is 
labelled by the percentage of the sampled groups.
}\label{fig:scatter-plot}
\end{figure}
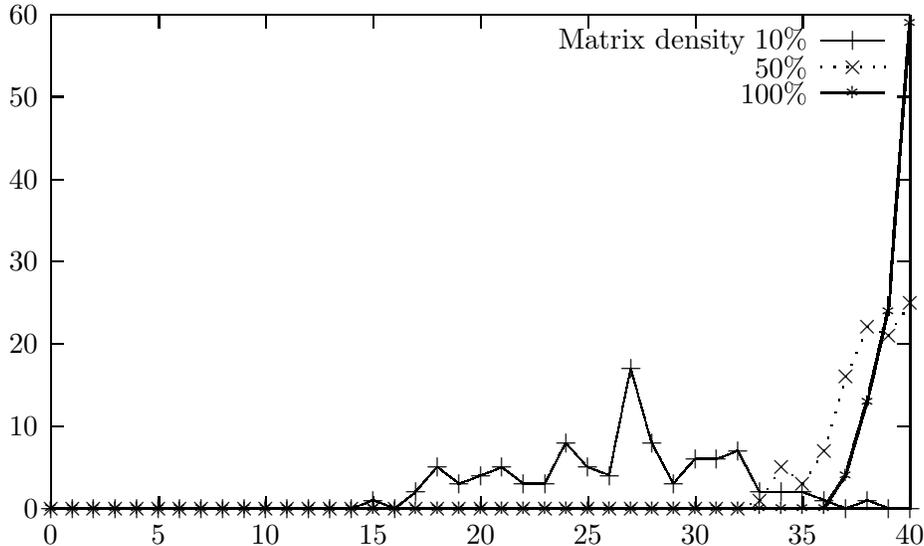

Secondly, once we have selected a suitably random upper unitriangular group $U$,
an extension to this group is selected by adding to its block-diagonal.  That process 
consists of choosing a partition of the series of common generalized eigen $1$-spaces
(the fixed point flag) of the group $U$. In each block we select a random (almost) quasisimple group
with a representation of dimension at most the size of the block.  We further allow for 
multiplicity and for permutations of isomorphic modules.  This extends $U$ 
first by a block-diagonal abelian group, then a product of simple groups, followed by
a layer of abelian groups, and a final layer of permutations.  It is well known that every finite
group has such a decomposition, often referred to as the 
\emph{Babai--Beals} filtration \cite{BB99}.  We note our own filtration descends
to the Fitting subgroup instead of to the solvable radical as in the Babai--Beals treatment;
revisit Figure~\ref{fig:extend-random} for an illustration.

Along with the proposal of such a model inevitably come questions as to its efficacy.
We address two of the more critical issues here. First, our model samples a large number of
groups:

\begin{mainthm}
\label{thm:coverage}
A random $d \times d$ group over $\mathbb{Z}/b$, as above, 
samples from each of the following classes of groups.
\begin{enumerate}[(i)]
\item finite abelian groups of exponent dividing $b$
and order at most $O(b^{d^2/2})$.
\item For each $\mathbb{Z}/b$-bilinear map $*:U\times V\bmto W$,
with ${\rm rank}~U+{\rm rank}~V+{\rm rank}~W\leq d$, the Brahana groups 
\cite{Bra35}
$Bh(*)=U\times V\times W$ with product (also denoted by $*$) as
\begin{align*}
    (u,v,w)*(u',v',w') & = (u+u',v+v',w+w'+u*v').
\end{align*} 
\item For each alternating $\mathbb{Z}/b$-bilinear map 
$*:U\times U\bmto W$,
with ${\rm rank}~U+{\rm rank}~W\leq d$, the Baer groups \cite{Bae}
$Br(*)=U\times W$ with product (also denoted by $*$) as
\begin{align*}
    (u,w)*(u',w') & = (u+u',w+w'+u*v').
\end{align*}

\item All classical groups $T(r,q)$ for rank
$r$ over $\F_q$ where $r\log q\leq d$.

\item All permutation groups of degree at most $d$.

\end{enumerate}
In particular this class of groups samples from $p^{\Theta(d^3)}$ pairwise non-isomorphic groups 
of order $p^d$ which is a logarithmically dense set of all isomorphism types of groups of order $p^d$.
Furthermore, this class of groups is closed to direct products and subdirect products.
\end{mainthm}

Secondly, for groups selected from our model, even a genus-1, 1-WL refinement results in a filter with 
constant \emph{average} width. (Note, constant max width would result in a polynomial-time isomorphism test.)

\begin{mainthm}\label{thm:random-refine}
For a random group $G\leq U(d,b)$ sampled by our model, one of the following cases 
occurs on average when $d$ and $b$ are large enough:
\begin{enumerate}[(a)]
\item $O_b(G)$ is abelian; or
\item $G$ has characteristic WL-filter refinement of length $\Theta(\log |G|)$.
\end{enumerate}
\end{mainthm}

It was predicted in~\cite{Wilson:alpha} that most $p$-groups $P$ had
characteristic filters of length $O(\log |P|)$, owing in part to a result of 
Helleloid--Martin \cite{Helleloid-Martin}.  However,
outside of examples in \citelist{\cite{Wilson:alpha}\cite{Maglione:filters}} there
where no large classes of groups where it could be demonstrated that such a filter could
be efficiently computed.  
In a survey of 500,000,000 groups of order $2^{10}$ 
conducted by J. Maglione and the fifth author, it was discovered that 
$96\%$ of groups admitted a filter refinement by algebraic methods, with most
stabilizing at $10=\log_2 1024$ terms.
Furthermore, in a sample of 100,000 $p$-groups having orders between
$100$ and $3^{70}$, most filters refined to a factor of about $10$ times the original length.
Theorem~\ref{thm:random-refine} offers a theoretical explanation for those experimental results.
    
\subsection{Testing pseudo-isometry of alternating bilinear maps}
\label{subsec: average-case in intro}
One base case for which the application of Weisfeiler--Leman is unlikely to go much 
further is $p$-groups of class $2$ and exponent $p$. (This special case has long been
considered as difficult as the general group isomorphism problem.)   
As we have seen in Baer's correspondence \cite{Bae} (cf.~Theorem~\ref{thm:coverage} (iii)), 
when $p$ is odd
testing isomorphism of such groups is equivalent to the following 
problem: given two alternating bilinear maps $\alpha, \beta: 
U\times U\bmto V$, decide whether they are pseudo-isometric, that is, whether they 
are the same under the natural action of $\GL(U)\times \GL(V)$. 

Let $\Lambda(n, q)$ denote the linear 
space of all $n\times n$ alternating matrices over $\F_q$, namely 
the $n\times n$ matrices $G$ such that $v^{t}Gv=0$ for all 
$v\in\F^n$. 
Note, $v\mapsto v^t$ and $G\mapsto G^t$ denotes transposition on vectors and matrices, respectively.
An alternating bilinear map $\alpha\colon U\times U\bmto V$ with 
$U\cong \F_q^n$ and 
$V\cong \F_q^m$ will be represented by an 
$m$-tuple of $n\times n$ such matrices. 
Testing pseudo-isometry of alternating bilinear maps translates to the following: 
given two $m$-tuples of $n\times n$ alternating matrices over $\F_q$, $\bG=(G_1, 
\dots, G_m)$ and $\bH=(H_1, \dots, H_m)$, decide whether there exists $T\in \GL(n, 
q)$, such that the \emph{linear spans} of $T^t\bG T:=(T^tG_1T, \dots, T^tG_mT)$ and $\bH$ 
are the same. For an odd prime $p$, testing the pseudo-isometry of 
alternating bilinear maps over $\F_p$ in time $p^{O(n+m)}$ is equivalent to testing 
isomorphism of $p$-groups of 
class $2$ and exponent $p$ in time polynomial in group order.
Also note that the na\"{\i}ve brute-force algorithm---enumerating all possible 
$T\in\GL(n, q)$---takes time $q^{n^2}\cdot \poly(n, \log q)$. 

In \cite{LQ} it was shown that when $n$ and $m$ are linearly 
related, for 
all but at most $1/q^{\Omega(n)}$ fraction of $\bG\in\Lambda(n, q)^m$, there is an 
algorithm in time $q^{O(n)}$ to test isometry of $\bG$ with an arbitrary 
$\bH\in\Lambda(n, q)^m$.\footnote{The main result in \cite{LQ} is stated in a 
so-called linear algebraic Erd\H{o}s--R\'enyi model. This model is not essentially 
different from sampling random alternating matrix tuples. See also 
Remark %s~\ref{rem:random_tuple} and~
\ref{rem:LinER} for some details.} 
The technique used to derive this result merits further comment. It was inspired by, 
and can be viewed as a linear algebraic analogue of, a 
classical combinatorial idea from graph isomorphism testing, namely the 
individualization and refinement technique. More specifically, it follows the use 
and analysis of this technique by Babai, Erd\H{o}s, and 
Selkow, in the first efficient average-case algorithm for graph isomorphism 
\cite{BES}.
By incorporating the genus concept \cite{BMW} into the individualization 
and refinement scheme as used in \cite{LQ,BES} we can both extend and 
improve this result and at the same time greatly simplify the algorithm. 
Indeed, we have implemented an effective version of this new algorithm
in {\sc Magma} \cite{MAGMA}. We prove:

\begin{mainthm}
\label{thm:average}
Suppose $m$ is larger than some constant. There is an algorithm that, for all 
but at most $1/q^{\Omega(nm)}$ fraction of $\bG\in\Lambda(n, q)^m$, 
tests the pseudo-isometry of $\bG$ to an arbitrary $m$-tuple of alternating 
matrices $\bH$, in time $q^{O(n+m)}$. 
\end{mainthm}

We briefly outline a simplified version of the algorithm, which 
is easy to describe and straightforward to implement. A more detailed description can be 
found in Section~\ref{subsec:average-case-algo}.
The simplified version has already captured the essence 
of the strategy, but it comes with two small drawbacks. First, it does not work over fields of 
characteristic $2$.
Secondly, the average-case analysis does not achieve the level stated in Theorem~\ref{thm:average}. 
Both issues will be remedied in the algorithm 
presented in Section~\ref{subsec:average-case-algo2}, followed by a rigorous 
average-case analysis. 

Assume we are given two $m$-tuples of 
 $\bG=(G_1, \dots, G_m)$ and 
$\bH=(H_1, \dots, H_m)$ from $\Lambda(n,q)^m$ for sufficiently large $m$ and odd $q$.
Let $\cH$ be the  subspace of $\Lambda(n, q)$ spanned by $\bH$.
Take the first $c$ matrices of $\bG$ to form a tuple $\bA=(G_1,\dots,G_c)$ for some constant $c<m$. 
Note, 
\begin{quotation}
\noindent {\em every pseudo-isometry from $\bG$ to $\bH$ maps $\bA$ to a $c$-tuple 
$\bB$ of matrices in $\cH$.}
\end{quotation} 
This simple observation leads to the following algorithm. 
(We say two $c$-tuples of alternating matrices $\bA$ and $\bB$ are 
{\em isometric} if there exists
an invertible matrix $T\in\GL(n,q)$ such that $T^t\bA T=\bB$,
and the {\em autometry} group of $\bA$ is $\{T\in\GL(n,q):T^t\bA T=\bA\}$.)
First, check if the autometry group of $\bA$ is too large (larger than 
$q^{\Omega(n)}$). If so, $\bG$ does not satisfy our generic condition. Thus, 
suppose the autometry group is not too large, and
enumerate all possible $c$-tuples $\bB$ in $\cH$. Exhaustively check if 
any of them is isometric to $\bA$, and, in the case of isometry, check if 
any isometry between $\bA$ and $\bB$ extends to an pseudo-isometry between $\bG$ 
and $\bH$. The number of isometries between $\bA$ and $\bB$ is also not too large, 
because it is equal to the order of the autometry group of $\bA$.

Note that the coset of isometries between $\bA$ and $\bB$ can be computed in time 
$\poly(n,c,\log q)$ over fields of characteristic not
$2$~\cite{BW:isom,IQ}. Enumerating all possible $c$-tuples in $\bH$ incurs a 
multiplicative cost $q^{cm}$. Given an isometry between $\bA$ and $\bB$, we can 
check whether $\bG$ and $\bH$ are pseudo-isometric in $\poly(n,m,\log q)$. Thus, 
the overall time complexity is bounded above by $q^{cm}\cdot s\cdot \poly(n,m,\log 
q)$, where $s$ is the
order of the autometry group of $\bA$. As 
we shall prove in Section~\ref{sec:LiQiao}, there is an absolute constant $c$ such that
for almost all $m$-tuple of $n\times n$ alternating matrices $\bG$, 
the first $c$ matrices have autometry group of order at most $q^{O(n)}$. Thus, 
the overall time complexity of the aforementioned isometry test is $q^{O(n+m)}$ 
for almost all $\bG$ and arbitrary $\bH$.
\medskip

\noindent {\bf Performance.}~ We implemented the above algorithm in {\sc Magma} with 
some key adjustments (see Section~\ref{subsec:average-case-algo} for 
details).
The implementation is publicly available on GitHub as part of a comprehensive
collection of tools---developed and maintained by the first and last authors
and their collaborators---to compute with groups, algebras, and multilinear functions~\cite{Git}. 

Absent additional characteristic structure that can be exploited, the traditional approach to deciding 
pseudo-isometry between alternating bilinear maps $\alpha,\beta\colon V\times V\bmto W$ is as follows.
Let $\hat{\alpha},\hat{\beta}\colon V\wedge V\to W$ denote the linear maps induced by $\alpha,\beta$.
Compute the natural (diagonal) action of $\GL(V)$ on $V\wedge V$, and decide if $\ker\hat{\alpha}$
and $\ker\hat{\beta}$---each of codimension $\dim W$ in $V\wedge V$---belong to the same orbit.
An alternative version of brute force is to enumerate
$\GL(W)$ and check if one of these transformations lifts to a pseudo-isometry from $\alpha$ to $\beta$.
Which of these two brute-force options represents the best 
choice depends on the dimensions of $V$ and $W$.

Our implementation is typically an improvement over both options.
For example, in a preliminary experiment, our implementation readily decides pseudo-isometry 
between randomly selected alternating bilinear maps
$\mathbb{F}_3^5\times \mathbb{F}_3^5\bmto \mathbb{F}_3^4$, while both brute-force options
failed to complete. Note that the worst-case for all methods should be when $\alpha,\beta$ are {\em not}
isometric, since in that case one must exhaust the entire enumerated list (or orbit) to confirm non-equivalence.
However, the modifications we made tend to detect non-equivalence rather easily, since other (easily computed) invariants typically
do not align in this case. We were therefore careful to also run tests with
equivalent inputs, so as to ensure a fair comparison with default methods.

\subsection{On groups with genus-$2$ radicals}
There are examples by the fifth author of non-isomorphic $p$-groups having all 
proper nontrivial subgroups 
of a common order isomorphic, and likewise for quotients \cite{Wilson:profile}.  
No amount of local invariants will distinguish such groups, so when a WL-refinement
style algorithm such as ours encounters such a group it can go no further.  Even so, those
examples are low genus and thus isomorphism can be decided efficiently by unrelated methods \cite{BMW}.  
However, should these groups arise as $O_p(G)$ for a non-nilpotent
group $G$ it remains to contend with them as a base case.
Combining the code equivalence technique of \cite{BCGQ}, the cohomological 
techniques of \cite{GQ}, and results on the automorphism groups of low-genus 
groups \cite{BMW}, we are able to get a nearly-polynomial running time for testing 
isomorphism in an important subclass of such groups.

\begin{mainthm} 
\label{thm:central}
Let $\mathcal{G}$ be the class of groups $G$ such that $\Rad(G)$---the largest solvable normal subgroup of $G$---is 
a $p$-group of class 2, exponent $p \neq 2$, such that $G$ acts on $\Rad(G)$ by
inner automorphisms. Given groups $G_1, G_2$ of order $n$, 
it can be decided in $\poly(n)$ time if they lie in $\mathcal{G}$. If so, isomorphism can be decided, and a generating 
set for $\Aut(G_i)$ found, in time $n^{O(g + \log \log n)}$, where $g$ is the genus of $\Rad(G)$.
\end{mainthm}

\paragraph{Structure of the paper.} After presenting some preliminaries in 
Section~\ref{sec:background}, we detail the construction of the colored 
hypergraphs and prove Theorem~\ref{thm:hypergraph} in 
Section~\ref{sec:hypergraph}. We then explain the combination of filters and composition series isomorphism in \GpI, proving Theorem~\ref{thm:width-and-color} in 
Section~\ref{sec:gen-iso}. The model of random groups, and the effect of the 
refinement procedure in this model, are the subject of 
Section~\ref{sec:random}, where Theorems~\ref{thm:coverage} 
and~\ref{thm:random-refine} are proved. Finally, we provide the average-case 
algorithm for $p$-groups of class $2$ and exponent $p$ (Theorem~\ref{thm:average}) in Section~\ref{sec:LiQiao}, and the worst-case algorithm for groups with 
genus-$2$ radical (Theorem~\ref{thm:central}) in Section~\ref{sec:genus2-rad}.

%%%%%
\section{Preliminaries}
\label{sec:background}
% !TEX root = WL.tex
\paragraph{Notation.} Let $[m]=\{1,\dots,m\}$ for $m\in\N$. We use 
$\gbinom{n}{d}{q}$ to denote the Gaussian binomial coefficient with 
parameters $n$, $d$ and with base $q$. Let $\Mat(n\times n',\F)$ (resp. $\Mat(n,\F)$) 
be the linear space of all $n\times n'$ (resp. $n\times n$) matrices over $\F$. 
The general linear group of degree $n$ over $\F$ is denoted by $\GL(n,\F)$.
When $\F=\F_q$ for some prime power $q$, we write simply  
$\Mat(n,q)$ and $\GL(n,q)$ in place of $\Mat(n,\F_q)$ and $\GL(n,\F_q)$.

\paragraph{Definitions of bilinear maps.} Let $U, V, W$ be vector spaces over a field $\F$.  
A ($\F$-)bilinear map is a function $\alpha\colon U \times V \bmto W$  such that 
\begin{align*}
(\forall u\in U,\;\forall v,v'\in V,\;\forall a,b\in\F) & &
\alpha(u,av + bv') = a\alpha(u,v) + b \alpha(u,v')\\
(\forall u,u'\in U,\;\forall v\in V,\;\forall a,b\in\F) & &
\alpha(au+bu',v) = a\alpha(u,v) + b \alpha(u',v).
\end{align*} 
If  $\beta\colon U' \times V' \to W'$ 
is another $\F$-bilinear map, we regard $\beta$ as a function
on the same domain and codomain as $\alpha$ by selecting arbitrary
linear isomorphisms $U\to U'$, $V\to V'$, and $W\to W'$. We say
 $\alpha,\beta\colon U\times V\bmto W$
are \emph{isotopic} if there exists $(f,g,h)\in\GL(U)\times\GL(V)\times\GL(W)$ such 
that $\beta(f(u), g(v)) = h(\alpha(u,v))$ for all $u \in U, v 
\in V$, and \emph{principally isotopic} if there is an isotopism of the form $(f,g,1_W)$.
If $U=V$, we often require that $f=g$. We say $\alpha,\beta\colon V \times V \to W$ are 
\emph{pseudo-isometric} if there 
is an isotopism of the form $(g,g,h)$, and that they are {\em isometric}
if there is a pseudo-isometry of the form $(g,g,1_W)$. A bilinear map 
$\alpha\colon V\times V\to W$ is \emph{alternating}, if for any $v\in V$, 
$\alpha(v, v)=0$.

\paragraph{Computational models.} Suppose, after fixing bases, that $U=\F^\ell$, 
$V= \F^n$, and $W= \F^m$, which we regard as column spaces. 
A bilinear map $\alpha\colon U\times V\bmto W$ can be 
represented as a tuple of matrices $\bA=(A_1, \dots, A_m)\in \Mat(\ell\times n, \F)^m$, where 
\begin{align*}
(\forall u\in U,\,v\in V) && \alpha(u, 
v)=(u^tA_1v, \dots, u^tA_mv)^t.
\end{align*}
Suppose $\beta:U\times V\to W$ 
is represeted by $\bB=(B_1, \dots, B_m)\in\Mat(\ell\times n, \F)^m$. The concepts of isotopism and principal isotopism then have natural and straightforward interpretations in terms of these matrices. Namely, we say $\bA,\bB\in \Mat(\ell\times n, \F)^m$ are isotopic, if there exist invertible matrices 
$T\in\GL(\ell,\F)$, $S\in\GL(n,\F)$ and $R\in\GL(m,\F)$, such that 
\begin{align*}
T^t\bA S=(T^tA_1S,\dots,T^tA_mS)=\left(\sum_{i=1}^mr_{1,i}B_i,\dots,\sum_{i=1}^mr_{m,i}B_i\right)=\bB^R,
\end{align*}
where $r_{i,j}$ denotes the $(i,j)$-th entry of $R$ for $i,j\in[m]$. We say $\bA$ and $\bB$ are principal isotopic if they are isotopic with $R=I_m$.

Similarly, an alternating bilinear map $\alpha:V\times V\bmto W$ can be represented 
by a tuple of alternating matrices. Recall that an $n\times n$ matrix $G$ 
over $\F$ is alternating if for every $v\in \F^n$, $v^tGv=0$. When $\F$ is not of 
characteristic $2$, this is equivalent to the skew-symmetry condition. 
Let $\Lambda(n,\F)$ be the linear space of all $n\times n$ alternating matrices over 
$\F$ (and $\Lambda(n,q)$ when $\F=\F_q$). 
Then {\em pseudo-isometry} and {\em isometry} have analogous formulations in terms of alternating 
matrix tuples.

Given two tuples of alternating matrices $\bG,\bH\in\Lambda(n,q)^m$, the set of isometries between $\bG$ and $\bH$ is 
denoted as 
\[
\Isom(\bG,\bH)=\{T\in\GL(n,\F) : T^t\bG T=\bH\}; 
\]
the group of autometries (or self-isometries) of $\bG$ is 
denoted as $\Aut(\bG)=\Isom(\bG,\bG)$. 
The set of pseudo-isometries between $\bG$ and $\bH$ is defined as
\[
\pseudo(\bG,\bH)=\{T\in\GL(n,\F) : \exists~T'\in\GL(m,q),~T^t\bG T=H^{T'}\};
\] 
the group of pseudo-autometries (or self-pseudo-isometries) of $\bG$ is 
denoted as $\pAut(\bG)=\pseudo(\bG,\bG)$. It is straightforward to see that 
$\Isom(\bG,\bH)$ is a (possibly empty) coset of $\Aut(\bG)$, and 
$\pseudo(\bG,\bH)$ is a (possibly empty) coset of $\pAut(\bG)$.

\paragraph{Some algorithms for bilinear maps.}

 We note several of the algorithms we cite are described as
Las Vegas randomized algorithm in that they depend on factoring polynomials over finite
fields.  That is known to be deterministic if the characteristic of the field is bounded.
In our input model we are given a list of the group elements, so all primes are bounded and so we cite
these as deterministic algorithms.

\begin{thm}
\label{thm:isometry-algorithms}
Let $\alpha,\beta \colon U \times V \bmto W$  be bilinear maps
of vector spaces over a finite field $\F$.
\begin{enumerate}
\item In time $\poly(\dim U, \dim V, 
|\F|)$ 
one can decide if $\alpha,\beta$ are principally isotopic~\cite{BOW}*{Theorem~3.7}.
\item If $U=V$ and the characteristic of $\F$ is not $2$, in 
time $\poly(\dim U, |\F|)$ one can decide if $\alpha,\beta$ 
are isometric~\cite{IQ}. 
\end{enumerate}
In each case an affirmative answer is accompanied by a principal isotopism (or isometry).
\end{thm}

We also require the following, which follows directly from Theorem~\ref{thm:isometry-algorithms}
by enumerating $\GL(W)$. 
\begin{thm} 
\label{cor:isometry-algorithms}
Let $\alpha,\beta \colon U \times V \to W$  be bilinear maps
of vector spaces over a finite field $\F$.
\begin{enumerate}
\item In time $\poly(\dim U, \dim V, 
|W|^{\dim |W|})$ one can decide if $\alpha,\beta$ are isotopic.\cite{BOW}
\item If $U=V$ and the characteristic of $\F$ is not $2$,  in time $\poly(\dim 
U, |W|^{\dim |W|})$ one can decide if $\alpha,\beta$ are pseudo-isometric \cite{IQ}.
\end{enumerate}
\end{thm}

The following theorem is the automorphism version of 
Theorem~\ref{thm:isometry-algorithms}. Note that, unlike the case of graph 
isomorphism, for the problems here there are no known reductions from the 
isomorphism version to the automorphism version. 

\begin{thm}\label{thm:autometry-algorithms}
Let $\alpha \colon U \times V \to W$  be a bilinear map
of vector spaces over a finite field $\F$.
\begin{enumerate}
\item In time $\poly(\dim U, \dim V,
|\F|)$, one can compute a generating set for the group of principal 
autotopisms of 
$\alpha$ \cite{BOW}. 
\item If $U=V$ and the characteristic of $\F$ is not $2$,  in  time
$\poly(\dim U, |\F|)$, one can compute a generating set for the group of 
autometries of $\alpha$ \cite{BW:isom}.
\end{enumerate}
\end{thm}

\begin{remark}[Shuffles] \label{rmk:shuffles}
A bilinear map $*:U\times V\bmto W$ can be encoded as a $3$-dimensional array.  Transposing
that array allows us to change swap the roles of $U, V,W$, for example
creating a bilinear map $*:V\times U\bmto W$ or $*:W^{\dagger} \times V\bmto U^{\dagger}$,
etc. (Here $U^{\dagger}$ is the dual space of $U$).  This swapping is functorial
and therefore isotopisms are permuted accordingly; cf. \cite{BOW}.  So while we 
highlight the situation for principal isotopisms we could indeed specialize
any one of the three spaces.  We shall assume throughout that when necessary a bilinear
map is shuffled.
\end{remark}

%%%%%
\section{The colored hypergraph algorithm} 
\label{sec:hypergraph}
% !TEX root = WL.tex
% new version of lcolored hypergraph algorithm drafted by PAB
\newcommand{\WL}{\text{WL}}
\newcommand{\Iso}{\text{Iso}}

A high-level description of our algorithm to construct a colored hypergraph associated to a finite group
was given in the introduction. We now provide the details; for convenient reference, 
an outline is given in  Algorithm~\ref{algo:outline} below.

%%%  
\begin{algorithm}
\caption{Colored Hypergraph}
\label{algo:outline}
\begin{algorithmic}[1]
\Require a finite group $G$, and integers $g,k \geq 1$
\Ensure a characteristic filter $\phi\colon \N^d\to \Norm(G)$ and a colored 
hypergraph $\mathcal{H}_{\chi}^{(g,k)}(\phi)$ upon which
$\Aut(G)$ acts as color-preserving automorphims.
\vspace*{2mm}

\State $\phi\;\leftarrow$ initial characteristic filter for $G$. \hfill Section~\ref{subsec:WL-outline}
\vspace*{1mm}

\State Repeat the following steps until $\phi$ stops changing (stabilizes):
   \begin{enumerate}[a:] 
   \item Build $\mathcal{H}_{\chi}^{(g)}(\phi)$ on each layer of $\phi$. \hfill  Section~\ref{subsec:ind-color} 
   \vspace*{-1mm}
   
   \item Extend $\mathcal{H}_{\chi}^{(g)}(\phi)$ between layers of $\phi$. \hfill  Section~\ref{subsec:mix-color} 
   \vspace*{-1mm}
   
   \item Apply $k$-dimensional Weisfeiler--Leman to $\mathcal{H}_{\chi}^{(g)}(\phi)$ \hfill  Section~\ref{subsec:WL} 
   \vspace*{-1mm}
   
   \item $S\;\leftarrow\;\{\Aut(G)$-invariant subgroups extracted from $\WL(k,\mathcal{H}_{\chi}^{(g)}(\phi))\}$. \hfill  Section~\ref{subsec:char} 
   \vspace*{-1mm}
   
   \item Refine $\phi$ using $S$. \hfill  Section~\ref{subsec:refine} 
   \vspace*{-1mm}
   
   \end{enumerate}
\State Return $\phi$ and $\WL(k,\mathcal{H}_{\chi}^{(g)}(\phi))$.
\end{algorithmic}
\end{algorithm}

%%%%%
\subsection{Coloring within layers: low-genus quotients and restrictions}
\label{subsec:ind-color}
\newcommand{\PG}{\text{PG}}
For $s\in\N^d$, $L_s$ is a $\mathbb{Z}_{p}$-vector space for some prime $p=p_s$ of 
dimension $d_s$. Recall that for any vector space $L$, $\PG(L)$ denotes the 
\emph{projective geometry} of $L$, which we may think of as a poset whose elements 
are the vector subspaces of $L$, (partially) ordered by inclusion, and we use 
$\PG_k(L)$ to denote the set of 
$k+1$-dimensional subspaces. Let 
$L_s^{*}=\Hom(L_s,\mathbb{Z}_{p})$ denote the set of linear maps from $L_s$ to 
$\Z_p$, i.\,e., the dual vector space of $L_s$. Then
the map $X\mapsto X^*=\{\nu\in L_s^*\colon \nu(X)=0\}$ is an order-reversing bijection 
${\rm PG}(L_s)\to{\rm PG}(L_s^*)$. By the Fundamental Theorem of Projective 
Geometry, there is a bijective linear 
transformation\footnote{We note that in some cases, it makes sense to consider a layer $L_s$ as being defined over a larger field $\F_{p^k}$, thus effectively reducing its dimension, and reducing the size of the hypergraph. In such cases, this map is only guaranteed to be \emph{semi}-linear, that is, $f_s(a+b) = f_s(a) + f_s(b)$, but $f_s(\lambda a) = \alpha(\lambda) f_s(a)$ where $\alpha \in \Gal(\F_{p^k})$ is an automorphism of the field $\F_{p^k}$. This doesn't present any essential difficulties, but needs to be kept track of.}
$f_s\colon L_s\to L_s^*$ such that $X^*=f_s(X)$. Let $b_s\colon L_s\times L_s\bmto \mathbb{Z}_p$
be the linear form defined by $b_s(x,y)=f_s(y)(x)$. For 
$X\leq L_s$, let $X^{\perp}=\{x\in L_s\colon b_s(x,X)=0\}$.

The vertices and hyperedges of $\mathcal{H}^{(g)}(\phi)$ are, respectively,
\begin{align}
\label{eq:vertices}
\cV=\bigcup_{s\in\N^d}{\rm PG}_0(L_s), && \cE=\left( \bigcup_{s\in\N^d : \dim L_s > 
g} \left( {\rm PG}_{g-1}(L_s) \cup {\rm PG}_{d_s - g - 1}(L_s)\right) \right) \cup 
\bigcup_{s \in \N^d : \dim L_s \leq g} {\rm PG}_{d_s}(L_s).
\end{align}
(Recall that $L_s\cong \Z_{p_s}^{d_s}$.)
To regard $X\in{\rm PG}_{d}(L_s)$ as a hyperedge, when convenient we identify the 
$d$-subspace
$X$ with the set of points (1-spaces) it contains. The initial coloring is as follows. 
\begin{itemize}
\item Vertices. The initial color $\chi(v)$ of a vertex $v \in \cV$ is simply the index $s$ of the layer $L_s$ such that $v \in {\rm PG}_0(L_s)$. 

\item Hyperedges corresponding to subspaces of codimension $g$ (dimension $d_s - g - 1$), when $\dim L_s > g$. The initial color $\chi(X)$ of these hyperedges $X\in {\rm PG}_{d_s - g-1}(L_s)$ is determined by $s$ together with a
set of labels indexed by pairs $t,u \in \N^d$ such that $t+u = s$ as follows:
 if $t \neq u$, the label of $X$ corresponding to the pair $(t,u)$ is the isotopism 
 type of the projection $L_t \times L_u \bmto L_s\to L_s / X^{\perp}$; when $t=u$ it is the 
 pseudo-isometry type of this projection. 
 
 \item Hyperedges corresponding to subspaces of dimension $g$ (elements of ${\rm PG}_{g-1}(L_s)$), when $\dim L_s > g$. The initial color $\chi(X)$ of these hyperedges is determined by $s$ together with a set of labels indexed by $t \in \N^d$ $t \neq s$ as follows: the label of $X$ corresponding to $t$ is the isotopism type of the restriction of the bimap $L_s \times L_t \bmto L_{s+t}$ to $X \times L_t \bmto L_{s+t}$. (When the dimension is such that dimension $g$ and codimension $g$ subspaces are the same, this set of labels is appended to the set of labels for codimension $g$ subspaces; the two sets of labels are kept separate by their indexing.)
 
 \item Hyperedges when $\dim L_s \leq g$. In this case, there is only a single hyperedge $X$ corresponding to the entire layer $L_s$. It is given a color that is similar to the previous two, namely, for each $t,u \in \N^d$ such that $t + u = s$, $\chi(X)$ gets a set of labels indexed by the pairs $(t,u)$, labeled by the isotopism type of $L_t \times L_u \bmto L_s$ (resp., pseudo-isometry type if $t=u$), together with, for each $t \in \N^d$ (now including $t=s$) the isotopism (resp., pseudo-isometry) type of the bimap $L_s \times L_t \bmto L_{s+t}$.
\end{itemize}
 
 Observe, one need not pre-compute all isotopism (resp. pseudo-isometry) types. Instead,
 one can generate labels on the fly by pairwise comparison. Namely, given a new hyperedge $X$ to label,  
 test for isotopism (or pseudo-isometry) between $L_t\times L_u \bmto L_s/X^{\perp}$ and all distinctly labelled  
 $L_t\times L_u \bmto L_s/Y^{\perp}$, introducing a new label for $X$ if necessary.

 By Theorem~\ref{cor:isometry-algorithms}, isotopism and pseudo-isometry of bilinear maps $U \times V \to W$ 
 can be decided in time $\poly(\dim U, \dim V, |W|^{\dim W})$, and also (by Remark~\ref{rmk:shuffles}) in time $\poly(|U|^{\dim U}, \dim V, \dim W)$. 
 (When $g=2$, this can be decided very efficiently using the algorithm in~\cite{BMW}.)
 It follows that we can label all hyperedges 
 in time $|G|^{O(g)}$. Note that if the charactistic is $2$, then even for maps of the form $L_s \times L_s \bmto L_{s+s}$,  we only use the isotopism
 label instead of pseudo-isometry label, because the results of \cite{IQ} are not yet known to extend to characteristic 2. While this is less refined information, it is still useful.

%%%%%
\subsection{Coloring between layers}
\label{subsec:mix-color}
The colored hypergraph $\mathcal{H}_{\chi}^{(g)}(\phi)$ described in the previous section already 
contains much local information from which global characteristic structure may be inferred, extracted, and used.
However, we can often elucidate further characteristic structure by examining individual commutator relations 
{\em between} the layers. Of the various possible strategies one could try, we propose one that is both
elementary and effective.

For each distinct pair $s,t\in\N^d$, add to $\cE$ the the following edges. For each $x \in L_s, y \in L_t$ such that $[x,y] = 0$ in $L_{s+t}$ (that is, $[x,y] \in \partial \phi_{s+t}$), we add an edge from $x$ to $y$. For each $x,y$ which do not commute modulo $\partial \phi_{s+t}$, we add a hyperedge of size 3, connecting $x \in L_s$, $y \in L_t$, and $[x,y] \in L_{s+t}$. Upon refinement, this allows the vertex colors within each layer to affect the colors in the other layers.

%%%%%
\subsection{The Weisfeiler--Leman procedure}
\label{subsec:WL}
Given a vertex-and-hyperedge-colored (hereafter just ``colored'') hypergraph $H = (\cV, \cE, \chi)$, where $\chi\colon \cV \cup \cE \to C$ ($C$ a finite set of colors), we show here how to apply the $k$-dimensional Weisfeiler--Leman procedure $k$-WL, originally developed in the context of graphs independently by Babai--Mathon \cite{Babai79} and Immerman--Lander \cite{ImmermanLander} (see \cite{CFI} and \cite{Bab16} for more detailed history). For the case of $k=1$ (color refinement) applied to hypergraphs, the same procedure was proposed and studied in the very recent preprint by B\"{o}ker \cite{Boker}. In particular, B\"{o}ker shows that when we consider a graph as a (2-uniform) hypergraph, this procedure coincides with the usual color refinement procedure on graphs.

Let $\WL(k,H)$ denote the colored hypergraph resulting from applying $k$-WL to $H$. The two key properties we will need in our application of this procedure are that: (1) $\WL(k,H)$ can be computed from $H$ in $|H|^{O(k)}$ time, and (2) If $H'$ is another colored hypergraph, then $H$ and $H'$ are isomorphic (as colored hypergraphs) iff $\WL(k,H)$ and $\WL(k,H')$ are isomorphic as colored hypergraphs. (In fact, the set of isomorphisms will be the same: $\Iso(H, H') = \Iso(\WL(k,H), \WL(k,H'))$).

We find it simplest to describe the application of WL to hypergraphs by using instead their ``incidence (bipartite) graphs.'' We believe this bijection between vertex-and-edge-colored hypergraphs and vertex-colored bipartite graphs is essentially folklore; we include it here for completeness. Given a hypergraph $H = (\cV, \cE)$, its \emph{incidence graph} is the bipartite graph $I(H) = (V_L, V_R, E)$ where $V_L = \cV$, $V_R = \cE$, $E = \{(v,e) \in \cV \times \cE : v \in e\}$. It is not hard to see that every bipartite graph arises from a unique hypergraph in this manner, so $I$ is a bijection and $I^{-1}$ is well-defined.

An isomorphism between two vertex-and-edge-colored hypergraphs $H_i = (\cV_i, \cE_i, \chi_i)$ ($i=1,2$) is a bijection $f\colon \cV_1 \to \cV_2$ such that (1) $f(\cE_1) = \{ f(e) : e \in \cE_1 \} = \{ \{f(v) : v \in e\} : e \in \cE_1\} = \cE_2$, (2) $\chi_1(v) = \chi_2(f(v))$ for all $v \in \cV_1$, and (3) $\chi_1(e) = \chi_2(f(e))$ for all $e \in \cE_1$. We say that two vertex-colored bipartite graphs $G_i = (V_{L,i}, V_{R,i}, E_i, \chi_i \colon V_{L,i} \cup V_{R,i} \to C)$ ($i=1,2$) are isomorphic if there are bijections $f_L\colon V_{L,1} \to V_{L,2}$ and $f_R \colon V_{R,1} \to V_{R,2}$ such that (1) $f(E_1) = \{ (f_L(u), f_R(v)) : (u,v) \in E_1 \} = E_2$ and (2) $\chi_1(u) = \chi_2(f_L(u))$ for all $u \in V_{L,1}$ and $\chi_1(f_R(v)) = \chi_2(v)$ for all $v \in V_{R,1}$. 

\begin{prop}[Folklore]
Given two vertex-and-edge-colored hypergraphs $H_1, H_2$, there is a natural bijection between $\Iso(H_1, H_2)$ and $\Iso(I(H_1), I(H_2))$; in particular, $H_1$ is isomorphic to $H_2$ iff their vertex-colored bipartite incidence graphs are isomorphic. Furthermore, both $I$ and $I^{-1}$ can be computed in $O(V + E)$ time.\footnote{$V = |\cV|$ for hypergraphs and $|V_L| + |V_R|$ for bipartite graphs; $E = |\cE|$ for hypergraphs and $|E|$ for bipartite graphs.} 
\end{prop}

\begin{proof}[Proof sketch]
Notation as above. Given $\chi\colon \cV \cup \cE \to C$, a vertex-and-edge coloring on a hypergraph $H = (\cV, \cE)$, we get a coloring on the vertices of $I(H)$, which we also denote by $\chi$ by abuse of notation. The coloring on $V(I(H)) = V_L \cup V_R$ is the same as before, since $V_L = \cV$ and $V_R = \cE$. The inverse is similar. The running time results from the fact that $H$ and $I(H)$ can essentially be described by identical underlying data structures.

We show the natural bijection between $\Iso(H_1, H_2)$ and $\Iso(I(H_1), I(H_2))$. Given an isomorphism $f\colon \cV_1 \to \cV_2$ from $H_1$ to $H_2$, we define an isomorphism $\hat{f}$ from $I(H_1)$  to $I(H_2)$ in the natural way: $\hat{f}(v) = f(v)$ for $v \in V_{L,1} = \cV_1$, and for $e \in V_{R,1} = \cE_1$ we define $\hat{f}(e) = f(e)$, that is, $\hat{f}(e)$ is the vertex in $V_{R,1} = \cE_1$ which corresponds to the hyperedge $\{f(u) : u \in e\}$. To see that $\hat{f}$ is an isomorphism we must check that it preserves incidences and colors. For incidences, we have $(v,e) \in E(I(H_1))$ iff $v \in e$ (thinking of $v \in V_{L,1} = \cV_1$ and $e \in V_{R,1} = \cE_1$) iff $f(v) \in f(e)$ (since $f$ is an isomorphism of hypergraphs) iff $f(v) = \hat{f}(v) \in \hat{f}(e) = f(e)$, by the definition of $\hat{f}$. To see that the colors are preserved, for $v \in V_{L,1} = \cV_1$, we have, by definition (and abuse of notation), that $\chi(v) = \chi(f(v)) = \chi(\hat{f}(v))$, and for $u \in V_{R,1} = \cE_1$ we have $\chi(u) = \chi(f(u)) = \chi(\hat{f}(u))$. The inverse construction of an isomorphism $H_1 \to H_2$ from an isomorphism $I(H_1) \to I(H_2)$ is essentially gotten by reading all the preceding equations in reverse.
\end{proof}

Our $k$-WL procedure is to apply standard (graph) $k$-WL to $I(H)$, then applying $I^{-1}$ to get back a refined colored hypergraph. 

Finally, we recall the $k$-WL procedure as applied to a vertex-colored graph. If 
the graph is bipartite and we want to preserve the bipartition $(V_L, V_R)$---as 
in our setting---we assume that the vertices in $V_L$ have distinct colors from 
those in $V_R$. Given a vertex-colored graph $G = (V,E,\chi\colon V \to C)$, 
$k$-WL refinement is the following procedure. Each $k$-tuple of vertices $(v_1, 
\dotsc, v_k)$ is initially assigned a color according to its colored, ordered 
isomorphism type; that is, two such $k$-tuples $(v_1, \dotsc, v_k)$ and $(u_1, 
\dotsc, u_k)$ are given the same initial color iff (1) $\chi(v_i) = \chi(u_i)$ for 
all $i=1,\dotsc,k$, (2) $v_i = v_j$ iff $u_i = u_j$ for all $i,j \in [k]$, and (3) 
$(v_i, v_j) \in E(G)$ iff $(u_i, u_j) \in E(G)$ for all $i,j \in [k]$. Two 
$k$-tuples $v = (v_1, \dotsc, v_k)$ and $u$ are said to be $i$-neighbors if they 
are equal except that $v_i \neq u_i$. In each step of the refinement procedure, 
the coloring is refined as follows: the new color of a tuple $v$ is a $k$-tuple of 
multisets, where the $i$-th multiset is the multiset of colors of all the 
$i$-neighbors of $v$. At each stage, the coloring partitions $V^k$; the procedure 
terminates when this partition doesn't change upon further refinement. Once the 
coloring on $V^k$ has stabilized, we get a new coloring on $V = V(G)$ by defining 
$\chi'(v)$ for $v \in V$ to be the color of the diagonal $k$-tuple $(v,v,\dotsc,v) 
\in V^k$. We denote the resulting colored graph by $\WL(k,G)$. From $G$, 
$\WL(k,G)$ can be trivially computed in time $O(k^2 n^{2k+1})$; the current 
best-known running time is still $O(k^2 n^{k+1} \log n)$ 
\cite[Section~4.9]{ImmermanLander}. For more details on running time, 
implementation, and the properties of $k$-WL on graphs, see, e.\,g., 
\cite{Weisfeiler, WL68, ImmermanLander,AFKV, DGR}. 

%%%%%
\subsection{Extracting characteristic structure}
\label{subsec:char}
Each color class of vertices of $\WL(k, \mathcal{H}_\chi^{(g)}(\phi))$ provides (by lifting from $\phi_s / \partial \phi_s$ to $\phi_s$ along the natural projection) characteristic sub\emph{sets} of $G$, but not necessarily characteristic sub\emph{groups}; it is only the latter which can be used to refine the filter $\phi$. To get characteristic subgroups instead, we consider the subgroup generated by all the vertices in a given color class. We now write out this procedure more formally.

Let $\chi'$ denote the refined coloring function of $\WL(k,\mathcal{H}_\chi^{(g)}(\phi))$. For each $s \in \N^d$, let $\chi'_s$ denote the restriction of $\chi'$ to the vertices in ${\rm PG}_0(L_s)$. For each color $c$ in the image of $\chi'_s$, let $X_{s,c} = \sum_{x \in {\rm PG}_0(L_s) : \chi'(x)=c} \la x \ra$ be the subgroup of $L_s$ generated by the elements that are colored $c$. Finally, let $\pi_s \colon \phi_s \to \phi_s / \partial \phi_s = L_s$ be the natural projection; we lift $X_{s,c}$ to a characteristic subgroup of $\phi_s$ (and hence of $G$) as $\pi^{-1}(X_{s,c})$. 

Finally, the set of new characteristic subgroups we consider is 
\begin{align}
S &= \{ \pi_s^{-1}(X_{s,c}) \colon s\in\N^d,\,c\in \N\} - \{\phi_s\colon s\in \N^d \}.
\end{align}
If $S\neq\emptyset$, its members may be supplied to Theorem~\ref{thm:refine} to refine $\phi$, in which case step 3 of 
Algorithm~\ref{algo:outline} is repeated. If not,
then our colored hypergraph is now stable and Algorithm~\ref{algo:outline} terminates.

%%%%%
\subsection{Refining filters}
\label{subsec:refine}
One filter $\phi$ \emph{refines} another filter $\psi$ on the same group if the image of $\phi$ contains that of 
$\psi$ (the image is the collection of all subgroups in the filter). If $\phi$ is a characteristic filter and $H$ is a 
characteristic subgroup such that $\partial \phi_s \leq H \leq \phi_s$ for some $s$, then $\phi$ can be refined 
to a characteristic filter that includes $H$. This was first introduced in \cite{Wilson:alpha}, and shown to be 
computable in polynomial time by Maglione \cite{Maglione:filters}:

\begin{thm}[{\cite[Theorem~1]{Maglione:filters}}] 
\label{thm:refine}
Let $\phi$ be a filter on $G$, and $H \unlhd G$ such that there exists $s \in \N^d$ with $\partial \phi_s < H < \phi_s$. 
Then a filter refining $\phi$ and including $H$ can be computed in polynomial time. Furthermore, if $\phi$ and $H$ 
are characteristic, then so is the refined filter.
\end{thm}

We proceed sequentially through the characteristic subgroups of $S$, refining $\phi$ as we go.

%%%%%
\subsection{Proof of Theorem~\ref{thm:hypergraph}}
\label{subsec:proof}
For part (i), let $s\in\N^d-\{0\}$. Observe, if Step 3 (c) was omitted from Algorithm~\ref{algo:outline}, then colors would only be assigned
to hyperedges on points {\em in fixed layers}. In that case, moreover, the color of a hyperedge in layer $L_s$ is determined completely 
by pairs $t,u\in\N^d$ with $t+u=s$; the coloring function $\chi$ does not depend at all on layers $v\in \partial\phi_s$. That is to say,
if Step 3 (c) is omitted, then $\mathcal{H}_{\chi}^{(g)}(\phi)$ restricted to $N/\phi_s$ would be identical to the colored hypergraph
based on $\phi$ truncated at $\phi_s$. Step 3 (c) colors edges {\em between} layers using information from layers `lower' in the filter;
this means the restricted hypergraph is a refinement of the hypergraph on the truncated filter.

For part (ii), let $G$ and $G'$ be two finite groups.
Suppose we first construct $\mathcal{H}_{\chi}^{(g)}(\phi)$. 
Next, we construct $\mathcal{H}_{\chi'}^{(g)}(\phi')$ introducing new color for $\chi'$
only when it is new to both colored hypergraphs. Evidently, if $G\cong G'$, then $\mathcal{H}_{\chi}^{(g)}(\phi)$ and
$\mathcal{H}_{\chi'}^{(g)}(\phi')$ are isomorphic with identical color sets.

Finally, we analyze the running time. Computing the Fitting subgroup $O_\infty(G)$ and the initial characteristic filter can be done in $\poly(|G|)$ time, even by naive algorithms (which can be improved significantly when $G$ is given by generating permutations, generating matrices, or black-box generators). Building the hypergraph $\mathcal{H}_\chi^{(g)}(\phi)$ can be done in time linear in the number of hyperedges, which is the number of codimension-$g$ subspaces of each layer $L_s$, which is $\sim |L_s|^{O(g)}$, and thus in total is at most $|G|^{O(g)}$. The hyperedges can then be colored in $\poly(|G|) \times |G|^{O(g)} = |G|^{O(g)}$ time using the isotopism and isometry algorithms (Theorem~\ref{cor:isometry-algorithms}). As with $k$-WL for graphs, $k$-WL for hypergraphs can be computed in $|V + E|^{O(k)}$, which in our case is $|G|^{O(gk)}$. Extracting the characteristic subgroups from $\WL(k, \mathcal{H}_\chi^{(g)}(\phi))$ can easily be done in $\poly(|G|)$ time, and refining the filter $\phi$ can then be done in $\poly(|G|)$ time as well \cite{Maglione:filters} (reproduced as Theorem~\ref{thm:refine} above). The only remaining question is how many times the main refinement loop can run. Because we only refine when a characteristic subgroup $K$ is found which lies strictly in between some $\phi_s$ and $\partial \phi_s$, and the indices $|\phi_s : K|$ and $|K : \partial \phi_s|$ are both at least 2, refinement can happen at most $\log_2 |G|$ times. Thus the total running time is $|G|^{O(gk)} \log |G| = |G|^{O(gk)}$. \qed

%%%%%
\subsection{Incorporating additional invariants}
\label{subsec:rep-refine}
Our algorithm is not particular to the initial characteristic filter we choose. In any given group class, further characteristic subgroups (or subsets, or collections of subgroups) may be available which could be used to refine the filter, either at the beginning, or in each iteration of the main loop of Algorithm~\ref{algo:outline}. We give two examples here without much discussion, just to illustrate the concept, without detracting from the main foci of the paper.

First, it may be the case that some of the bimaps $L_s \times L_t \to \L_{s+t}$ are defined over a field larger than $\Z_p$, i.\,e., $\F_{p^k}$ for some $k > 1$. If this is true for sufficiently many of the bimaps, we may be able to treat some layers $L_s$ entirely over $\F_{p^k}$, thus reducing their dimension by $k$, and reducing the number of vertices in the corresponding factor of the hypergraph by a factor of $k$ in the exponent (from $p^{k\ell}$ to $p^\ell$).

Second, as $G$ acts on $N$ by conjugation, and the layers of $\phi$ are $\Aut(G)$-invariant,  
for each $s\in\N^d$ we can compute a linear representation of $G/N$
on the elementary abelian layer $L_s:=\phi_s/\partial\phi_s$. Using standard module machinery---for example, the version
of the Meataxe algorithm described in~\cite{HR}---in
time polynomial in $\log|G|$ each $G/N$-module may be decomposed first into indecomposable summands,
and then into isotypic components. The collection of isotypic components is a characteristic subset of subgroups---namely, they can be permuted amongst themselves by the action of $\Aut(G)$, but that's it. We can either group these into $\Aut(G)$-orbits of isotypic components to get characteristic subgroups to refine the filters, or keep the characteristic subset of subgroups and incorporate it into Rosenbaum's composition series isomorphism technique, discussed in Section~\ref{sec:gen-iso}.

\subsection{The procedure through an example}

We examine the procedure with a toy example as follows. Consider the following 
alternating matrix tuple in $\Lambda(4, 3)^3$, which was also considered in 
\cite{BOW}.
$$
\bA=(A_1, A_2, A_3)=
\begin{pmatrix}
\begin{bmatrix}
0 & 1 & 0 & 0 \\
-1 & 0 & 0 & 0 \\
0 & 0 & 0 & 1 \\
0 & 0 & -1 & 0
\end{bmatrix},
\begin{bmatrix}
0 & 0 & 0 & 0 \\
0 & 0 & 1 & 0 \\
0 & -1 & 0 & 0 \\
0 & 0 & 0 & 0 
\end{bmatrix},
\begin{bmatrix}
0 & 0 & 0 & 1 \\
0 & 0 & 0 & 0 \\
0 & 0 & 0 & 0 \\
-1 & 0 & 0 & 0 
\end{bmatrix}
\end{pmatrix}.
$$

We construct a bipartite graph $G_\bA=(L\cup R, E)$, where $L={\rm PG}_0(\F_3^3)$, 
and $R={\rm PG}_1(\F_3^3)$, so that for $v\in L$ and $U\in R$, $(v, U)\in E$ if 
and only if $v\in U$. In particular, note that $|L|=|U|=13$. 

For each $v\in L={\rm PG}_0(\F_3^3)$, we choose a non-zero vector on $v$ as 
its representative. So 
\begin{multline*}
L=\{(0,0,1), (0,1,0), (0,1,1), (0,1,2), (1,0,0), (1,0,1), \\
(1,0,2), (1,1,0), (1,1,1), (1,1,2), (1,2,0), (1,2,1), (1,2,2)\}.
\end{multline*}

For each $U\in R={\rm PG}_1(\F_3^3)$, since $U$ is a $2$-dimensional subspace of 
$\F_3^3$, we choose one defining linear equation $u^*$, $u\in \F_3^3$, as its 
representative. So $U$ is also 
\begin{multline*}
U=\{(0,0,1), (0,1,0), (0,1,1), (0,1,2), (1,0,0), (1,0,1), \\
(1,0,2), (1,1,0), (1,1,1), (1,1,2), (1,2,0), (1,2,1), (1,2,2)\}.
\end{multline*}
In this notation, $v=(v_1, v_2, v_3)\in L$ connects to $u=(u_1, u_2, u_3)\in U$, 
if and only if $v_1u_1+v_2u_2+v_3u_3= 0$. 

For $v=(v_1, v_2, v_3)^t\in L$, we define $\bA_{v}=v_1A_1+v_2A_2+v_3A_3$ in 
$\Lambda(4, 3)$. We use 
$rk(\bA_{ v})$ to give $v$ the vertex color. Using red for rank $2$ and blue 
for rank $4$, we have 
\begin{multline*}
L=\{{\color{red} (0,0,1), (0,1,0)}, {\color{blue} (0,1,1), (0,1,2), (1,0,0), 
(1,0,1)}, \\
{\color{blue} (1,0,2), (1,1,0), (1,1,1)}, {\color{red} (1,1,2)}, {\color{blue} 
(1,2,0)}, {\color{red} (1,2,1)}, {\color{blue} (1,2,2)}\}.
\end{multline*}

The first step of refinement uses the colors on the $L$ side to color the vertices 
on the $U$ side. For example, $(0, 0, 1)$ on the $U$ side is adjacent to 
${\color{blue} (1,0,0)}$, ${\color{red} (0,1,0)}$, ${\color{blue} (1,1,0)}$, 
${\color{blue} (1,2,0)}$. So $(0,0,1)$ obtains the color as ``3 blues and 1 red'', 
or 3B1R for short. We therefore let blue for 4B,  
green be 3B1R, and red for 2B2R. 
We then have
\begin{multline*}
U=\{{\color{green} (0,0,1)}, {\color{green} (0,1,0)}, {\color{red} (0,1,1)}, 
{\color{blue} (0,1,2)}, {\color{red} (1,0,0)}, {\color{red} (1,0,1)}, \\
{\color{red} (1,0,2)}, {\color{red} (1,1,0)}, {\color{blue} (1,1,1)}, 
{\color{green} (1,1,2)}, 
{\color{red} (1,2,0)}, {\color{green} (1,2,1)}, 
{\color{blue} (1,2,2)}\}.
\end{multline*}
Note that these colors, which comes from genus-1 
information, already gives the genus-2 isomorphism types. 

The second refinement uses 
the colors of the $U$ side to recolor the vertices on 
the $L$ side. For example, $(1,1,1)$ on the $L$ side is adjacent to ${\color{blue} 
(1,1,1)}$, 
${\color{red} (1,2,0)}$, ${\color{red} (1,0,2)}$, ${\color{blue} (0,1,2)}$ on the 
$U$ side. So $(1, 1, 1)$ obtains the color as ``2 blues and 2 reds'', or 2B2R for 
short. We therefore let red for 3R1G, blue for 1R2G1B, and green for 2R2B. We then 
have 
\begin{multline*}
L=\{{\color{red} (0,0,1), (0,1,0)}, {\color{blue} (0,1,1)}, {\color{green} 
(0,1,2)}, {\color{blue} (1,0,0), 
(1,0,1)}, \\
{\color{blue} (1,0,2), (1,1,0)}, {\color{green} (1,1,1)}, {\color{red} (1,1,2)}, 
{\color{blue} 
(1,2,0)}, {\color{red} (1,2,1)}, {\color{green} (1,2,2)}\}.
\end{multline*}
It can be checked that we reach at a stable coloring after this step.

Note that these colors suggest the green points form a characteristic set. This 
characteristic set would generate the whole group, so it does not yield a 
non-trivial characteristic subgroup. However, this characteristic set is already 
interesting, because it does suggest that the Weisfeiler--Leman procedure, or even 
the naive refinement, gives non-trivial information regarding group elements under 
the action of automorphisms. We discuss how to take advantage of such characteristic subsets in isomorphism testing in the next section.

%%%%
\section{Isomorphism testing using the filter and hypergraph}
\label{sec:gen-iso}
% !TEX root = WL.tex
Our colored hypergraph and filter constructions can be used to refine the composition-series isomorphism method of 
Rosenbaum and Wagner \cite{RosenbaumWagner}, thereby speeding up the resulting isomorphism test. 
Here, we present an isomorphism algorithm which runs in $\poly(|G|)$ time if the filter output by 
Algorithm~\ref{algo:outline} with $k,g \leq O(1)$ also has ``width'' at most $O(1)$ (defined below). 
Though we claim no asymptotic improvements in the worst case, 
we expect our test to perform well for many specific group classes, 
as well as for groups chosen randomly (including groups selected from the random model we 
discuss in detail in the following section). 
In practice, one should also apply Rosenbaum's bidirectional collision technique \cite{RosenbaumBidirectional} to get a square-root speed-up, but this causes no new technical difficulties.

In fact, the running time we get is $n^{(1/2)\width(\phi_{g,k}) + O(1)} + n^{O(gk)}$. We note that the largest that both $g$ and the width can be is $\log n$; if we allow $g$ to be near-maximal (take $kg = \log n / \log \log n$), and this results in a filter whose width is just slightly \emph{less} than maximal, say, $O(\log n / \log \log n)$, then the entire algorithm runs in time $n^{O(\log n / \log \log n)}$, asymptotically beating the trivial algorithm by a $\log \log n$ factor in the exponent. Because this is such a generous bound on $kg$ and a weak desired outcome for the width, we expect this runtime to hold for many classes of groups.

We begin with a simple version, building up to Theorem~\ref{thm:width-and-color} in steps.

\subsection{Simple version: choose composition series compatible with the filter}
We begin by recalling the composition-series isomorphism technique of Rosenbaum and Wagner \cite{RosenbaumWagner}, and show the simplest way to incorporate our characteristic filter into that technique. (Recall that, although we are not using the colored hypergraph here directly, it contributed to the construction of the filter.) Composition Series Isomorphism is the following problem: given two groups $G,H$, and a composition series of each $1 \unlhd G_1 \unlhd \dotsb \unlhd G_m = G$ and $ 1\unlhd H_1 \unlhd \dotsb \unlhd H_m = H$, decide whether there an isomorphism $\varphi \colon G \to H$ such that $\varphi(G_i) = H_i$ for all $i=1, \dotsc, m$. Rosenbaum and Wagner \cite{RosenbaumWagner} show how to reduce $p$-group isomorphism to Composition Series Isomorphism, and then how to reduce the resulting Composition Series Isomorphism Problem to Graph Isomorphism on graphs of degree at most $p + O(1)$; Rosenbaum more generally showed how to reduce \GpI to Composition Series Isomorphism in $n^{(1/2)\log n + o(\log n)}$ time. Luks \cite{LuksCompSeries} showed how to solve Composition Series Isomorphism in $\poly(n)$ time.
Recall the \emph{socle series} of a group $G$ is defined as follows: the \emph{socle} $\Soc(G)$ is the subgroup generated by all minimal normal subgroups. $\Soc(G)$ is always a direct product of simple groups. We then recursively define $\Soc^{i+1}(G)$ to be the preimage of $\Soc(G/\Soc^{i}(G))$ in $G$, that is, if $\pi_i\colon G \to G/\Soc^{i}(G)$ is the natural projection, then $\Soc^{i+1}(G) = \pi_i^{-1}(\Soc(G / \Soc^i(G)))$. The reduction is to pick a composition series for $G$ that is compatible with its socle series, and then to try all possible composition series for $H$ compatible with its socle series. One of the keys to their running time is to show that the number of composition series compatible with the socle series is bounded by $n^{(1/2) \log n}$. 

Within $O_\infty(G)$, we refine the socle series with our characteristic filter. Without loss of generality, we may assume that the restriction of our characteristic filter $\phi$ to the Fitting subgroup $O_\infty(G)$ refines the socle series of $O_\infty(G)$. If it doesn't originally, we may further refine it using the socle series, then iterate the main loop of Algorithm~\ref{algo:outline} until it stabilizes again. Our algorithm here is to reduce to Composition Series Isomorphism, but to only consider composition series that are compatible both with our filter $\phi$ \emph{and} with the socle series. If the filter has many small layers, this will cut down the number of composition series that need to be considered, thus reducing---for such groups---the dominant factor in the running time of \cite{RosenbaumWagner, RosenbaumBidirectional, RosenbaumSolvable}. 

To illustrate the potential savings, we define the \emph{width} of a filter $\phi$ with elementary abelian layers to be the maximum dimension of any layer:
\[
\width(\phi) := \max_s \dim_{p_s} (\phi_s / \partial \phi_s).
\]
Then we have:
\begin{thm} \label{thm:width}
Let $N$ be a solvable group of order $n$, and $\phi_N$ be a characteristic filter on $N$ computable in time $t(n)$. 
Then isomorphism of $N$ with any group can be tested, and an isomorphism found, in time 
$n^{(1/2)\max_P\width(\phi_P) + O(1)} + t(n)$.

In particular, using the characteristic filter $\phi_{g,k}$ output by Algorithm~\ref{algo:outline} with parameters $g$ and $k$, isomorphism of solvable groups can be solved in time 
\[
n^{(1/2) \width(\phi_{g,k}) + O(1)} + n^{O(kg)}.
\]
\end{thm}

\begin{proof}
The outline of the algorithm follows Rosenbaum--Wagner \cite{RosenbaumWagner}, also using Luks's polynomial-time algorithm for Composition Series Isomorphism \cite{LuksCompSeries}; the key difference here is that we only consider composition series which refine our characteristic filter $\phi$, rather than more general composition series as in Rosenbaum and Wagner. The runtime of their algorithm is a product of the running time to enumerate the desired composition series, and the time to solve Composition Series Isomorphism. Our improvement is in the first step. So we only need calculate the number of composition series of $N$ compatible with $\phi_{g,k}$. In our case, we must first compute the filter $\phi$, which takes time $t(n)$.

Let $M$ be a second solvable group.
Enumerating the composition series of $M$
compatible with $\phi_M$ can be achieved as follows. Go through $s \in \N^d$ in lexicographic order, starting with the lexicographically largest $s$ such that $\phi_s \neq 1$. Within each layer $L_s = \phi_s / \partial \phi_s$ we choose all possible composition series. By \cite[Lemma~3.1]{RosenbaumWagner}, this can be done in time $|L_s|^{(1/2) \log_{p_s} |L_s|} \leq |L_s|^{1/2 \width(\phi)}$. Taking the product over all layers, we get a bound of $|M|^{(1/2)\width(\phi)}$. For each such composition series, we then use Luks's $\poly(|M|)$-time algorithm for Composition Series Isomorphism, yielding the stated result.

For the ``in particular,'' we compute $\phi_{g,k}$ using Algorithm~\ref{algo:outline}, which takes $n^{O(gk)}$ time.
\end{proof}

\subsection{Intermediate version: choose composition series compatible with the filter and hypergraph}
The vertex coloring of the hypergraph $\mathcal{H}_\chi^{(g,k)}(\phi)$ may inform us of characteristic sub\emph{sets} that are not subgroups. Although the filter has been refined as much as possible (in particular, any one of the color classes of the hypergraph in a given layer $L_s$ must generate the whole layer), we can nonetheless take advantage of these characteristic subsets in the preceding algorithm, by further restricting the composition series that we need to consider. 

Towards this end, for each layer $L_s$ let $C_s$ denote the smallest color class in $L_s$, and define the \emph{color ratio} of a layer $L_s$ as $|L_s| / |C_s|$. Finally, define the color ratio of a solvable group $N$ as
\[
\colorratio(N) := \prod_{s \in \N^d} \colorratio(L_s) = \prod_s \frac{|L_s|}{|C_s|} = \frac{|N|}{\prod_s |C_s|}.
\]

We now restate (a slightly refined version) of Theorem~\ref{thm:width-and-color}:
\begin{thm-width-and-color}[Refined]
Let $N$ be a solvable group of order $n$. Let $\phi=\phi_{g,k}$ and $\mathcal{H}_\chi^{(g,k)}(\phi)$ 
be the filter and colored hypergraph for $N$  output by Algorithm~\ref{algo:outline} with parameters $g,k$. In each layer $L_s$, let $C_s$ denote the smallest color class. Then isomorphism of $N$ with any group can be tested, and an isomorphism found, in time
\begin{align*}
 & \left( \prod_{s \in \N^d} \min\{|L_s|^{1/2}, |C_s|\}\right)^{\width(\phi_{g,k})} \poly(n) + n^{O(gk)}\\
 \leq & \left(\frac{n}{\colorratio(N)} \right)^{\width(\phi_{g,k})} \poly(n) + n^{O(gk)}.
\end{align*}
\end{thm-width-and-color}

\begin{proof}[Proof of Theorem~\ref{thm:width-and-color}]
The outline of the algorithm is the same as in Theorem~\ref{thm:width}; the key difference is how we enumerate composition series within each layer $L_s$ (and how many we enumerate). To see how to take advantage of the size of the smallest color class $C_s \subseteq L_s$, we must recall the details of Rosenbaum \& Wagner's Lemma~3.1 \cite{RosenbaumWagner}, on enumerating composition series. In the algorithm of Theorem~\ref{thm:width} above we have already taken care of the ordering of the layers, so the only difference here will be on how we enumerate the part of the composition series within each layer $L_s$. That is, we may assume that we have already built a composition series of $\partial \phi_s$, which we now want to extend to a composition series of $\phi_s$. Since the subgroup generated by $C_s$ would be a characteristic subgroup of $L_s$, and $\phi$ has already been refined according to the coloring $\chi$ on $\mathcal{H}$, it must be the case that $C_s$ generates all of $L_s$. Thus we may select only those composition series where the generator of each step of the composition series comes from $C_s$. Since any generating set (and hence any composition series) for $L_s$ has size $\log_{p_s} |L_s| \leq \width(\phi)$, the number of choices of composition series where all the generators in the series are chosen from $C_s$ is bounded by 
\[
|C_s|(|C_s| - 1) (|C_s| - 2) \dotsb (|C_s| - \log_{p_s}|L_s| + 1) \leq |C_s|^{\width(\phi)}.
\]
This analysis already gives the second bound in the statement of the theorem. To get the more refined bound, within each layer $L_s$, if $|C_s| < |L_s|^{1/2}$, then we employ the above strategy, and otherwise we use the $|L_s|^{(1/2)\log_{p_s}|L_s|} \leq |L_s|^{(1/2)\width(\phi)}$ strategy from Rosenbaum--Wagner \cite{RosenbaumWagner}. 
\end{proof}

\subsection{Advanced version: refine the filter and hypergraph as you go (individualize and refine)} \label{sec:iso-IR}
Finally, we give a version of the individualize-and-refine paradigm from Graph Isomorphism as applied to composition series that are compatible with our filter and colored hypergraph. The algorithm is similar to that from the previous section, except now, each time we pick a subgroup in our composition series, we give a new color to the corresponding vertex in our hypergraph, and then we run more iterations of the main loop of Algorithm~\ref{algo:outline} until the filter and hypergraph again stabilize, before we pick the next subgroup in our composition series. This can potentially have the effect of reducing the width of the layers and/or the size of the smallest color class in each layer as we go.

In somewhat more detail: compute the filter $\phi$ and colored hypergraph $\mathcal{H}_{\chi}^{(g,k)}(\phi)$ as before. We build up a composition series in $G$ and simultaneously keep a list of partial composition series in $H$ that we want to test for isomorphism in the end. Suppose we are at the point where we already have built a composition series up to $\partial \phi_s$ in $G$, and we have a list $\mathcal{L}$ of composition series up to $\partial \phi_s$ in $H$. Then we extend the partial composition series of $G$ by picking an element of $C_s$ (the smallest color class in $L_s$). We then color the corresponding vertex in $\mathcal{H}$ a new color, and refine both $\mathcal{H}$ and $\phi$ until stabilization (as in the main loop of Algorithm~\ref{algo:outline}). Within $H$, we try each element of $C_s$ in turn, refining the filter and hypergraph for $H$. If for any $x \in C_s(H)$ the refinement does not agree with the refinement we got in $G$, we throw it away, otherwise we extend our composition series for $H$ by the subgroup generated by $x$ and $\partial \phi_s$, and add this new partial composition series to our list $\mathcal{L}$. This comes at a multiplicative cost of $|C_s|$. We then continue this process within the (potentially new, smaller $C_s$) until we get a composition series that now includes all of $\phi_s$. The total multiplicative cost within the layer $L_s$ is thus at most $|C_s|(|C_s| -1) \dotsb (|C_s| - \log_{p_s}|L_s| + 1) \leq |C_s|^{\width(\phi)}$, so this at most squares the total running time from Theorem~\ref{thm:width-and-color}. 

Thus, asymptotically, we get a similar worst-case upper bound. We could state a more refined upper bound along the lines of Theorem~\ref{thm:width-and-color}, but the definitions involved are somewhat delicate and recursive (because they depend on how the width and the color-ratio \emph{change} as the algorithm progresses). Nonetheless, in practice, we expect this individualize-and-refine technique to perform much better, as the layers and color class should decrease in size as the algorithm progresses.

%%%%%
\section{A random model of groups}
\label{sec:random}
% !TEX root = WL.tex
Inspired by a suggestion of A. Mann \cite{Mann1999}*{Question~8}
(answered in \cite{KTW}), we describe here a model for 
random finite groups.
We first give a simplified model that samples only random finite nilpotent groups.
Later we extend this to sample solvable, semisimple, and general
finite groups.

\subsection{A model for random finite nilpotent groups}
As a first approximation we 
choose $\ell$ random upper unitriangular 
$(d\times d)$-matrices $u_1,\ldots, u_{\ell}$,
over the integers modulo a fixed positive integer $b$.
The $u_i$ are drawn according to a fixed distribution 
$\mu(d,b)$.  Later we shall discuss the effect of $\mu$
on the group theory, but first we survey the possible outcomes.

An immediate observation is that $U=\langle u_1,\ldots,u_{\ell}\rangle$ is 
a subgroup of the full group of upper unitriangular matrices.  Therefore, $U$ 
is nilpotent of order at most $b^{d^2/2}$.  In particular, if $U_p$ denotes the  
Sylow $p$-subgroup of $U$, then
$U=\prod_{p|b} U_p$.  The choice of $\ell$ 
generators also has a fingerprint within the structure of our groups $U$.   
In particular by Burnside's Basis Theorem, for each $p|b$, 
\begin{align*}
	|U_p: [U_p,U_p]U_p^p | & \leq p^{\ell}.
\end{align*}
Thus, there is a certain amount of structure which is fixed by the choice of parameters
$(d,b,\ell)$. Nevertheless, the coverage asserted in Theorem~\ref{thm:coverage} 
shows the diversity of these groups.

\subsection{General model}
To sample a more general class of groups, we add
terms to the block-diagonal.  Sampling random invertible
square matrices will almost always generate the entire 
general linear group.  As noted in Section~\ref{sec:intro}, 
a more nuanced approach is called for.
\medskip

Our strategy is as follows:

\begin{enumerate}[(a)]
\item Add solvable groups by selecting any matrix
that is diagonalizable over the algebraic closure.  We call
this a {\em random toral subgroup}.

\item From the classification
of finite simple groups we can select at random, according
to a fixed distribution, a non-abelian finite simple group $T$
and let $T\leq S/Z(S)\leq \Aut(T)$---that is, choose $S$, a (possibly
trivial) central extension of an almost simple group.
Then we form the group algebra $A=(\mathbb{Z}/b)\langle S\rangle$.
We then sample from the minimal left ideals $I$ of
$A$.  This defines a linear representation $\rho:T\to \End(I)$ 
where $I$ is a $\mathbb{Z}/b$-module.  It is a straightforward exercise
to see that endomorphisms
of finite modules are representable as chequered matrices.
We copy the image of a generating set for $S$ into
chequered matrices, and then place this on the block diagonal.
We repeat until we exceed a bound on $d$.

\item Add permutation to the block diagonal to any two
terms with isomorphic representations.

\item As a final step we now sample block upper unitriangular matrices.
\end{enumerate}

It is important to proceed in this order to avoid redundant
choices.
The number of variability of the simple modules represented
on the block diagonal is again controlled by the distribution
and that can have substantial impact on the resulting group.

\begin{prop}
The class of groups sampled includes: A permutation group $P$, 
central extensions $S_a$ of almost simple groups,
$P\wr (S_1\times \cdots \times S_s)\ltimes U$ where 
$U$ is sampled as above and $L(U)$ is a $\prod_a T_a$-module.
\end{prop}
    
\subsection{Coverage: Proof of Theorem~\ref{thm:coverage}}
\label{subsec: coverage}
For (i) consider matrices of the form $u_{ij}=I+a_{ij}E_{ij}$.
If $1\leq i<d/2\leq j\leq d$ then all such $u_{ij}$ commute and
are independent.  So fix a divisor chain $e_1|\cdots |e_s|b$
and coefficients $a_{ij}$ (in some index order) having 
additive order $e_{im+j}$, it follows that these $u_{ij}$
generate an abelian group with the specified invariants.

For (ii-iii), let $R_v:U\to W$ where $R_v(u)=u*v$, represented
as a matrix.  Let $\bar{U}$ be representation of $U$ as
in (i), and likewise with $\bar{W}$.  Then
\begin{align*}
    Bh(*) & \cong \left\{\begin{bmatrix}
        1 & u & w \\ 0 & I_r & R_v \\ 0 & 0 & I_s
    \end{bmatrix} ~\middle|~ u\in \bar{U}, v\in V, w\in \bar{W}
    \right\}\\
    Br(*) & \cong \left\{\begin{bmatrix}
        1 & u & w \\ 0 & I_r & R_u \\ 0 & 0 & I_s
    \end{bmatrix} ~\middle|~ u\in \bar{U}, w\in \bar{W}
    \right\}
\end{align*}

For the count note that it suffices to count the number of distinct
bilinear maps $*:U\times V\bmto W$.  As Higman demonstrates \cite{Hig60}, there 
are 
$p^{\dim U\dim V\dim W}/|\GL(U)\times \GL(V)\times \GL(W)|\in p^{\Theta(n^3)}$
such maps.

Finally, for a given list of groups sampled in smaller
dimensions, form block diagonal representations. This 
affords the direct product of the list.  For subdirect
products take a subgroup of the block-diagonal group.
This completes the proof of Theorem~\ref{thm:coverage}.

\subsection{The importance of the right distribution}
If we sample dense matrices when we shall call the result
the \emph{dense
random subgroup model} for  the general linear group 
$\GL(d, \mathbb{F}_p)$.  While this is an easy model to 
reason about it is also fairly rigid, as the following 
result illustrates.

\begin{thm}\label{thm:dense-random}
    If $u_1,\ldots,u_{\ell}$ are chosen uniformly at random
    from the group of upper uni-triangular matrices $U_d(\mathbb{Z}/b)$
    and $\ell\in \Omega(\sqrt{d})$ then 
    \begin{align*}
        \Pr\left(|\langle u_1,\ldots,u_{\ell}\rangle|
            =b^{\ell+\binom{d-2}{2}}\right) & \to 1.
    \end{align*}
\end{thm}

In fact we shall prove the following stronger claim: 
with high probability, such groups $\langle u_1,\ldots,u_{\ell}\rangle$
contain the group of commutators of $U_d(\mathbb{Z}/b)$.

In this model, groups can range widely in isomorphism types, 
one does not see much variability
in coarse isomorphism invariants such as group order, 
numbers of subgroups or quotients, conjugacy classes, and so forth.

Our proof of Theorem~\ref{thm:dense-random} relies on some details of Sims' 
proof on the asymptotic upper bound on the number of isomorphism types of 
$p$-groups \cite{Sim65}.
It begins as follows. For a group $G$ let $\gamma_i(G)$ be the 
$i$th term in the 
lower central series.
Every $p$-group $G$ has a subgroup $H\leq G$ such that $\gamma_2(H)\gamma_3(G)=
\gamma_2(G)$
and where $d(H)$ the least number of group elements to generate 
$H$, is minimal with that property 
\cite{BNV:enum}*{Proposition~3.8}.  
We call $d(H)$ the
the {\em Sims' rank} of $G$.

\begin{defn}
A {\em Sims subgroup} of a nilpotent group $G$ is a subgroup $H\leq G$
minimal with respect to $\gamma_2(H)\gamma_3(G)=\gamma_2(G)$.  The
{\em Sims rank} of $G$ is the minimum number of generators needed to
generate a Sims subgroup.
\end{defn}

Fix $G=U(d,k)$, $V=G/\gamma_2(G)\cong \F^{d-1}$, 
$W=\gamma_2(G)/\gamma_3(G)\cong \F^{d-2}$.  Then there is a bimap 
$*:V\times V\rightarrowtail W$ given by commutation:
\begin{align}
	[(\gamma_2(G)x),(\gamma_2(G)y)] & \equiv [x,y] \pmod{\gamma_3(G)}.
\end{align}
Fix a subgroup $H$ of $G$, and put $U=H\gamma_2(G)/\gamma_2(G)$.  Observe that 
$H$ is a Sims subgroup if, and only if, $[U,U]=[V,V]$.  Also observe 
that after taking natural bases for $\F^{d-1}$ and 
$\F^{d-2}$, the bimap $*$ can 
be represented as follows. % is the following.
Let $[,]:\F^{d-1}\times \F^{d-1}\rightarrowtail \F^{d-2}$ be defined, 
in a parametrized form, by $[u,v]=uBv^t$
where
\begin{align}\label{def:comm-U}
	B & = \begin{bmatrix}
		0 & f_1 & \\
		-f_1 & 0 & -f_2 & \\
		& -f_2 & \ddots & \ddots &\\
		& & \ddots & & f_{d-2}\\
		& & & -f_{d-2} & 0
	\end{bmatrix}. 
    %M_{d-1}(\F^{d-2}).
\end{align}
That is, $B$ could be understood as a $3$-tensor of size 
$(d-1)\times(d-1)\times(d-2)$, whose $i$th frontal slice is given according to 
$f_i$.

\begin{proof}[Proof of Theorem~\ref{thm:dense-random}]
Our approach is to show that a subgroup generated by enough elements is a Sims subgroup.
To do this it suffices to show that for most sufficiently large dimensions,
the bilinear map of \eqref{def:comm-U} has the property that most $X\leq \F^{d-1}$ 
satisfy $[X,X]=\F^{d-2}$. For notation we let $V=\F^{d-1}$
with basis $\{e_1,\ldots,e_{d-1}\}$ and $W=\F^{d-2}$ with basis 
$\{f_1=[e_1,e_2],\ldots,f_{d-2}=[e_{d-2},e_{d-1}]\}$.

If $X\leq V$ is the row span of the full rank $(s\times (d-1))$-matrix $M$ then
\begin{align*}
	(MBM^{\dagger})_{ij} & = 
	\sum_{k=1}^{d-2} (M_{ik}M_{j(k+1)}-M_{i(k+1)}M_{jk})f_k.
\end{align*}
This defines a natural $3$-tensor of size $s\times s\times (d-2)$ by
\begin{align*}
	m_{i,j,k} & = (M_{ik}M_{j(k+1)}-M_{i(k+1)}M_{jk}).
\end{align*}
Notice $[X,X]=W$ if, and only if,  $\langle \sum_k m_{ijk} f_k | 1\leq i,j\leq 
s\rangle=W$.
That is, if we flatten the tensor into a  $(s^2\times (d-2))$-matrix $\tilde{m}$,
as follows,
\begin{align*}
	\tilde{m}_{(s\cdot(i-1)+j),k} & = (M_{ik}M_{j(k+1)}-M_{i(k+1)}M_{jk})
		=\det\begin{bmatrix} M_{ik} & M_{i(k+1)}\\ M_{jk} & M_{j(k+1)}\end{bmatrix};
\end{align*}
then we are asking that $\tilde{m}$ is of full rank.  Now we argue that for 
$s\geq 2\sqrt{d}$ this is the expected behavior.

By our model, each entry in $M$ is drawn independently at random.  However 
the entries of $m$ (and therefore $\tilde{m}$) are dependent. 
Nevertheless, we can observe that the values of $m_{ijk}$ are almost independent
of $k$.  Certianly $m_{ijk}$ is independent of $m_{ijk'}$ if $|k-k'|>1$.  
Also, if $k'=k+1$, if $m_{ijk}\neq 0$ then nothing can be said about $m_{ij(k+1)}$.
Even if $m_{ijk}=0$ it may be impossible to predict $m_{ij(k+1)}$, the
exception is when $M_{i(k+1)}=0=M_{j(k+1)}$.  So there is $ 1/q^2$ chance of dependence
between with the exception of pairs of $0$.  Each such dependency will be compensated for
by adding a row $j'$ such that $M_{j'(k+1)}\neq 0$.  Thus $m_{ij'k}$ will be independent
of $m_{ij'(k+1)}$. Since the selection of a 
nonzero entry is a $1-1/q\geq 1/2$ event
the addition of one row is highly likely to break dependence.  Thus, at a cost of
sampling $O(\sqrt{d})$ rows we obtain with high probability a matrix $M$ whose associated
matrix $\tilde{m}$ is full rank.
\end{proof}

\subsection{Sparsity}

For added variation a different distribution is required,
one which favors sparse matrices. 
Fix positive integers $b$ and $d$.  Let $wt(u)$ be the
number of non-zero values in the upper unitriangular $u$.  Fix a 
distribution $\mu$
on $\mathbb{Z}/b-\{0\}$ and a distribution $\nu$ on 
$\{1,\ldots,\binom{n-1}{2}\}$. Define a 
\emph{$(b,d,\mu,\nu)$-random triangular matrix} as an 
$\alpha:\binom{d}{2}\to \mathbb{Z}/b$ sampled 
according to a distribution $|\mathrm{supp}~\alpha|=k$ 
with probability $\nu(k)$ and
for each $\{i,j\}\in \mathrm{supp}~\alpha$, $\alpha_{i,j}$ is sampled according to 
$\mu$.  Notice $\alpha$ uniquely determines 
an upper unitriangular matrix:
\begin{align}
    u(\alpha) & = I_{d} + \sum_{i=1}^{d}\sum_{j=i+1}^{d}
        \alpha_{ij} E_{ij}.
\end{align}
The distribution $\nu$ describes how large the support
of $\alpha$ is expected to be, and $\mu$  describes what
non-zero values in $\mathbb{Z}/b$ will be used as entries.

Finally define a \emph{$(\mu,\nu)$-random
unitriangular group} as the group generated by independently
sampling $\ell$ upper unitriangular $(d\times d)$-matrices 
over $\mathbb{Z}/b$ according to their $(\mu,\nu)$-distribution.
The precise outcomes of this distribution appear intricate.
Through some empirical testing (e.g. Figure~\ref{fig:scatter-plot}) we have produced the following
question:

\begin{quote}
 If $\nu(|A|)\to 0$ for $|A|>C$, does
$\log |\langle u_1,\ldots,u_{\ell}\rangle|$ approach
a discrete Gaussian distribution on $\{1,\ldots,\binom{d}{2}\}$?
\end{quote}

Our model makes several constraining choices in order that it avoids the
analysis that would otherwise create rather similar groups.  The cost of
this is that we can so far only offer heuristic explanations for the behavior.
Even so, we explain what we understand and encourage a thorough exploration
in the future.

The first question is what to expect the length of the block diagonal to be in $U$.
Suppose we assume that the block diagional is chosen uniformly at a partition of $d$.
From Vershik's theorem~\cite{Vershik}, the shape of the tableaux of random 
partition of $d$ with at least $\sqrt{d}$
terms tends to $O(e^{-t})$.  That implies that there are relatively few large blocks
as those are in the tail of the random distribution. Thus there would be many blocks
of small size.  This however requires one justify that sampling $U$ at random samples
partitions of $d$ uniformly at random.  That need not be the case.  So we ask
\begin{quote}
Is the typical sparsely sample group
$\langle u_1,\ldots,u_{\ell}\rangle$ convex (tending toward the middle) or concave (tending away from the middle)?
\end{quote}

The answer to this speaks to the expected nilpotence class of the groups $U$.
The length of this block diagonal is a bound on the nilpotence class.
For example, if there are just two blocks, then 
\begin{align*}
	U& \leq\left\{\begin{bmatrix} I_a & * \\ 0 & I_b\end{bmatrix}\right\}
\end{align*}
implies that $U$ is abelian.  In general, 
if $\mathcal{F}$ denotes the subspace flag determining the block structure of $U$,
then the nilpotence class of $U$
is at most $|\mathcal{F}|-1$. 

To see a reason that 
sparse matrices should sample a wider class of groups than dense matrices we consider a sufficient condition to \emph{avoid} being
a Sims subgroup.

\begin{lemma}
Fix an alternating bimap $[,]:V\times V\bmto W$ with $W=[V,V]$.
Let $\pi_1,\ldots,\pi_{d-2}$ be a basis of $W^*$ and define $(u,v)_i=\pi_i[u,v]$.
For $X\leq V$, if there exists an $i$ such that $(X|X)_i=0$, then $[X,X]\neq W$.
\end{lemma}
\begin{proof}
If $(X|X)_i=0$ then for $u\in W$ with $\pi_i(u)=1$, $u\notin [X,X]$.
\end{proof}

Now here is the situation.  The maps $(|)_i:V\times V\bmto K$ are alternating bilinear
forms, possibly degenerate.  The subspaces $X\leq V$ with $(X|X)_i$ are what are known
as \emph{totally isotropic}.  The number of maximal totally isotropic subspaces of 
$V$ is $q^{O(m^2)}$ where $m=\dim V-\dim \{v: (v|V)=0\}$.  Therefore the smaller the radical
the much large the number of totally isotropic subspaces there are and therefore the
less likely that a subspace $X$ generates $W$.  So as we move towards bimaps for 
unipotent hulls for flags of fixed length at least 3, then the commutator involved
will have quotients to alternating forms with large numbers of totally isotropic
subspaces.  Thus more subspaces will fail to generate $W$.  As result, fewer subgroups
will be Sims subgroups.  This however is only a crude guide to the number of 
Sims subgroups and we encourage an actual analysis with better insights.

\subsection{WL-refinement in our random model}

So now let us consider the effects of refinement in our random model. 
Our proof is in two parts.  Either our unipotent groups $U$ have long block
diagonal series or it has bounded class.  In the former case we reduce the
refinement analysis to a result of Maglione  \cite{Maglione:filters}.  In
the later case we appeal to classical results on nonsingular products.
In either case we discover refinements.  We aim to prove Theorem~\ref{thm:random-refine}.

\paragraph{Refinements for many blocks.}
First let us consider groups with many blocks.

\begin{thm}\label{thm:longchain-refine}
The refinement length of a random subgroup $U\leq U(d,p)$ is on average at $\Omega(\ell^2)$
where $\ell$ is the length of is generalized eigen $1$-space flag.  
\end{thm}

Primarily we want to appeal to the following. Note that in this case 
we do not apply the Weisfeiler--Leman procedure developed in this paper; instead, 
it will be used 
in the next setting. 

\begin{thm}[Maglione \cite{Maglione:filters-U}]
The group $U(d,p)$ has an (adjoint) characteristic filter refinement of length $\Theta(d^2)$.
\end{thm}

However we do not have the group $U(d,p)$. Instead, we have a subgroup sampled at 
random
either with dense or sparse matrices.  First we dispense with the dense case.

\begin{cor}\label{cor:Malgione-chain}
A subgroup $H\leq U(d,p)$ generated by dense matrices $u_1,\ldots,u_{\ell}$ with 
$\ell\geq 2\sqrt{d}$ has on average a characteristic filter refinement of length $\Theta(d^2)$.
\end{cor}
\begin{proof}
By Theorem~\ref{thm:dense-random}, $H$ is almost certainly a Sims subgroups of
$U(d,p)$ and therefore $[H,H]=[U(d,p),U(d,p)]$.  As a scholium to Maglione's
theorem we observe that the adjoint filter refinement of $U(d,p)$ can be defined
as refinement through terms $L_s$ for $s<(2,0,\ldots,0)$ in the filter.  As a
result these same terms appear in the filter of $H$ and so $H$ refines to a length
of $\Omega(d^2-d)=\Omega(d^2)$.  Since $\log_p |H|\in O(d^2)$ the result follows.
\end{proof}

Next we need to consider the sparse case as this is where our model presides.  What
we do is demonstrate a form of \emph{Morita condensation theory}
that transports our
sparse problem into a dense problem~\cite{Wilson:Skolem-Noether}.
What we observe is that each right-hand edge $j$ of block on the block diagonal of $U$ is 
defined by the presence of an element $u\in U$ with a non-zero value $u_{ij}$, 
otherwise the block would be wider.  We select one such row $i_s$ for each block 
$s$,
and one such column $j_s$.  Thus out of the $(\ell\times \ell)$-block matrix $u\in U$, 
we create an $(\ell\times \ell)$-matrix  by copying the entire in $u_{i_s j_s}$.
For example
\begin{align*}
	\left[
	\begin{array}{cc|c|ccc}
	1 & 0 & a_{13} & a_{14} & a_{15} & a_{16}\\
	0  & 1 & a_{23} & a_{24} & a_{25} & a_{26}\\
	\hline
	  &   & 1      & a_{34} & a_{35} & a_{36}\\
	 \hline
	  &   &        & 1      & 0 & 0\\
	  &   &        & 0     & 1 & 0\\
	  &   &        & 0      & 0 & 1\\	  
	\end{array}\right]
	\mapsto 
	\begin{bmatrix}
	1 & a_{13} & a_{15}\\
	& 1 & a_{35} \\
	&  & 1
	\end{bmatrix}.
\end{align*}

This may seem a bit unnatural but in fact it is applying a functorial property
not on the level of groups but on the level of the enveloping algebra of the matrices
and more importantly on the level of bilinear maps.
While this function has no relationship in the context of groups, it is by 
considering the
associated ring context that we see that we have simply performed a condensation 
of modules, that is we have changed to an equivalent category.
So for each of the blocks $B_1,\ldots B_d$ we
let $e_s$ be the $(d_s\times d_s)$ matrix with zero in every position except $jj$.
Set $e=e_1\oplus\cdots \oplus e_{\ell}$.  Then $eue$ is matrix with at most $\ell\times \ell$
nonzero entries.  Removing the all zero rows and columns produces an $(\ell\times \ell)$-matrix.
In the example above,
\begin{align*}
	e & =
	\left[
	\begin{array}{cc|c|ccc}
	1 & 0 &  &  &  & \\
	0  & 0 &  &  &  & \\
	\hline
	  &   & 1      &  &  & \\
	 \hline
	  &   &        & 0      & 0 & 0\\
	  &   &        & 0     & 1 & 0\\
	  &   &        & 0      & 0 & 0\\	  
	\end{array}\right]	
\end{align*}
In particular this induces a functorial Morita condensation of each bilinear map 
$L_s\times L_t\bmto L_{s+t}$, see \cite{Wilson:Skolem-Noether}.  We therefore denote
this group $\overline{eUe}$ to remind us of the natural process to create this 
smaller matrix group.

Having applied this transform, notice $\overline{eUe}$ is now a dense subgroup of $U(\ell,p)$.
Therefore we arrive at the following.

\begin{proof}[Proof of Theorem~\ref{thm:longchain-refine}]
Suppose $U\leq U(d,p)$ generate by random matrices $u_1,\ldots,u_t$.  If
the $u_i$ are dense then by Corollary~\ref{cor:Malgione-chain} there is a computable
filter refinement of length $O(\ell^2)$ where $\ell$ is the number of blocks of $U$.
If on the other hand the $u_i$ are sparse, then $\overline{eUe}$ has a refinement 
of 
length $\Omega(\ell^2)$.  As the map $U\mapsto \overline{eUe}$ is functorial in the 
bilinear maps used to select refinement, it follows that $U$ also has a refinement
of length $\Omega(\ell^2)$.
\end{proof}

\paragraph{Refinements for few blocks.}
The last case to concern us is when $U$ has a bounded number of blocks on the diagonal, 
but that the number of blocks is at least 3. (Otherwise the group $U$ is abelian which
is the first case of Theorem~\ref{thm:random-refine}.)    Because the number
of blocks is bounded at least one block has dimension proportional to $d$ as $d\to \infty$.

Let us consider coloring with $g=1$. This means that with a selected layer
$L_s\times L_t\bmto L_{s+t}$ we consider labels on $1$-dimensional subspaces
$\langle x\rangle\leq L_s$ by labeling the restriction $\langle x\rangle\times L_t\bmto L_{s+t}$.
One observes this structure is nothing more than a linear transformation $L_t\to L_{s+t}$
and is thus defined up to change of basis solely by the rank of the transformation.
Therefore to each element of ${\rm PG}_0(L_{s})$ we record the rank of the 
associated matrix.
We do likewise with ${\rm PG}(L_t)$.  Finally we label the edges between ${\rm 
PG}(L_{s})$ and
${\rm PG}(L_t)$ by whether or not the pair of points commutes.  

In order to model this behavior in colors we make the following 
observation.  Treating
$x=(x_1:\cdots:x_{d})$ as homogeneous point in $d$ variables, the evaluation
$[x,-]:L_t\to L_{s+t}$ produces a matrix $M(x)$ with entries in $\F[x_1,\ldots,x_d]$.
The rank of this matrix changes as we evaluate $x$ but certainly there are two natural
states: either $M(x)$ has rank at most $r$ or it does not.  If $M(x)$ has rank
at most $r$ then all $(r\times r)$-minors must vanish, and this produces a polynomial
number degree $r$-polynomials that must all vanish on $x$.  That is to say, the 
condition of the rank of $M(x)$ is a variety (or more generally a scheme).  It is
in fact a \emph{determinantal variety} and the subject of considerable study in the
algebraic community as well as the computer science community 
\citelist{\cite{Harris92}\cite{Minrank}}.  It is important to observe 
that many results in the field are only known over algebraically closed fields.
However it is known that these varieties are reduced and irreducible \cite{Harris92}.
Therefore to count points we can use Lang--Weil theorem \cite{LW54}, but that 
requires
that we allow for a large field.  So this portion of our estimate assume $b\to\infty$
and $d\to\infty$. 

Let us assume for now that $M(x)$ has points, i.e. that for some $x\in {\rm 
PG}_0(L_s)$,
$[x,-]$ does not have full rank, and for other points it does.  Thus our vertex set
has (at least) 2 colors, say white if $[x,-]$ has full rank and black otherwise.  We do the same 
for ${\rm PG}(L_t)$. 
Recall that we are including a hyperedge $(x,y,[x,y])$ only if $[x,y]\neq 0$.

Now consider the situation. The number of black points is in general a solution to
a system of random nonlinear homogeneous polynomials of degree $r$.  That this is nonlinear
means we can expect that the number of black points is not a subspace.  Now the points
in ${\rm PG}(L_t)$ not connected to black points $x$ are the points $y\in 
x^{\bot}:=\ker[x,-]$.
In particular we have a nonlinear set parameterizing a subspace arrangement within 
${\rm PG}(L_t)$.
If we write the generator matrix of each subspace $\ker[x,-]$ it will be the dual
of a linear combination of the matrices used to define $[,]$, which we sampled at random.
Therefore we have a random subspace arrangement.  In general this incidence relation
is not equitable, so proper refinements will be discovered in the WL-refinement 
process.

With that we have proved the following.

\begin{thm}\label{thm:shortchain-refine}
If $U\leq U(d,p)$ has a bounded number of blocks and $d,p$ are large, then 
there exists a proper refinement of the standard filter.
\end{thm}

To remove the assumption that $d, p$ are large here, the following 
interesting question needs to be addressed. 
\begin{quote}
Let $A_1, \dots, A_m$ be random $n\times n$ matrices over $\F_q$. What is the 
typical number of non-full-rank matrices in the linear span of $A_i$'s?
\end{quote}

\subsection{Proof of Theorem~\ref{thm:random-refine}}
Let us suppose $G$ is sampled according to our model.  Let $U$ be the intersection of $G$ 
with $U(d,p)$ and begin with the initial filter of our introduction.  Then if $U$ is abelian
we are in case (i).  Otherwise $U$ has at least $3$ blocks so we can use either
Theorem~\ref{thm:longchain-refine} for the case of large blocks, or Theorem~\ref{thm:shortchain-refine}
in the case of small blocks.  In either case we obtain a proper 
refinement.  Note that 
after refinement of the bounded number of blocks several times we cross over to the large
number of blocks and so the result follows.
\qed

%%%%%
\section{Testing pseudo-isometry of alternating bilinear maps}
\label{sec:LiQiao}
% !TEX root = WL.tex

\subsection{The simplified main algorithm}
\label{subsec:average-case-algo}
In this subsection we formally describe the simplified main algorithm presented in Section~\ref{subsec: 
average-case in intro}, that is Algorithm~\ref{algo:first-average}. We also 
discuss some important adjustments used in the implementation. 
We need the following observation, which follows easily by computing the 
closure of the given generating set.
\begin{obs}\label{obs:group_order}
Let $C_1, \dots, C_t\in \GL(n, q)$, and let $G$ be the group generated by $C_i$'s. 
Let $s\in\N$. Then there exists an algorithm that 
either reports that $|G|>s$, or lists all elements in $G$,  in time $\poly(s, n, 
\log q)$.
\end{obs}
\begin{algorithm}
\begin{description}
\item[{\it Input:}] $\bG=(G_1, \dots, G_m)\in\Lambda(n,q)^m$, $\bH=(H_1, \dots, H_m)\in\Lambda(n,q)^m$, $c, s\in \N$, 
and $q$ is odd.  
%\%~Denote by $\cH=\langle  H_1, \dots, H_m\rangle$ the linear span of matrices in $\bH$.

\item[{\it Output:}] Either (1) $|\Aut(\bA)|>s$, where $\bA=(G_1, \dots, G_c)$, or (2) $\pseudo(\bG, \bH)$.
\item[{\it Algorithm procedure:}]
{\ }
\begin{enumerate}
\item Set $L\gets \{\}$. Set $\bA=(G_1, \dots, G_c)$, the first $c$ matrices from $\bG$. 
\item Use Theorem~\ref{thm:autometry-algorithms} (2) to compute a generating set 
for 
$\Aut(\bA)$. 
\item Use Observation~\ref{obs:group_order} with input $s$ and the generating set 
of $\Aut(\bA)$.

(If $|\Aut(\bA)|>s$, we terminate the algorithm and report that ``$|\Aut(\bA)|>s$.'')

%Otherwise,
\item Put $\cH=\langle  \bH\rangle$, the linear span of $\bH$;
for every $\bB=(B_1, \dots, B_c)\in \cH^c$, do the following. 
\begin{enumerate}
\item[a.] Use Theorem~\ref{thm:isometry-algorithms} (2) to decide whether 
$\bA$ 
and 
$\bB$ are isometric. 
\item[b.] If not, go to the next $\bB$. Otherwise, we get the non-empty coset 
$\Isom(\bA, \bB)$.
\item[c.] For every $T\in \Isom(\bA, \bB)$, do the following.
\begin{enumerate}
\item[\ ] Test whether the linear spans of $T^t\bG T$ and $\bH$ are the same. If not, 
go to the next $T$. If so, add $T$ into $L$.
\end{enumerate}
\end{enumerate}
\item Output $L$. 
\end{enumerate}
\end{description}
\caption{The first average-case algorithm for alternating space isometry.}
\label{algo:first-average}
\end{algorithm}

Let us first examine the running time of Algorithm~\ref{algo:first-average}.
\begin{prop}\label{prop:first-average-time}
Algorithm~\ref{algo:first-average} runs in time 
$\poly(q^{cm}, s, n)$. 
\end{prop}
\begin{proof}
If Algorithm~\ref{algo:first-average} outputs $|\Aut(\bA)|>s$, then its running 
time is determined by Theorem~\ref{thm:autometry-algorithms} (2) and 
Observation~\ref{obs:group_order}, which together require $\poly(s, n, \log q)$.

If $|\Aut(\bA)|\leq s$, we analyze the two For-loops at Step 4 and Step 4.c, 
respectively. The first loop adds a multiplicative factor of 
$q^{cm}$, since enumerating a single element in $\cH$ costs $q^m$. The second loop 
adds a multiplicative factor of $s$, due to the fact that $|\Isom(\bA, 
\bB)|=|\Aut(\bA)|\leq s$, as $\Isom(\bA, \bB)$ is a coset of $\Aut(\bA)$. 
Other steps can be carried out in time $\poly(n, \log q)$. Therefore the overall running time is upper bounded by $\poly(q^{cm}, s, n)$. 
% Furthermore, the number of $T$'s to be added to $L$ is upper bounded by $q^{cm}\cdot s$.
\end{proof}

We then prove the correctness of 
Algorithm~\ref{algo:first-average}, in the case 
that it does not report $|\Aut(\bA)|>s$.
\begin{prop}\label{prop:first-average-correctness}
If Algorithm~\ref{algo:first-average} does not report $|\Aut(\bA)|>s$, then it lists the 
set of pseudo-isometries (possibly empty) between $\bG$ and 
$\bH$. In particular, $|\pseudo(\bG, \bH)|\leq q^{cm}\cdot s$.
\end{prop}
\begin{proof}
By Step 4.c, every $T$ added to $L$ is a pseudo-isometry. We are left to show 
that $L$ contains all the 
pseudo-isometries. For this, take any pseudo-isometry $T$. Since the linear spans 
of 
$T^t\bG T$ and $\bH$ are the same, we know $T^t\bA T$ is equal to some 
$\bB\in\cH^c$. So when enumerating this $\bB$ in Step 4, $T$ will pass all the 
tests in the following, and then be added to $L$. This concludes the proof. 
\end{proof}

It remains to specify the choices of $c$ and $s$ in 
Algorithm~\ref{algo:first-average} in the average-case analysis. This is stated in 
the following, 
whose proof 
will be deferred to Section~\ref{subsec:average-case-analysis}.
\begin{prop}\label{prop:first-average-analysis}
Let $c:=20$. For all but at most 
$1/q^{\Omega(n)}$ fraction of $\bA=(G_1, \dots, 
G_c)\in\Lambda(n, q)^c$, we have $|\Aut(\bA)|\leq s:=q^{n}$.
\end{prop}

Combining 
Propositions~\ref{prop:first-average-time},~\ref{prop:first-average-correctness} 
and~\ref{prop:first-average-analysis}, we have the following theorem.
\begin{thm}
Let $m\geq 20$, and let $\F_q$ be a finite field of 
odd size. 
For all but at most $1/q^{\Omega(n)}$ fraction of $\bG=(G_1, \dots, 
G_m)\in\Lambda(n, q)^m$, Algorithm~\ref{algo:first-average} tests the isometry of 
$\bG$ with an arbitrary $\bH\in\Lambda(n, q)^m$ in time $q^{O(n+m)}$.
\end{thm}

\noindent{\bf Implementation details.}~
We now explain some issues in the implementation of Algorithm~\ref{algo:first-average}.

To make this algorithm suitable for practical purposes, recall that the 
algorithm's running time is dominated by the two For-loops which give 
multiplicative factors of $q^{cm}$ and $s$, respectively. For the average-case 
analysis we used $c=20$, but having this standing on the exponent is too 
expensive. In practice, actually using $c=3$ already imposes a severe restriction 
on $s$, the order of $\Aut(\bA)$. So we use $c=3$ in the implementation which 
gives a reasonable performance. 

But having $q^{3m}$ in the For-loop is still too demanding. Indeed, in practice 
the 
tolerable enumeration is around $5^{10}$, namely $q=5$ and $10$ on the exponent. 
So 
with $c=3$, the range of $m$ is still severely limited. (Interestingly, the 
algorithm seems to have a better dependence on $n$.) It is most desirable if we 
could let $c=1$, namely simply $q^m$. 

To achieve that we use the following heuristic. Note that if $G_1, \dots, G_c$ are 
low-rank matrices, then we will only need to match them with the low-rank matrices 
from $\cH$. Our experiement shows that, for a random $\cG$ over $\F_q$ 
when $q$ is a small constant, the number of low-rank (i.e. 
non-full-rank) matrices in $\cG$ is 
expected to be small (i.e. much 
smaller than $q^m$) and 
non-zero (i.e. no less than $3$) at the same time. 
So we can 
use 
$q^m\cdot \poly(n, \log q)$ to choose $3$ low-rank matrices from $\cG$. Then use 
$q^m\cdot \poly(n, \log q)$ to compute the set of low-rank matrices from $\cH$, 
denoted as $\cH_c$. We can then replace enumerating $\cH^c$ with $\cH_l^c$, which 
in general is much smaller.
\medskip

\subsection{The main algorithm}\label{subsec:average-case-algo2}
To state our algorithm, we need the concept of adjoint algebra. 
For two tuples of alternating matrices $\bG,\bH\in\Lambda(n,\F)^m$, the adjoint algebra of $\bG$ is defined as
\[
\Adj(\bG)=\{(A, D)\in \Mat(n, \F)\oplus \Mat(n, \F) : A\bG=\bG D\}, 
\]
%Given $\bG, \bH\in\Lambda(n, \F)^m$, 
and the adjoint space from $\bG$ to $\bH$ is 
\[
\Adj(\bG, \bH) 
= \{(A, D)\in \Mat(n, \F)\oplus \Mat(n, \F) : A\bG=\bH D\}. 
\]
Clearly, if 
$T\in\Aut(\bG)$, then $(T^t, T^{-1})\in \Adj(\bG)$. Furthermore, if $\bG$ and $\bH$ 
are isometric, then $|\Adj(\bG, \bH)|=|\Adj(\bG)|$.

We now introduce the algorithm (see Algorithm~\ref{algo:second-average}) that 
supports Theorem~\ref{thm:average}. We point out that Algorithm~\ref{algo:second-average} differs from the
algorithm presented in Section~\ref{subsec: average-case in intro} in two places. 
\begin{enumerate}
\item The first and major difference is to replace the uses of $\Aut(\bG)$ and 
$\Isom(\bG, \bH)$ with $\Adj(\bG)$ and $\Adj(\bG, \bH)$, thereby avoiding using 
Theorem~\ref{thm:isometry-algorithms} (2) and 
Theorem~\ref{thm:autometry-algorithms} (2). Since $\Adj(\bG)$ and $\Adj(\bG, \bH)$ 
are easy to compute over any field, this resolves the characteristic-$2$ field 
issue. Furthermore, $\Adj(\bG)$ and $\Adj(\bG, \bH)$ are also easier to analyze. 
But $\Adj(\bG)$ and $\Adj(\bG, \bH)$ could be larger than $\Aut(\bG)$ and 
$\Isom(\bG, \bH)$, so they are less useful from the practical viewpoint. 
\item The second place is step 2 in 
Algorithm~\ref{algo:second-average}: instead of just using the first $c$ matrices 
as in the algorithm presented in Section~\ref{subsec: average-case in intro}, Algorithm~\ref{algo:second-average} 
slices the $m$ matrices of $\bG$ into $\lfloor m/c\rfloor$ segments of $c$-tuples 
of 
matrices, and tries each segment until it finds one segment with a small 
adjoint algebra. This step helps in improving the average-case analysis, and can 
be applied to the algorithm presented in Section~\ref{subsec: average-case in intro} as well. 
\end{enumerate}

\begin{algorithm}
\begin{description}
\item[{\it Input:}] $\bG=(G_1, \dots, G_m)\in\Lambda(n, q)^m$, $\bH=(H_1, \dots, 
H_m)\in\Lambda(n, q)^m$ and $c, s\in \N$. 
%\\ \%~Denote $\cH=\langle H_1, \dots, H_m\rangle$ as the linear span of matrices in $\bH$.
\item[{\it Output:}] Either (1) $|\Aut(\bA)|>s$; or (2) $\pseudo(\bG, \bH)$ as a 
set, which may be empty.
\item[{\it Algorithm procedure:}]

\begin{enumerate}
\item Set $L\gets \{\}$. Set $F\gets\mathtt{false}$.
\item For $i=1, \dots, \lfloor m/c\rfloor$, do the following.
\begin{enumerate}
\item[a.] Set $\bA=(G_{c(i-1)+1}, \dots, G_{ci})$. 
\item[b.] Compute a linear basis of $\Adj(\bA)\subseteq \Mat(n, q)\oplus \Mat(n, q)$.
\item[c.] If $|\Adj(\bA)|\leq s$, set $F$ to be $\mathtt{true}$, and break 
the For-loop. 
\end{enumerate}
\item If $F=\mathtt{false}$, return ``$\bG$ does not satisfy the generic 
condition.'' and terminate.

Otherwise,
\item Put $\cH=\langle  \bH\rangle$, the linear span of $\bH$;
for every $\bB=(B_1, \dots, B_c)\in \cH^c$, do the following. 
%For every $\bB=(B_1, \dots, B_c)\in \cH^c$, do the following. 
\begin{enumerate}
\item[a.] Compute a linear basis for $\Adj(\bA, \bB)\subseteq \Mat(n, q)\oplus \Mat(n, 
q)$. 
\item[b.] If $|\Adj(\bA, \bB)|>s$, go to the next $\bB$. 
\item[c.] For every $(T, S)\in \Adj(\bA, \bB)$, do the following.
\begin{enumerate}
\item[\ ] If $S$ and $T$ are invertible and $S=T^{-t}$, test whether the linear spans 
of $T\bG T^t$ and $\bH$ are the 
same. If not, go to the next $(T, S)$. If so, add $T^t$ into $L$.
\end{enumerate}
\end{enumerate}
\item Output $L$. 
\end{enumerate}
\end{description}
\caption{The second average-case algorithm for alternating space isometry.}
\label{algo:second-average}
\end{algorithm}

Let us first examine the running time of 
Algorithm~\ref{algo:second-average}.
\begin{prop}\label{prop:second-average-time}
Algorithm~\ref{algo:second-average} runs in time $\poly(q^{cm}, s, n)$.
\end{prop}
\begin{proof}
If Algorithm~\ref{algo:second-average} outputs ``$\bG$ does not satisfy the 
generic 
condition,'' then it just executes the For-loop in Step 3, which together runs in 
time $\poly(m, n, \log q)$. 

Otherwise, there are two For-loops at Step 4 and Step 4.c, which add 
multiplicative factors $q^{cm}$ and $s$, respectively. Other steps can be 
carried out in time $\poly(n, \log q)$. Therefore the whole algorithm runs in 
time $\poly(q^{cm}, s, n)$. 
% Furthermore, the number of $T$'s to be added to $L$ is upper bounded by $q^{cm}\cdot s$. 
\end{proof}

We then prove the correctness of Algorithm~\ref{algo:second-average} in the case 
that 
that it does not report ``$\bG$ does not satisfy the generic 
condition.''

\begin{prop}\label{prop:second-average-correctness}
Suppose that Algorithm~\ref{algo:second-average} does not report ``$\bG$ does not 
satisfy the generic 
condition.''
Then the 
algorithm lists the set of pseudo-isometries (possibly empty). In particular, 
$|\pseudo(\bG, \bH)|\leq q^{cm}\cdot s$.
\end{prop}
\begin{proof}
By Step 5.c, every $T$ added to $L$ is a pseudo-isometry. So we are left to show 
that $L$ contains all the pseudo-isometries. For this, take an arbitrary 
pseudo-isometry $T$. Then $T$ sends $\bA$ to some $\bB\in\cH^c$, i.e., $T^t\bA 
T=\bB$. In particular, $(T^t, T^{-1})\in \Adj(\bA, \bB)$. So when enumerating this 
$\bB\in \cH^c$, $(T^t, T^{-1})$ will pass all the tests in the following, and then 
be added to $L$. This concludes the proof.
\end{proof}

Therefore, to prove Theorem~\ref{thm:average}, the key is to analyze when a random 
$\bG$ satisfies the generic condition as in Algorithm~\ref{algo:second-average}. 
\begin{prop}\label{prop:second-average-analysis}
Let $m\geq c=20$, and let $\ell=\lfloor m/20\rfloor\in\N$. For all but at most 
$1/q^{\Omega(n\cdot \ell)}=1/q^{\Omega(nm)}$ fraction of $\bG=(G_1, \dots, 
G_m)\in\Lambda(n, q)^m$, there exists some $i\in[\ell]$, such that, letting 
$\bA=(G_{c(i-1)+1, \dots, c(i-1)})$, we have $|\Adj(\bA)|\leq q^{n}$.
\end{prop}

Clearly, Theorem~\ref{thm:average} follows from 
Propositions~\ref{prop:second-average-time},~\ref{prop:second-average-correctness},
 and~\ref{prop:second-average-analysis}.

\subsection{The average-case analysis}
\label{subsec:average-case-analysis}
We now formulate the key proposition that supports the proof of Proposition~\ref{prop:second-average-analysis}.
\begin{prop}\label{prop:average-analysis-key}
Let $c=20$. For all but at most 
$1/q^{\Omega(n)}$ fraction of $\bA=(G_1, \dots, 
G_c)\in\Lambda(n, q)^c$, we have $|\Adj(\bA)|\leq q^{n}$.
\end{prop}
Given Proposition~\ref{prop:average-analysis-key}, we easily obtain the following.
\begin{proof}[Proof of Proposition~\ref{prop:first-average-analysis}]
This is because, if $T\in\Aut(\bA)$, then $(T^t, T^{-1})\in\Adj(\bA)$. So 
$|\Aut(\bA)|\leq|\Adj(\bA)|$.
\end{proof}
\begin{proof}[Proof of Proposition~\ref{prop:second-average-analysis}]
We slice $\bG$ into 
$\ell=\lfloor m/c\rfloor$ segments, where each segment 
consists of $c$ random alternating matrices. 
Each segment is some $\bA\in\Lambda(n, q)^c$, with $\Pr[|\Adj(\bA)|> q^{ 
n}]\leq 1/q^{\Omega(n)}$. Since each $G_i$ is chosen independently and uniformly at random, the 
probability of every $(G_{c(i-1)+1}, \dots, G_{ci})$, $i\in[\ell]$, with 
$|\Adj((G_{c(i-1)+1}, \dots, G_{c(i-1)+c}))|>q^{ n}$, is upper bounded by 
$(1/q^{\Omega(n)})^\ell=1/q^{\Omega(nm)}$.
\end{proof}

The rest of this subsection is devoted to the proof of 
Proposition~\ref{prop:average-analysis-key}. For this we need the following from 
\cite{LQ}.
Given a tuple 
%of $r$ $n\times n$ matrices 
$\bA=(A_1, \dots, A_r)\in\Mat(n, q)^r$, define the image of $U\leq \F_q^n$ under 
$\bA$ as $\bA(U):=\langle\cup_{i=1}^rA_i(U)\rangle$.
\begin{defn}
We say $\bA=(A_1, \dots, A_r)\in\Mat(n, q)^r$ is \emph{stable}, if for any 
nonzero, proper $U\leq \F_q^n$, we have $\dim(\bA(U))>\dim(U)$.
\end{defn}
\begin{prop}[{\cite[Proposition 10 in arXiv version]{LQ}}]\label{prop:stable}
If $\bA\in \Mat(n, q)^r$ is stable, then $|\Adj(\bA)|\leq q^n$.
\end{prop}

A key technical result in \cite{LQ} is that, a random $\bA\in \Mat(n, q)^4$ is 
stable with probability $1-\frac{1}{q^{\Omega(n)}}$ \cite[Proposition 20 in 
arXiv version]{LQ}. However, we cannot directly apply that result to prove 
Proposition~\ref{prop:average-analysis-key}, because here we have alternating 
matrices instead of general matrices. So we have to run the arguments for the 
proof of \cite[Proposition 20 in 
arXiv version]{LQ} again, and carefully adjust some of the details there to 
accommodate the structure of alternating matrices. 

To start with, we need the following easy linear algebraic result, which suggests 
the 
connection between random alternating matrices and random general matrices.
\begin{lemma}\label{lemma: random alter to random matrix}
Let $d\in \Z^+$ and $d\geq 2$. Given two random alternating matrix $X, Y\in 
\Lambda(d, q)$, we can construct a 
matrix $P\in\Mat(d\times (d-1), q)$, whose columns are linear combinations of the 
columns of $X$ and $Y$, such that $P$ is a random matrix from $\Mat(d\times (d-1), 
q)$. 
\end{lemma}
\begin{proof}
Let $X$ and $Y$ be given as
$$X=\begin{bmatrix}
0 & x_{1,2} & x_{1,3} & \dots & x_{1,d}\\
-x_{1,2} & 0 & x_{2,3} & \dots & x_{2,d}\\
-x_{1,3} & -x_{2,3} & 0 & \dots & x_{3,d}\\
\vdots & \vdots & \vdots & \ddots & \vdots\\
-x_{1,d} & -x_{2,d} & -x_{3,d} & \dots & 0 \\
\end{bmatrix},~Y=\begin{bmatrix}
0 & y_{1,2} & y_{1,3} & \dots & y_{1,d}\\
-y_{1,2} & 0 & y_{2,3} & \dots & y_{2,d}\\
-y_{1,3} & -y_{2,3} & 0 & \dots & y_{3,d}\\
\vdots & \vdots & \vdots & \ddots & \vdots\\
-y_{1,d} & -y_{2,d} & -y_{3,d} & \dots & 0 \\
\end{bmatrix},$$
where each $x_{i,j}$ and $y_{i,j}$ are independent random variables from 
$\F_q$. Define
\begin{equation*}\label{eq: construct M from random alt matrix}
\begin{split}
M&=\begin{bmatrix}
y_{1,2} & x_{1,2}+y_{1,3} & x_{1,3}+y_{1,4} & \dots & 
x_{1,d-1}+y_{1,d} & x_{1,d} \\
-x_{1,2} & y_{2,3} & x_{2,3}+y_{2,4} & \dots & x_{2,d-1}+y_{2,d} 
&x_{2,d}-y_{1,2} \\
-x_{1,3}-y_{2,3} & -x_{2,3} & y_{3,4} & \dots & x_{3,d-1}+y_{3,d} 
&x_{3,d}-y_{1,3} \\
\vdots & \vdots & \vdots & \ddots & \vdots &\vdots\\
-x_{1,d}-y_{2,d} & -x_{2,d}-y_{3,d} & -x_{3,d}-y_{4,d} & \dots & 
-x_{d,d-1} &-y_{1,d} 
\end{bmatrix}\\
&:=\begin{bmatrix}
z_{1,1} & z_{1,2} & z_{1,3} & \dots & z_{1,d-1} & z_{1,d}\\
z_{2,1} & z_{2,2} & z_{2,3} & \dots & z_{2,d-1} & z_{2,d} \\
z_{3,1} & z_{3,2} & z_{3,3} & \dots & z_{3,d-1} & z_{3,d} \\
\vdots & \vdots & \vdots & \ddots & \vdots & \vdots \\
z_{d,1}&z_{d,2} & z_{d,3} & \dots & z_{d, d-1} & z_{d,d} \\
\end{bmatrix},
\end{split}
\end{equation*}
be the matrix obtained by adding the $(i+1)$th column of $Y$ to 
the $i$th column of $X$ for $i\in[d-1]$, and add the first column of $Y$ to the 
$d$th column of $X$. Let $P$ be the $d\times (d-1)$ matrix consisting of the first 
$(d-1)$ columns of 
$M$. We need to show that $P$ is uniformly sampled from $\Mat(n\times (d-1), q)$ 
as $X$ and $Y$ are uniformaly sampled from $\Lambda(d, q)$. 

To see this, first note that for any two random variable $x$ and $y$, which are 
chosen independently and uniformly at random from $\F_q^d$, $x\pm y$ are also new 
random variables which are chosen uniformly at random from $\F_q^d$, and is 
independent with either $x$ or $y$. Thus each $z_{i,j}$ is again a random 
variable which is chosen uniformly at random from $\F_q^d$ for $i,j\in[d]$. 

We then exploit the linear relations among the $z_{i,j}$'s. In fact, we only need to 
focus on the anti-diagonal 
directions, as 
$$z_{1,i}+z_{2,i-1}+\cdots+z_{i,1}+z_{i+1,d}+z_{i+2,d-1}+\dots+z_{d,i+1}=0$$
for any $i\in[d]$. Thus, we can view 
$z_{1,i},z_{2,i-1},\cdots,z_{i,1},z_{i+2,d-1},\dots, z_{d,i+1}$ (note the missing 
$z_{i+1,d}$) to be mutually independent for each $i\in[d]$, then every entry in 
$P$ can be viewed as 
chosen independently and uniformly at random. This can be verified in a 
straightforward way, and we can conclude the proof.
\end{proof}

\begin{remark}\label{remark: construct square random matrix}
Following the similar argument, if we would like to get an $d\times d$ random 
matrix over $\F_q$, we can in turn do the following:
take two $d\times d$ random alternating matrices $X$ and $Y$ and construct $M$ as 
in Lemma~\ref{eq: construct M from random alt matrix}. 
We then take another two random alternating matrices $Z$ and $W$. We add up the 
first column of $Z$ and $W$, of which each coordinates can be viewed as chosen 
independently and uniformly at random. We replace the last column of $M$ by the 
new random vector, which gives an $d\times d$ random matrix.
\end{remark}

We are now ready to prove Proposition~\ref{prop:average-analysis-key}.

\begin{proof}[Proof of Proposition~\ref{prop:average-analysis-key}]
Given Proposition~\ref{prop:stable}, we need upper bound the probability of a 
random $\bA\in\Lambda(n, q)^c$, such that 
$\bA$ is \emph{not} stable, by $1/q^{\Omega(n)}$. 

By the union bound, we know that
\begin{equation}\label{eq: upper bound}
\begin{split}
\Pr[\bA\in\Lambda(n, q)^c\text{ is not stable}]\leq 
\sum_{\substack{U\leq\F_q^n,\\1\leq\dim(U)\leq 
n-1}}\Pr[\bA\in\Lambda(n, q)^c, \dim(\bA(U))\leq \dim(U)].
\end{split}
\end{equation}

We first simplify the right-hand-side. For a non-zero, proper $U\leq \F_q^n$, 
let $A_U:=\{\bA \in\Lambda(n,q)^r:\dim(\bA (U))\leq \dim(U)\}$. Clearly,
$$\Pr[\bA \in \Lambda(n, q)^c,\dim(\bA (U))\leq 
\dim(U)]=\frac{|A_U|}{|\Lambda(n,q)^c|}.$$ 

We show that for any two dimension-$d$ subspaces $U$ and $V$, $|A_U|=|A_V|$. 
To see this, let $T\in\GL(n,q)$ be any invertible matrix that sends $V$ to $U$. 
Note that $T$ further induces a linear map from $\Lambda(n,q)^r$ to 
itself by sending $\bA$ to $T^t \bA  T$. Since $T$ is invertible, this map is a
bijection. Moreover, 
for any $\bA \in A_U$, we claim that $T^t\bA  T\in A_V$. This is because
\begin{equation*}
\dim((T^t\bA  T) (V))=\dim((T^t\bA)  (U))=\dim(\bA  (U))\leq 
\dim(U)=\dim(V),
\end{equation*}
where the second equality holds since left and right multiplying invertible 
matrices does not change the rank of a matrix. To summarize, if $\dim(U)=\dim(V)$, 
then
$$\Pr[\bA \in\Lambda(n, q)^c,\dim(\bA (U))\leq \dim(U)]=\Pr[\bA \in\Lambda(n, 
q)^c,\dim(\bA (V))\leq \dim(V)].$$

The right-hand-side of~\ref{eq: upper bound} can be then 
simplified as
\begin{equation}\label{eq: simplified version}
\Pr[\bA \in\Lambda(n, q)^c~is~not~stable]\leq 
\sum_{d=1}^{n-1}\gbinom{n}{d}{q}\cdot \Pr[\bA \in\Lambda(n, q)^c,\dim(\bA 
(U_d))\leq d].
\end{equation}
where $U_d$ is the $d$-dimensional subspace of $\F_q^n$ spanned by the first $d$ 
standard basis $e_1,\dots,e_d$. 

The next goal is to upper bound $\gbinom{n}{d}{q}\Pr[\bA 
\in \Lambda(n, q)^c,\dim(\bA (U_d))\leq d]$ for $d=1, \dots, n-1$.

Let $A_i^d$ be the $n\times d$ matrices consists of the first $d$ columns of $A_i$ for $i\in 
[c]$. (Note that the superscript here does not denote exponentiation.) Let 
$A^d=[A_1^d,\cdots,A_c^d]\in \Mat(n\times cd,q)$. 
Then
$\dim(\bA (U_d))$ is just the rank of $A^d$. Note 
that for $i\in[c]$, the first $d$ row of $A_i^d$ can be viewed as a random 
alternating matrix from $\Lambda(d,q)$, and the last $n-d$ rows of $A_i^d$ can be 
viewed as a
$(n-d)\times d$ random matrix. Moreover, these two matrices can be viewed as being 
chosen independently. 

By Lemma~\ref{lemma: random alter to random matrix} together 
with Remark~\ref{remark: construct square random matrix}, there exist a series of 
column operations represented by an invertible matrix $R\in \GL(cd\times cd,q)$, 
such that the following holds. Let $V^d\in \Mat(n\times 5d,q)$ be the matrix 
consists of the first $5d$ 
columns of $A^dR$. Then $V^d$ can be viewed as chosen independently and uniformly 
at random from $M(n\times 5d,q)$, as $\bA$ is chosen uniformly at random from $\Lambda(n,q)^c$. 
Note that when $d=1$, the first row of $A_i^d$ is $0$ for all $i\in[c]$. This degenerate case
suggest us to consider $V^1$ as randomly choosing from $M((n-1)\times 5,q)$.
Note that
$$\Pr[\bA 
\in \Lambda(n, q)^c,\dim(\bA (U_1))\leq 1]\leq \Pr[V^1 \in M((n-1)\times 
5,q),~rk(V^1)\leq 1]$$
and
$$\Pr[\bA 
\in \Lambda(n, q)^c,\dim(\bA (U_d))\leq d]\leq \Pr[V^d \in M(n\times 
5d,q),~rk(V^d)\leq d]$$
for $2\leq d\leq n-1$.

We consider how to construct an $(n-1)\times 5$ matrix such that its rank is not larger than $1$.
One way to do so is to pick one column fix its coordinates; then let the rest $4$ columns be scalar of the picked ones. This procedure gives the bound 
$$
\Pr[V^1 \in M((n-1)\times 
5,q),~rk(V^1)\leq 1]\leq \frac{\binom{5}{1}\cdot q^{n-1}\cdot 
q^{5-1}}{q^{5(n-1)}}=\frac{5}{q^{4n-8}}.
$$
So we have 
\begin{equation}\label{eq:d=1}
\gbinom{n}{1}{q}\cdot\Pr[\bA 
\in \Lambda(n, q)^c,\dim(\bA (U_1))\leq 1]\leq \frac{5}{q^{3n-8}}.
%\in \frac{1}{q^{\Omega(n)}}.
\end{equation}

Using the same idea, we deal with $2\leq d\leq n-1$. All possible $V^d$ such that $rk(V^d)\leq d$ can be constructed by 
first 
choosing $d$ columns in $V^d$ and fixing their entries, and then choosing the other 
columns from their linear span. This gives the bound 
$$\Pr[V^d\in M(n\times 5d,q),rk(V^d)\leq d]\leq\frac{\binom{5d}{d}\times 
q^{nd}\times q^{4d^2}}{q^{5nd}}\leq \frac{1}{q^{4nd-4d^2-5d}},$$
where the last inequality uses $\binom{5d}{d}\leq 2^{5d}\leq q^{5d}$. For $d\leq \frac{n}{2}$, we upper bound $\gbinom{n}{d}{q}$ by $q^{nd}$. This gives that
\begin{equation}\label{eq: 1 - n/2}
\gbinom{n}{d}{q}\Pr[\bA \in \Lambda(n, q)^c,\dim(\bA (U_d))\leq d]\leq 
\frac{1}{q^{3nd-4d^2-5d}}\leq \frac{1}{q^{6n-26}}.
%\in\frac{1}{q^{\Omega(n)}}.
\end{equation}% when $d=1$, $\frac{1}{q^{3nd-4d^2-5d}}$ takes the maximum.
For $\frac{n}{2}< d\leq n-2$, we upper bound $\gbinom{n}{d}{q}$ by $q^{n(n-d)}$. This gives that
\begin{equation}\label{eq:n/2 - n-2}
\gbinom{n}{d}{q}\Pr[\bA \in \Lambda(n, q)^c,\dim(\bA (U_d))\leq d]\leq 
\frac{1}{q^{5nd-n^2-4d^2-5d}}\leq \frac{1}{q^{n-6}}.
%\in\frac{1}{q^{\Omega(n)}}.
\end{equation}%when $d=n-2$, $\frac{1}{q^{5nd-n^2-4d^2-5d}}$ takes the maximum.
For $d=n-1$, we note that $\Pr[V^d\in M(n\times 5(n-1),q),rk(V^d)\leq n-1]$ is 
the 
probability that $V^d$ is not of rank $n$ when $n\geq 2$~\cite[Fact 4 in arXiv 
version]{LQ}. This gives the bound
\begin{equation}\label{eq: n-1}
\gbinom{n}{n-1}{q}\Pr[\bA \in \Lambda(n, q)^c,\dim(\bA (U_d))\leq d]\leq 
\frac{n\times n}{q^{5(n-1)-n+1}}=\frac{n^2}{q^{4(n-1)}}. 
%\in \frac{1}{q^{\Omega(n)}}.
\end{equation}

Combining equations from~\ref{eq: upper bound} to~\ref{eq: n-1}, 
we have
\begin{equation*}
\begin{split}
\Pr[\bA\in\Lambda(n, q)^c\text{ is not stable}] \leq& \sum_{\substack{U\leq\F_q^n,\\1\leq\dim(U)\leq n-1}}\Pr[\bA \in\Lambda(n, q)^c,\dim(\bA (U))\leq d]\\
\leq& \sum_{d=1}^{n-1}\gbinom{n}{d}{q}\Pr[\bA \in \Lambda(n, q)^c,\dim(\bA 
(U_d))\leq d]\leq \frac{1}{q^{\Omega(n)}},\\
\end{split}
\end{equation*}
which concludes the proof.
\end{proof}

\begin{remark}[Upgrading to the linear algebraic Erd\H{o}s-R\'enyi 
model]\label{rem:LinER}
In 
\cite{LQ}, the linear algebraic Erd\H{o}s-R\'enyi model, $\LinER(n, m, q)$, was 
introduced as the uniform distribution over all $m$-dimensional subspaces of 
$\Lambda(n, q)$. Randomly sampling $m$-tuples of $n\times n$ alternating matrices 
was termed as the naive model in \cite{LQ}. It was also shown in \cite{LQ} that 
the analysis in the naive model can be upgraded, with a mild loss in the 
parameters, to an analysis in $\LinER(n, m, q)$. Such an upgrade can also be done 
similarly for the analysis here, though with a little bit more work than in 
\cite{LQ}. We omit the details.
\end{remark}

%%%%%
\section{On testing isomorphism of groups with genus 2 radicals}
\label{sec:genus2-rad}
% !TEX root = WL.tex
In this section we show how to combine the methods of \cite{GQ} for groups with abelian radicals and the methods of \cite{BMW} 
to study subclasses of groups whose solvable radicals are $p$-groups of class 2. 
Recall that $p$-groups of class 2 are considered as difficult as the general case for group isomorphism, so we did not
expect to beat the $n^{\log n}$ bound for this entire class. However,
as a corollary of the results in this section, we give an $n^{O(\log \log n)}$-time isomorphism test for a class of groups
whose radicals have genus 2. We shall work throughout with the following class of groups:
\begin{quotation}
\noindent Let $\mathcal{G}$ be the class of groups $G$ whose solvable radical, $\Rad(G)$, is a $p$-group of exponent $p\neq 2$ and
class 2 upon which $G$ acts as inner automorphisms of $\Rad(G)$.
\end{quotation}

In \cite{GQ} the classical strategy of using actions and cohomology was 
formally analyzed, showing that \GpI ``splits'' into two problems: Action 
Compatibility (\ActComp), and Cohomology Class Isomorphism (\CohIso); we state 
their definitions in the relevant sections below. When $G$ has a normal subgroup 
$N$ 
we may consider $G$ as an extension of $N$ by 
$Q=G/N$; both \ActComp and \CohIso have as their witnesses certainly elements of 
$\Aut(N) \times \Aut(Q) \times (Q \to N)$, and two groups are isomorphic if, and only if, there 
is a witness that works simultaneously for \ActComp and \CohIso (see \cite{GQ} for 
a leisurely exposition). Furthermore, \ActComp and \CohIso each reduce to \GpI.

The two key cases to handle first are the extreme situations with regards to this 
natural splitting: semi-direct products, where the isomorphism problems reduce to 
just \ActComp; and ``central'' products (or rather, where $G/\Rad(G)$ acts 
trivially on the radical $\Rad(G)$), where the problem reduces to (nonabelian) 
\CohIso. The class $\mathcal{G}$ that we consider here is of the second type
of extreme situation. We expect the first yield to techniques in~\cite[Section~3]{GQ},
perhaps using methods to solve isometry~\cite{IQ}, but we are not yet able to
see a clear path to this case.

\subsection{Preliminaries on genus 2 groups} 
\label{sec:prelim:genus2}
We briefly recall definitions and results on the automorphism group of groups of genus 2; see \cite{BMW} for details. For any group $G$, let $Z=Z(G)$ and $G'=[G,G]$; then we define the commutator map of $G$ as $\circ_G \colon G/Z \times G/Z \to G'$. Two groups $G,H$ are \emph{isoclinic} if there are isomorphisms $\varphi\colon G / Z(G) \to H / Z(H)$ and $\hat{\varphi}\colon G' \to H'$ such that $g_1^\varphi \circ_H g_2^\varphi = (g_1 \circ_G g_2)^{\hat{\varphi}}$. When $G,H$ are nilpotent of class 2, their commutator maps are in fact $\Z$-bilinear (note that in this case $G/Z(G)$ is abelian), and the groups are isoclinic iff $\circ_G$ and $\circ_H$ are pseudo-isometric, by definition (recall \S\ref{sec:background}). Given a bilinear map $\circ\colon U \times V \to W$ ($U,V,W$ abelian groups), its \emph{centroid} is 
\begin{align*}
C(\circ) &:= \{ (\varphi, \psi, \rho) \in \End(U) \times \End(V) \times \End(W) : (\forall u \in U, v \in V)[u^{\varphi} \circ v = (u \circ v)^\rho = u \circ (v^\psi)]  \}; 
\end{align*}
the centroid is the largest ring of scalars over which $\circ$ is bilinear. A nilpotent group $G$ of class 2 is isoclinic to a direct product $H_1 \times \dotsb \times H_s$ of directly indecomposable groups; the \emph{genus} of $G$ is the maximum rank of $[H_i, H_i]$ as a $C(\circ_{H_i})$-module. Although the concept of genus is fully general, we focus on $p$-groups of exponent $p$ and class 2; in this case isoclinism and isomorphism coincide, and centrally indecomposable $p$-groups of class 2 and exponent $p$ have their centroids a finite field of characteristic $p$. For a biadditive map $\circ\colon U \times U \to V$, let $\pseudo(\circ)$ denote its group of pseudo-isometries; if $\circ$ is bilinear over a field $\F$, let $\pseudo_\F(\circ) = \pseudo(\circ) \cap (\GL_\F(U) \times \GL_\F(V))$. Given a finite field $\F$ of characteristic $p$, its \emph{Galois group} denoted $\Gal(\F)$, consists of those field automorphisms of $\F$ that act trivially on the prime subfield $\Z_p \leq \F$; $\Gal(F)$ is cyclic of order $[\F : \Z_p] = \log_p |\F|$, generated by the Frobenius automorphism $a \mapsto a^p$. 

\begin{prop}[{See, e.\,g., \cite[Prop.~2.4]{BMW}}] \label{prop:genus2-aut}
Let $P$ be a $p$-group of class 2 and exponent $p$ satisfying $Z(P) = [P,P]$. Then $\Aut(P) = \pseudo(\circ_P) \ltimes \Hom(P/Z(P), Z(P))$. If $\circ_P$ is $\F$-bilinear, then $\pseudo_\F(\circ_P) \unlhd \pseudo(\circ_P)$, with quotient $\pseudo(\circ_P) / \pseudo_\F(\circ_P) \cong \Gal(F)$.

Note that elements of $\pseudo_\F(\circ_P) \ltimes \Hom(P/Z(P), Z(P))$ are faithfully represented by matrices $\smallmat{\alpha_V & d\alpha \\ 0 & \alpha_Z}$, where $\alpha_V \in \Aut(P/Z(P))$, $\alpha_Z \in \Aut(Z(P))$, and $d\alpha\colon P/Z(P) \to Z(P)$ is linear.
\end{prop}

Recall that a map $\alpha\colon V \to W$ of $\F$-vector spaces is $\F$-semilinear if it is additive ($\alpha(v + v') = \alpha(v) + \alpha(v')$) and it is ``twisted'' linear, that is, $\alpha(\lambda v) = \lambda^\gamma \alpha(v)$, where $\gamma \in \Gal(\F)$. From the preceding, it follows immediately that:

\begin{obs} \label{obs:semilinear}
Let $P$ be a $p$-group of class 2 and exponent $p$ such that $\circ_P$ is $\F$-bilinear. For any $\alpha \in \Aut(P)$, the induced automorphisms on $[P,P]$ and $P/[P,P]$ are both $\F$-semilinear.
\end{obs}

\begin{obs} \label{obs:center-vs-derived}
If $P$ is a $p$-group of class 2 and exponent $p$ such that $Z(P) \neq [P,P]$, then $P \cong Q \times A$, where $Q$ is characteristic subgroup of $P$ and satisfies $Z(Q) = [Q,Q]$, and $A$ is an elementary abelian $p$-group. Moreover, $Q$ and $A$ and the isomorphism $P \cong Q \times A$ can be constructed in polynomial time in the number of generators, even when the groups are given as a black box.
\end{obs}

\begin{proof}[Standard proof sketch]
$Z(P) \geq [P,P]$ since $P$ is of class 2. Since $P$ is of exponent $p$, $Z(P)$ is elementary abelian, and thus is a vector space $\Z_p^e$. Let $\{g_1, \dots, g_s\}$ be a generating set of $P$. Let $Q = \langle g_i : g_i \notin Z(P) \rangle$. Then $Q \cap Z(P) = [P,P]$. Let $A$ be a $\Z_p$-linear complement to $[P,P]$ in $Z(P)$. 
\end{proof}

\begin{thm}[{\cite{BMW, IQ}}] \label{thm:extends}
Let $P$ be a $p$-group of class 2, exponent $p \neq 2$, and genus $g$. Given $\alpha \in \Aut(Z(P))$, one can test whether $\alpha$ extends to an automorphism $\hat{\alpha} \in \Aut(P)$ in poly-logarithmic time when $g \leq 2$, and in polynomial time otherwise.
\end{thm}

\begin{proof}
When $g=2$, the result is immediate from \cite[Thm.~3.22]{BMW}, and their comments about its constructive nature (see \cite[\S6.2]{BMW}). In general, this is an isometry problem, which is solvable in polynomial time \cite{IQ}.
\end{proof}

\begin{thm} \label{thm:genus-iso}
Isomorphism of $p$-groups of class 2, exponent $p \neq 2$, and genus $\leq 2$ can be decided in poly-logarithmic time \cite[Thm.~1.1]{BMW}, and of genus $\leq \sqrt{\log |G|}$ can be decided in polynomial time \cite[Thm.~3]{IQ}.
\end{thm}

\subsection{Testing isomorphism in the class $\mathcal{G}$}

Our goal in this final section is to prove Theorem~\ref{thm:central}, which for convenience we now recall:
\begin{thm-central}
Let $\mathcal{G}$ be the class of groups $G$ defined at the start of Section~\ref{sec:genus2-rad}. 
%$\Rad(G)$ a $p$-group of class 2, exponent $p \neq 2$, and such that the induced outer action of $G/\Rad(G)$ on $\Rad(G)$ is trivial. 
Given groups $G_1, G_2$ of order $n$, it can be decided in $\poly(n)$ time if they lie in $\mathcal{G}$. 
If so, isomorphism can be decided, and a generating set for $\Aut(G_i)$ found, in time $n^{O(g + \log \log n)}$, where $g$ is the genus of $\Rad(G)$.
\end{thm-central}

We will need the following two results from Grochow--Qiao \cite{GQ}, which first require a few concepts we haven't yet discussed. 
Recall that a pair of subgroups $H_1,H_2 \leq G$ is a \emph{central decomposition} of $G$ if $\langle H_1,H_2 \rangle = G$ and $[H_1,H_2]=1$. 
Given two groups $M_1, M_2$ and an isomorphism 
$\varphi\colon Y_1 \to Y_2$ between two subgroups $Y_i \leq Z(M_i)$, the quotient of $M_1 \times M_2$ by $\{(y^{-1}, \varphi(y)) : y \in Y_1\}$ is the \emph{central product} of $M_1$ and $M_2$ along $\varphi$, denoted $M_1 \times_{\varphi} M_2$, and $\varphi$ is called the \emph{amalgamating map}. In this case, $\{M_1,M_2\}$ is a central decomposition of $M_1 \times_\varphi M_2$; conversely, if $\{H_1,H_2\}$ is a central decomposition of a group 
$G$, then there exist $Y_i \leq Z(H_i)$ and an isomorphism $\varphi\colon Y_1 \to Y_2$ such that $G \cong H_1 \times_\varphi H_2$. 

\begin{lem}[{\cite[Lem.~3.10]{GQ}}] \label{lem:central}
Let $N \unlhd G$, and suppose $G$ acts on $N$ as inner automorphisms of $N$.
 Then there is a subgroup $H \leq G$, constructible in time  $\poly(|G|)$, such that $H \cap N = Z(N)$, $H / N = Q$, 
 and $\{N, H\}$ is a central decomposition of $G$. We denote this subgroup $H$ by $G|_{Z(N)}$.
\end{lem}

\begin{prop}[{Special case of \cite[Prop.~3.13]{GQ}}] \label{prop:central_amalgam}
Let $G_i$ ($i=1,2$) be a group such that $\Rad(G_i)=P$ is a $p$-group of class 2, exponent $p$, and genus 2, and such that $Q=G_i / \Rad(G_i)$ acts on $\Rad(G_i)$ by inner automorphisms of $\Rad(G_i)$. Suppose that $G_1|_{Z(P)} \cong G_2|_{Z(P)}$ (as in Lem.~\ref{lem:central}), which we denote by $\hat{Q}$, and let $\varphi_i \colon Z(P) \to Z(\hat{Q})$ be the corresponding amalgamating maps. Then $G_1 \cong G_2$ iff there exist $(\alpha,\beta) \in \Aut(P) \times \Aut(\hat{Q})$ such that $\varphi_1 = \beta^{-1}|_{Z(\hat{Q})} \circ \varphi_2 \circ \alpha|_{Z(P)}$. 
\end{prop}

\begin{prop}[{See \cite[\S6.1.2, p.~1186]{GQ}}] \label{prop:coho-iso}
Let $G$ be a group with $\Rad(G) = Z(G)$, and let $Q = G/Z(G)$ an elementary abelian group. Given $\beta \in \Aut(Q)$, one can compute in $\poly \dim Z(G)$ time a single $\alpha \in \Aut(Z(G))$ and a basis of a linear subspace $L \subseteq \End(Z(G))$ such that $(\beta,\gamma) \in \Aut(G)$ iff $\gamma \in \alpha + L$.
\end{prop}

\begin{proof}[Proof of Thm.~\ref{thm:central}]
Let $G_1, G_2$ be groups satisfying the hypotheses. In $\poly(|G|)$ time, find $\Rad(G_i)$ and denote this by $P_i'$. By Lem.~\ref{lem:central}, construct $\hat{Q}_i = G_i|_{Z(P_i')}$ and the amalgamating maps $\varphi_i'\colon Z(P_i') \to Z(\hat{Q}_i)$. Using Thm.~\ref{thm:genus-iso} \cite{BMW, IQ}, decide whether $P_1' \cong P_2'$; if not, then $G_1 \not\cong G_2$ and we can stop, and if so, then let $\rho'\colon P_1' \to P_2'$ be such an isomorphism. 

Note (Observation~\ref{obs:center-vs-derived}) that it may be the case that $P_i' \cong P_i \times A_i$ for some abelian groups $A_i$; if this is the case, we can find $P_i$ and $A_i$ such that $Z(P_i) = [P_i, P_i]$ in polynomial time. Replace $P_i'$ by $P_i$ and replace $\rho'$ by $\rho := \rho'|_{P_1}$; this will not hurt us later because $P_i$ is characteristic in $P_i'$, and therefore also in $G_i$. Intuitively, the only place that $A_i$ interacts with $P_i'$ is as a direct product, and the only way $A_i$ interacts with $\hat{Q_i}$ is as a subgroup of its center, where $A_i$ still appears. 

Next, since $\hat{Q}_i$ is a group with $\Rad(\hat{Q}_i) \leq Z(\hat{Q}_i)$, by \cite{GQ} we can decide whether $\hat{Q}_1 \cong \hat{Q}_2$ in time $n^{O(\log \log n)}$; if not, then $G_1 \not\cong G_2$ and we can stop, and if so, let $\tau\colon \hat{Q}_1 \to \hat{Q}_2$ be such an isomorphism. Let $\varphi_1 = \varphi_1'$ and $\varphi_2 = \tau^{-1} \circ \varphi_2' \circ \rho^{-1}$. These are both isomorphisms $Z(P_1) \to Z(\hat{Q}_1)$, so from now on we let $P = P_1$ and $\hat{Q} = \hat{Q}_1$, and we have $G_i \cong P \times_{\varphi_i'} \hat{Q}$ for $i=1,2$.

Now, by Proposition~\ref{prop:central_amalgam}, $G_1 \cong G_2$ iff there exists $(\alpha,\beta) \in \Aut(P) \times \Aut(\hat{Q})$ such that 
\begin{equation} \label{eqn:central-iso}
\varphi_1' = \beta^{-1}|_{Z(\hat{Q})} \circ \varphi_2' \circ \alpha|_{Z(P)}.
\end{equation}
By Observation~\ref{obs:semilinear}, $\alpha|_{Z(P)}$ is $\F$-semilinear, and since $P$ has genus $g$, $Z(P) \cong \F^g$. Enumerate $\GammaL(\F^g)$; for each $\alpha \in \GammaL(\F^g)$, check whether $\alpha$ extends to an automorphism of $P$ (Theorem~\ref{thm:extends} \cite{BMW, IQ}). Let $Q = \hat{Q} / Z(\hat{Q}) = \hat{Q} / \Rad(\hat{Q})$. Enumerate $\gamma \in \Aut(Q)$. For each $\alpha \in \GammaL(\F^g)$ that extends to an automorphism of $P$, and each $\gamma \in \Aut(Q)$, we seek $\beta \in \Aut(Z(\hat{Q}))$ such that $(\gamma,\beta)$ induces an automorphism of $\hat{Q}$ and $(\alpha,\beta)$ satisfies (\ref{eqn:central-iso}). By Proposition~\ref{prop:coho-iso}, the set of such $\gamma$ such that $(\gamma,\beta)$ is an automorphism of $\hat{Q}$ is an affine linear space $\beta_0 + B$, where $B$ is a linear subspace of $\End(Z(\hat{Q}))$, and we can compute $\gamma_0$ and a basis for $B$ in polynomial time. Once $\alpha$ is fixed, (\ref{eqn:central-iso}) is linear in $\beta$. Intersecting the linear space which solves (\ref{eqn:central-iso}) with the affine space $\beta_0 + B$ is standard linear algebra, and can thus be computed in polynomial time. 

To summarize, for each $\alpha \in \Aut_\F(Z(P)) \cong \GammaL_g(\F)$ and each $\gamma \in \Aut(Q)$, we can compute a single element and generating set for those $\beta$ such that $\alpha$ extends to an automorphism $P$, $(\beta,\gamma) \in \Aut(\hat{Q})$, and $(\alpha,\beta)$ satisfy (\ref{eqn:central-iso}). Taking the union over all choices in $\GammaL_g(\F)$ and $\Aut(Q)$ gives us the coset of isomorphisms $G_1 \to G_2$.

\textit{Analysis of running time.} When $g \leq O(\log \log |G|)$, we have $|\GammaL_g(\F)| \sim |\Gal(\F)| \cdot \F^{g^2} = k(p^k)^{g^2} = k (p^{kg})^g = k |Z(P)|^g \leq |G|^{g + o(1)}$ where $|G| \geq |\F| = p^k$, so their number is not too large, and $\GammaL_g(\F)$ is easily enumerated in $n^{O(g)}$ time. By \cite{BCGQ}, $\Aut(Q)$ can be listed in time $n^{O( \log \log n)}$. Since we are enumerating over both of these, we take their product $n^{O(g + \log\log n)}$, which ends up dominating the runtime. By \cite{GQ}, isomorphism of $\hat{Q}_1$ and $\hat{Q}_2$ can be tested in $n^{O( \log \log n)}$ time. The rest is polynomial time or poly-logarithmic time by previous results, or linear algebra (poly-logarithmic time in $|G|$). 
\end{proof}

\begin{remark}
There is some hope
when $g \leq 2$ in Theorem~\ref{thm:central}---due to the poly-logarithmic isomorphism test of \cite{BMW}---to improve this poly-logarithmic time.
However, a prerequisite is first solving isomorphism of groups with no abelian normal subgroups in poly-logarithmic time, rather than just polynomial \cite{BCQ}.
\end{remark}

\section*{Acknowledgments}
The authors would like to acknowledge V. Arvind and M. Grohe for useful comments 
on hypergraph $k$-WL, Avinoam Mann for discussions on random 
generation of $p$-groups, and L\'aszl\'o Babai and Xiaorui Sun for discussions on 
average-case algorithms for testing isomorphism of $p$-groups of class $2$ and 
exponent $p$.
P.~A.~B. was partially supported by NSF grant DMS-1620362. 
J.~A.~G. was partially supported by NSF grant DMS-1750319.
Y.~L. was partially supported by ERC Consolidator Grant 615307-QPROGRESS.
Y.~Q. was partially supported by the Australian Research Council DECRA
DE150100720.
J.~B.~W. was partially supported by NSF grant DMS-1620454.
P.~A.~B. and J.~B.~W. also acknowledge the Hausdorff Institute for Mathematics, and the University of Auckland where some of this research was conducted. 
P.~A.~B., J.~A.~G., J.~B.~W., and Y.~Q. also acknowledge the Santa Fe Institute, where some of this research was conducted.

\bibliographystyle{alpha}
\bibliography{reference}

\begin{bibdiv}
\begin{biblist}

\bib{Adi57}{article}{
      author={Adian, Sergei~I.},
       title={Unsolvability of some algorithmic problems in the theory of
  groups},
        date={1957},
     journal={Trudy Moskovskogo Matematicheskogo Obshchestva},
      volume={6},
       pages={231\ndash 298},
}

\bib{AFKV}{article}{
      author={Arvind, V.},
      author={Fuhlbr{\"u}ck, Frank},
      author={K{\"o}bler, Johannes},
      author={Verbitsky, Oleg},
       title={On {Weisfeiler--Leman} invariance: Subgraph counts and related
  graph properties},
        note={arXiv:1811.04801, 2018},
}

\bib{Bab16}{inproceedings}{
      author={Babai, L{\'{a}}szl{\'{o}}},
       title={Graph isomorphism in quasipolynomial time [extended abstract]},
   booktitle={Proceedings of the 48th {A}nnual {ACM} {SIGACT} {S}ymposium on
  {T}heory of {C}omputing, {STOC} 2016},
       pages={684\ndash 697},
         url={http://doi.acm.org/10.1145/2897518.2897542},
        note={arXiv:1512.03547, version 2},
}

\bib{Babai79}{misc}{
      author={Babai, L{\'{a}}szl{\'{o}}},
       title={Lecture on graph isomorphism},
         how={U. Toronto, Dept. of Computer Science, Mimeographed lecture
  notes},
        date={1979},
}

\bib{Bae}{article}{
      author={Baer, Reinhold},
       title={Groups with abelian central quotient group},
        date={1938},
     journal={Transactions of the American Mathematical Society},
      volume={44},
      number={3},
       pages={357\ndash 386},
}

\bib{BB99}{article}{
      author={Babai, L{\'a}szl{\'o}},
      author={Beals, Robert},
       title={A polynomial-time theory of black box groups {I}},
        date={1999},
     journal={London Mathematical Society Lecture Note Series},
       pages={30\ndash 64},
}

\bib{Boker}{article}{
      author={B{\"o}ker, Jan},
       title={Color refinement, homomorphisms, and hypergraphs},
        note={arXiv: 1903.12432, 2019},
}

\bib{Boo59}{article}{
      author={Boone, William~W.},
       title={The word problem},
        date={1959},
     journal={Annals of {Mathematics}},
       pages={207\ndash 265},
}

\bib{Bra35}{article}{
      author={Brahana, H.R.},
       title={Metabelian groups and trilinear forms},
        date={1935},
     journal={Duke Mathematical Journal},
      volume={1},
      number={2},
       pages={185\ndash 197},
}

\bib{BCGQ}{inproceedings}{
      author={Babai, L{\'a}szl{\'o}},
      author={Codenotti, Paolo},
      author={Grochow, Joshua~A.},
      author={Qiao, Youming},
       title={Code equivalence and group isomorphism},
   booktitle={Proceedings of the {Twenty-Second} {Annual} {ACM--SIAM}
  {Symposium} on {Discrete} {Algorithms} {SODA} 2011},
   publisher={SIAM},
     address={Philadelphia, PA},
       pages={1395\ndash 1408},
}

\bib{BCQ}{inproceedings}{
      author={Babai, L{\'{a}}szl{\'{o}}},
      author={Codenotti, Paolo},
      author={Qiao, Youming},
       title={Polynomial-time isomorphism test for groups with no abelian
  normal subgroups - (extended abstract)},
   booktitle={Automata, languages, and programming - 39th international
  colloquium, {ICALP} 2012},
       pages={51\ndash 62},
         url={https://doi.org/10.1007/978-3-642-31594-7\_5},
}

\bib{BES}{article}{
      author={Babai, L{\'{a}}szl{\'{o}}},
      author={Erd{\H{o}}s, Paul},
      author={Selkow, Stanley~M.},
       title={Random graph isomorphism},
        date={1980},
     journal={{SIAM} J. Comput.},
      volume={9},
      number={3},
       pages={628\ndash 635},
         url={http://dx.doi.org/10.1137/0209047},
}

\bib{MAGMA}{article}{
      author={Bosma, W.},
      author={{J. Cannon}, J.},
      author={Playoust, C.},
       title={The {M}agma algebra system {I}: the user language},
        date={1997},
     journal={J. Symb. Comput.},
       pages={235\ndash 265},
}

\bib{BMW}{article}{
      author={Brooksbank, Peter~A.},
      author={Maglione, Joshua},
      author={Wilson, James~B.},
       title={A fast isomorphism test for groups whose {Lie} algebra has genus
  2},
        date={2017},
        ISSN={0021-8693},
     journal={J. Algebra},
      volume={473},
       pages={545\ndash 590},
}

\bib{Git}{misc}{
      author={Brooksbank, Peter~A.},
      author={Maglione, Joshua},
      author={Wilson, James~B.},
       title={Thetensor.space},
   publisher={GitHub},
         how={\url{https://github.com/thetensor-space/}},
        date={2019},
}

\bib{BNV:enum}{book}{
      author={Blackburn, Simon~R.},
      author={Neumann, Peter~M.},
      author={Venkataraman, Geetha},
       title={Enumeration of finite groups},
      series={Cambridge Tracts in Mathematics},
   publisher={Cambridge University Press, Cambridge},
        date={2007},
      volume={173},
        ISBN={978-0-521-88217-0},
         url={https://doi.org/10.1017/CBO9780511542756},
}

\bib{BOW}{article}{
      author={Brooksbank, Peter~A.},
      author={O'Brien, Eamonn~A.},
      author={Wilson, James~B.},
       title={Isomorphism testing of graded algebras},
        note={arXiv:1708.08873, 2017},
}

\bib{BS84}{inproceedings}{
      author={Babai, L.},
      author={Szemeredi, E.},
       title={On the complexity of matrix group problems {I}},
   booktitle={Proceedings of the 25th annual symposium on foundations of
  computer science, {SFCS} 1984},
      series={SFCS '84},
   publisher={IEEE Computer Society},
     address={Washington, DC, USA},
       pages={229\ndash 240},
         url={https://doi.org/10.1109/SFCS.1984.715919},
}

\bib{BW:isom}{article}{
      author={Brooksbank, Peter~A.},
      author={Wilson, James~B.},
       title={Computing isometry groups of {Hermitian} maps},
        date={2012},
        ISSN={0002-9947},
     journal={Trans. Amer. Math. Soc.},
      volume={364},
      number={4},
       pages={1975\ndash 1996},
}

\bib{CFI}{article}{
      author={Cai, Jin{-}Yi},
      author={F{\"{u}}rer, Martin},
      author={Immerman, Neil},
       title={An optimal lower bound on the number of variables for graph
  identifications},
        date={1992},
     journal={Combinatorica},
      volume={12},
      number={4},
       pages={389\ndash 410},
         url={http://dx.doi.org/10.1007/BF01305232},
}

\bib{CH03}{article}{
      author={Cannon, John},
      author={Holt, Derek~F.},
       title={Automorphism group computation and isomorphism testing in finite
  groups},
        date={2003},
     journal={J. Symbolic Comput.},
      volume={35},
      number={3},
       pages={241\ndash 267},
}

\bib{DGR}{inproceedings}{
      author={Dell, Holger},
      author={Grohe, Martin},
      author={Rattan, Gaurav},
       title={Lov{\'{a}}sz meets {Weisfeiler} and {Leman}},
   booktitle={45th international colloquium on automata, languages, and
  programming, {ICALP} 2018},
       pages={40:1\ndash 40:14},
}

\bib{Minrank}{article}{
      author={Faug\`ere, Jean-Charles},
      author={Safey El~Din, Mohab},
      author={Spaenlehauer, Pierre-Jean},
       title={On the complexity of the generalized {M}in{R}ank problem},
        date={2013},
        ISSN={0747-7171},
     journal={J. Symbolic Comput.},
      volume={55},
       pages={30\ndash 58},
         url={https://doi.org/10.1016/j.jsc.2013.03.004},
}

\bib{Gilman}{article}{
      author={Gilmon, Robert},
       title={Algorithmic search in group theory},
        note={arXiv:1812.08116, 2018},
}

\bib{Gromov}{article}{
      author={Gromov, M.},
       title={Random walk in random groups},
        date={2003},
        ISSN={1016-443X},
     journal={Geom. Funct. Anal.},
      volume={13},
      number={1},
       pages={73\ndash 146},
         url={https://doi.org/10.1007/s000390300002},
}

\bib{GQ}{article}{
      author={Grochow, Joshua~A.},
      author={Qiao, Youming},
       title={Algorithms for group isomorphism via group extensions and
  cohomology},
        date={2017},
        ISSN={0097-5397},
     journal={SIAM J. Comput.},
      volume={46},
      number={4},
       pages={1153\ndash 1216},
}

\bib{Harris92}{article}{
      author={Harris, J.},
       title={Algebraic geometry},
        date={1992},
     journal={Graduate Texts in Mathematics 133},
}

\bib{Hig60}{article}{
      author={Higman, Graham},
       title={Enumerating {$p$}-groups. {I}: Inequalities},
        date={1960},
     journal={Proceedings of the London Mathematical Society},
      volume={3},
      number={1},
       pages={24\ndash 30},
}

\bib{HL}{article}{
      author={Heineken, Hermann},
      author={Liebeck, Hans},
       title={The occurrence of finite groups in the automorphism group of
  nilpotent groups of class {$2$}},
        date={1974},
        ISSN={0003-889X},
     journal={Arch. Math. (Basel)},
      volume={25},
       pages={8\ndash 16},
         url={https://doi.org/10.1007/BF01238631},
}

\bib{Helleloid-Martin}{article}{
      author={Helleloid, Geir~T.},
      author={Martin, Ursula},
       title={The automorphism group of a finite $p$-group is almost always a
  $p$-group},
        date={2007},
        ISSN={0021-8693},
     journal={J. Algebra},
      volume={312},
      number={1},
       pages={294\ndash 329},
}

\bib{HR}{article}{
      author={Holt, Derek~F.},
      author={Rees, Sarah},
       title={Testing modules for irreducibility},
        date={1994},
     journal={J. Austral. Math. Soc.},
      number={1},
       pages={1\ndash 16},
}

\bib{ImmermanLander}{incollection}{
      author={Immerman, Neil},
      author={Lander, Eric~S.},
       title={Describing graphs: a first-order approach to graph canonization},
        date={1990},
   booktitle={Complexity theory retrospective---in honor of {Juris Hartmanis}
  on the occasion of his 60th birthday},
      editor={Selman, Alan},
   publisher={Springer},
       pages={59\ndash 81},
}

\bib{IQ}{inproceedings}{
      author={Ivanyos, G{\'{a}}bor},
      author={Qiao, Youming},
       title={Algorithms based on $*$-algebras, and their applications to
  isomorphism of polynomials with one secret, group isomorphism, and polynomial
  identity testing},
   booktitle={Proceedings of the twenty-ninth annual {ACM-SIAM} symposium on
  discrete algorithms, {SODA} 2018},
       pages={2357\ndash 2376},
         url={https://doi.org/10.1137/1.9781611975031.152},
}

\bib{KL}{article}{
      author={Kantor, William~M.},
      author={Lubotzky, Alexander},
       title={The probability of generating a finite classical group},
        date={1990},
        ISSN={0046-5755},
     journal={Geom. Dedicata},
      volume={36},
      number={1},
       pages={67\ndash 87},
         url={https://doi.org/10.1007/BF00181465},
}

\bib{KS06}{article}{
      author={Kayal, Neeraj},
      author={Saxena, Nitin},
       title={Complexity of ring morphism problems},
        date={2006},
     journal={Computational Complexity},
      volume={15},
      number={4},
       pages={342\ndash 390},
         url={https://doi.org/10.1007/s00037-007-0219-8},
}

\bib{KTW}{article}{
      author={Kassabov, Martin},
      author={Tyburski, Brady},
      author={Wilson, James~B.},
       title={The number of isomorphism types of subgroups of simple groups is
  maximum possible},
        note={(in preparation)},
}

\bib{LGR}{article}{
      author={Le~Gall, Francois},
      author={Rosenbaum, David~J.},
       title={On the group and color isomorphism problems},
        note={arXiv:1609.08253, 2016},
}

\bib{LuksCompSeries}{article}{
      author={Luks, Eugene~M.},
       title={Group isomorphism with fixed subnormal chains},
        note={arXiv: 1511.00151, 2015},
}

\bib{Luks99}{inproceedings}{
      author={Luks, Eugene~M.},
       title={Hypergraph isomorphism and structural equivalence of boolean
  functions},
   booktitle={Proceedings of the thirty-first annual {ACM} symposium on theory
  of computing {STOC} 1999},
      editor={Vitter, Jeffrey~Scott},
      editor={Larmore, Lawrence~L.},
      editor={Leighton, Frank~Thomson},
   publisher={{ACM}},
       pages={652\ndash 658},
         url={https://doi.org/10.1145/301250.301427},
}

\bib{Luks93}{inproceedings}{
      author={Luks, Eugene~M.},
       title={Permutation groups and polynomial-time computation},
        date={1993},
   booktitle={{Groups and Computation}},
      series={DIMACS Ser. Discrete Math. Theoret. Comput. Sci.},
      volume={11},
}

\bib{LQ}{inproceedings}{
      author={Li, Yinan},
      author={Qiao, Youming},
       title={Linear algebraic analogues of the graph isomorphism problem and
  the erd{\H{o}}s-r{\'{e}}nyi model},
   booktitle={58th {IEEE} annual symposium on foundations of computer science,
  {FOCS} 2017},
       pages={463\ndash 474},
         url={https://doi.org/10.1109/FOCS.2017.49},
}

\bib{Lip}{article}{
      author={Lipyanski, Ruvim},
      author={Vanetik, Natalia},
       title={On {B}orel complexity of the isomorphism problems for graph
  related classes of {L}ie algebras and finite {$p$}-groups},
        date={2015},
        ISSN={0219-4988},
     journal={J. Algebra Appl.},
      volume={14},
      number={5},
       pages={1550078, 15},
         url={https://doi.org/10.1142/S0219498815500784},
}

\bib{LW54}{article}{
      author={Lang, Serge},
      author={Weil, Andr{\'e}},
       title={Number of points of varieties in finite fields},
        date={1954},
     journal={American Journal of Mathematics},
      volume={76},
      number={4},
       pages={819\ndash 827},
}

\bib{LW12}{article}{
      author={Lewis, Mark~L.},
      author={Wilson, James~B.},
       title={Isomorphism in expanding families of indistinguishable groups},
        date={2012},
     journal={Groups Complex. Cryptol.},
      volume={4},
      number={1},
       pages={73\ndash 110},
}

\bib{Maglione:filters2}{article}{
      author={Maglione, Joshua},
       title={Compatible filters for isomorphism testing},
        note={arXiv:1805.03732, 2018},
}

\bib{Maglione:filters-U}{article}{
      author={Maglione, Joshua},
       title={Longer nilpotent series for classical unipotent subgroups},
        date={2015},
        ISSN={1433-5883},
     journal={J. Group Theory},
      volume={18},
      number={4},
       pages={569\ndash 585},
         url={https://doi.org/10.1515/jgth-2015-0008},
      review={\MR{3365818}},
}

\bib{Maglione:filters}{article}{
      author={Maglione, Joshua},
       title={Efficient characteristic refinements for finite groups},
        date={2017},
        ISSN={0747-7171},
     journal={J. Symbolic Comput.},
      volume={80},
       pages={511\ndash 520},
}

\bib{Mann1999}{article}{
      author={Mann, Avinoam},
       title={Some questions about {$p$}-groups},
        date={1999},
        ISSN={0263-6115},
     journal={J. Austral. Math. Soc. Ser. A},
      volume={67},
      number={3},
       pages={356\ndash 379},
}

\bib{Nov55}{article}{
      author={Novikov, P.S.},
       title={On algorithmic undecidability of the word problem in the theory
  of groups},
        date={1955},
     journal={Trudy Mat. Inst. Steklov},
      volume={44},
       pages={1\ndash 144},
}

\bib{Rab58}{article}{
      author={Rabin, Michael~O.},
       title={Recursive unsolvability of group theoretic problems},
        date={1958},
     journal={{Annals of Mathematics}},
       pages={172\ndash 194},
}

\bib{RosenbaumBidirectional}{article}{
      author={Rosenbaum, David~J.},
       title={Bidirectional collision detection and faster deterministic
  isomorphism testing},
        note={arXiv: 1304.3935, 2013},
}

\bib{RosenbaumSolvable}{inproceedings}{
      author={Rosenbaum, David~J.},
       title={Breaking the nlog n barrier for solvable-group isomorphism},
   booktitle={{Proceedings of the Twenty-fourth Annual ACM-SIAM Symposium on
  Discrete Algorithms, SODA 13}},
   publisher={Society for Industrial and Applied Mathematics},
     address={Philadelphia, PA, USA},
       pages={1054\ndash 1073},
         url={http://dl.acm.org/citation.cfm?id=2627817.2627893},
}

\bib{RosenbaumWagner}{article}{
      author={Rosenbaum, David~J.},
      author={Wagner, Fabian},
       title={Beating the generator-enumeration bound for $p$-group
  isomorphism},
        date={2015},
        ISSN={0304-3975},
     journal={Theoret. Comput. Sci.},
      volume={593},
       pages={16\ndash 25},
}

\bib{Sim65}{article}{
      author={Sims, Charles~C.},
       title={Enumerating {$p$}-groups},
        date={1965},
     journal={Proceedings of the London Mathematical Society},
      volume={3},
      number={1},
       pages={151\ndash 166},
}

\bib{Vershik}{article}{
      author={Vershik, A.M.},
       title={Statistical mechanics of combinatorial partitions, an their limit
  shapes},
     journal={V.A. Steklov Institute of Mathematics, Russian Academy of
  Sciences},
}

\bib{WL68}{article}{
      author={Weisfeiler, Boris},
      author={Lehman, Andrei~A.},
       title={A reduction of a graph to a canonical form and an algebra arising
  during this reduction},
        date={1968},
     journal={Nauchno-Technicheskaya Informatsia},
      volume={2},
      number={9},
       pages={12\ndash 16},
        note={English translation available at
  \href{https://www.iti.zcu.cz/wl2018/wlpaper.html}{https://www.iti.zcu.cz/wl2018/wlpaper.html}},
}

\bib{Weisfeiler}{book}{
      editor={Weisfeiler, B.},
       title={On construction and identification of graphs},
      series={Lecture Notes in Mathematics},
   publisher={Springer-Verlag},
        date={1976},
      volume={558},
        note={With contributions by A. Lehman, G. M. Adelson-Velsky, V.
  Arlazaraov, I. Faragev, A. Uskov, I. Zuev, M. Rosenfeld, and B. Weisfeiler},
}

\bib{Wilson:profile}{article}{
      author={Wilson, James~B.},
       title={The threshold for subgroup profiles to agree is {$\Omega(\log
  n)$}},
        note={arXiv:1612.01444},
}

\bib{Wilson:alpha}{article}{
      author={Wilson, James~B.},
       title={More characteristic subgroups, {Lie} rings, and isomorphism tests
  for $p$-groups},
        date={2013},
        ISSN={1433-5883},
     journal={J. Group Theory},
      volume={16},
      number={6},
       pages={875\ndash 897},
      review={\MR{3198722}},
}

\bib{Wilson:Skolem-Noether}{article}{
      author={Wilson, James~B.},
       title={Skolem--{Noether} for nilpotent products},
        date={2015},
     journal={arXiv preprint arXiv:1507.04406},
}

\end{biblist}
\end{bibdiv}

\end{document}